\documentclass[sn-vancouver]{sn-jnl}

\normalbaroutside 

\usepackage{anyfontsize}
\usepackage[english]{babel}
\usepackage{amssymb}

\usepackage{tikz}
\usetikzlibrary{positioning,arrows,petri,calc,decorations.markings,arrows.meta}
\tikzset{
place/.style={
circle,
thick,
minimum size=4mm,
draw
},
transitionV/.style={
rectangle,
thick,
fill=black,
minimum height=6mm,
inner xsep=1pt
}
}
\usepackage{pgfplots}
\pgfplotsset{compat=newest}
\pgfplotsset{plot coordinates/math parser=false}
\newlength\figureheight
\newlength\figurewidth 
\usepackage{cancel}

\usepackage{mathtools}
\usepackage{etoolbox}
\usepackage[ruled,vlined,linesnumbered,algo2e]{algorithm2e}
\makeatletter
\patchcmd{\algocf@Vline}{\vrule}{\vrule\vspace{-.32em}}{}{}
\makeatother
\SetStartEndCondition{ }{ }{}%
\SetKw{KwTo}{to}\SetKwFor{For}{for}{\string do}{}%
\SetKwIF{If}{ElseIf}{Else}{if}{then}{elif}{else}{}%
\SetKwFor{While}{while}{\string do}{}%

\usepackage{caption}
\DeclareMathAlphabet{\pazocal}{OMS}{zplm}{m}{n}
\usepackage[autostyle, english = american]{csquotes} 
\MakeOuterQuote{"}
\usepackage{xargs}
\usepackage{subfigure}
\usepackage{stmaryrd} 

\usepackage{soul, color, xcolor}

\usepackage[numbers]{natbib}
\usepackage{enumitem}
\usepackage{nicematrix}
\usepackage{cleveref}
\usepackage{empheq} 

\definecolor{myblue}{RGB}{0, 101, 202} 
\definecolor{mygreen}{RGB}{34, 139, 34} 
\definecolor{myred}{RGB}{197, 14, 31}
\definecolor{mypurple}{RGB}{128, 0, 128}
\definecolor{myyellow}{RGB}{204, 204, 0}
\definecolor{mygrey}{RGB}{105, 105, 105}
\usepackage{fontawesome}

\usetikzlibrary{hobby}

\newcommand{\zor}[1]{#1} 

\newcommand{\graph}{\pazocal{G}}
\newcommand{\dint}[1]{\left\llbracket#1\right\rrbracket} 
\newcommand{\eqtop}[1]{\stackrel{\text{#1}}{=}} 

\renewcommand{\R}{\mathbb{R}}
\newcommand{\nat}{\mathbb{N}}
\newcommand{\nato}{\mathbb{N}_0}
\newcommand{\Rmax}{{\R}_{\normalfont\fontsize{7pt}{11pt}\selectfont\mbox{max}}}
\newcommand{\Rmin}{{\R}_{\normalfont\fontsize{7pt}{11pt}\selectfont\mbox{min}}}

\newcommand{\Rbar}{\overline{\R}}

\newcommand{\arcs}{E}
\newcommand{\nodes}{N}
\newcommand{\nonegset}{\Gamma}

\newcommand{\MP}{P}
\newcommand{\MI}{I}
\newcommand{\MC}{C}

\newcommand{\MA}{A}
\newcommand{\MB}{B}

\newcommandx{\solNCP}[1]{\Lambda_{\mbox{\normalfont\tiny NCP}}(#1)}
\newcommandx{\solPIC}[3][1=\MP,2=\MI,3=\MC]{\Lambda_{\mbox{\normalfont\tiny NCP}}(\lambda#1\oplus\lambda^{-1}#2\oplus#3)}
\newcommandx{\solPTEG}[4][1=A^0,2=A^1,3=B^0,4=B^1]{\Lambda_{\mbox{\normalfont\tiny P-TEG}}(#1,#2,#3,#4)}
\newcommandx{\solSLDI}[2][1=v,2=\pazocal{S}]{\Lambda^{#1}_{\mbox{\normalfont\tiny SLDI}}(#2)}
\newcommand{\wA}{\mathsf{\MakeLowercase{\MA}}}
\newcommand{\wB}{\mathsf{\MakeLowercase{\MB}}}
\newcommand{\wP}{\mathsf{\MakeLowercase{\MP}}}

\newcommand{\wC}{\mathsf{\MakeLowercase{\MC}}}
\newcommand{\winit}{\mathsf{i}}
\newcommand{\wfin}{\mathsf{f}}
\newcommand{\wZ}{\mathsf{z}}

\newcommandx{\TP}[1][1=d]{T_{#1}^{\MP}}
\newcommandx{\TI}[1][1=d]{T_{#1}^{\MI}}
\newcommandx{\TC}[1][1=d]{T_{#1}^{\MC}}

\newcommand\myOverwrite[2]{\!\makebox[0cm][l]{#1}#2\ \!} 
\newcommand{\maxldiv}{\textup{\hspace{0.85mm}\myOverwrite{$\circ$}{$\setminus$}\hspace{-0.5mm}}} 
\newcommand{\minldiv}{\textup{\hspace{0.85mm}\myOverwrite{$\bullet$}{$\setminus$}\hspace{-0.5mm}}} 

\newcommand{\places}{\pazocal{P}}
\newcommand{\transitions}{\pazocal{T}}

\renewcommand{\qed}{\hfill\ensuremath{\blacksquare}}

\DeclareMathOperator{\tr}{tr}

\makeatletter
\newcommand{\splus}{%
  \DOTSB\mathop{\mathpalette\mattos@splus\relax}\slimits@
}
\newcommand\mattos@splus[2]{%
  \vcenter{\hbox{%
    \sbox\z@{$#1\oplus$}%
    \resizebox{!}{0.9\dimexpr\ht\z@+\dp\z@}{\raisebox{\depth}{$\m@th#1\boxplus$}}%
  }}%
  \vphantom{\oplus}%
}
\makeatother

\makeatletter
\newcommand{\stimes}{%
  \DOTSB\mathop{\mathpalette\mattos@stimes\relax}\slimits@
}
\newcommand\mattos@stimes[2]{%
  \vcenter{\hbox{%
    \sbox\z@{$#1\otimes$}%
    \resizebox{!}{0.9\dimexpr\ht\z@+\dp\z@}{\raisebox{\depth}{$\m@th#1\boxtimes$}}%
  }}%
  \vphantom{\otimes}%
}
\makeatother

\makeatletter
\newcommand{\bigsplus}{%
  \DOTSB\mathop{\mathpalette\mattos@bigsplus\relax}\slimits@
}
\newcommand\mattos@bigsplus[2]{%
  \vcenter{\hbox{%
    \sbox\z@{$#1\sum$}%
    \resizebox{!}{0.9\dimexpr\ht\z@+\dp\z@}{\raisebox{\depth}{$\m@th#1\boxplus$}}%
  }}%
  \vphantom{\sum}%
}
\makeatother

\newcommand{\svdots}{\raisebox{3pt}{$\scalebox{.75}{\vdots}$}} 
\newcommand{\sddots}{\raisebox{3pt}{$\scalebox{.75}{$\ddots$}$}} 

\renewcommand{\preceq}{\leq}
\renewcommand{\succeq}{\geq}

\jyear{2023}%

\theoremstyle{thmstyleone}%
\newtheorem{theorem}{Theorem}
\newtheorem{proposition}[theorem]{Proposition}%
\newtheorem{lemma}[theorem]{Lemma}%

\theoremstyle{thmstyletwo}%
\newtheorem{example}{Example}%
\newtheorem{remark}{Remark}%

\theoremstyle{thmstylethree}%
\newtheorem{definition}{Definition}%

\begin{document}

\title[SLDIs: cycle time analysis and applications]{Switched max-plus linear-dual inequalities: cycle time analysis and applications}

\author*[1]{\fnm{Davide} \sur{Zorzenon}}\email{zorzenon@control.tu-berlin.de}

\author[2]{\fnm{Jan} \sur{Komenda}}\email{komenda@ipm.cz}

\author[1,3]{\fnm{J\"{o}rg} \sur{Raisch}}\email{raisch@control.tu-berlin.de}

\affil[1]{\orgdiv{Control Systems Group}, \orgname{Technische Universit\"{a}t Berlin}, \orgaddress{\country{Germany}}}

\affil[2]{\orgdiv{Institute of Mathematics}, \orgname{Czech Academy of Sciences}, \orgaddress{\city{Prague}, \country{Czech Republic}}}

\affil[3]{\orgdiv{Science of Intelligence}, \orgname{Research Cluster of Excellence}, \orgaddress{\city{Berlin}, \country{Germany}}}

\abstract{
P-time event graphs are discrete event systems suitable for modeling processes in which tasks must be executed in predefined time windows.
Their dynamics can be represented by max-plus linear-dual inequalities (LDIs), i.e., systems of linear dynamical inequalities in the primal and dual operations of the max-plus algebra.
We define a new class of models called switched LDIs (SLDIs), which allow to switch between different modes of operation, each corresponding to a set of LDIs, according to a sequence of modes called schedule.
In this paper, we focus on the analysis of SLDIs when the considered schedule is fixed and either periodic or intermittently periodic.
We show that SLDIs can model a wide range of applications including single-robot multi-product processing networks, in which every product has different processing requirements and corresponds to a specific mode of operation.
Based on the analysis of SLDIs, we propose algorithms to compute: i. minimum and maximum cycle times for these processes, improving the time complexity of other existing approaches; ii. a complete trajectory of the robot including start-up and shut-down transients.
}


\maketitle

\section{Introduction}

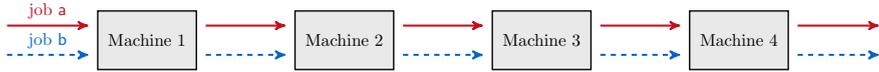
\begin{figure}[h]
    \centering
    \resizebox{1\textwidth}{!}{
%
%
%

\begin{tikzpicture}[node distance=1.2cm and 4cm,on grid,thick]


\newcommand{\mwidth}{2}
\newcommand{\mheight}{.6*\mwidth}
\newcommand{\offset}{.2}


\node [draw,thick,fill=mygrey!15!white,minimum height=\mheight cm,minimum width=\mwidth cm] at (0,0) (M1) {Machine 1};

\node [draw,thick,fill=mygrey!15!white,minimum height=\mheight cm,minimum width=\mwidth cm,right=of M1] (M2) {Machine 2};

\node [right=of M2] {\huge$\cdots$};

\node [draw,thick,fill=mygrey!15!white,minimum height=\mheight cm,minimum width=\mwidth cm,right=of M2] (M3) {Machine 3};

\node [draw,thick,fill=mygrey!15!white,minimum height=\mheight cm,minimum width=\mwidth cm,right=of M3] (M4) {Machine 4};

%


\draw [myred,very thick,-stealth',shorten >=5pt,shorten <=5pt] ($(M1.west)+(-2,\mheight/4)$) to node [auto,pos=.5] {job $\wA$} ($(M1.west)+(0,\mheight/4)$);
\draw [myred,very thick,-stealth',shorten >=5pt,shorten <=5pt] ($(M1.east)+(0,\mheight/4)$) to ($(M2.west)+(0,\mheight/4)$);
\draw [myred,very thick,-stealth',shorten >=5pt,shorten <=5pt] ($(M2.east)+(0,\mheight/4)$) to ($(M3.west)+(0,\mheight/4)$);
\draw [myred,very thick,-stealth',shorten >=5pt,shorten <=5pt] ($(M3.east)+(0,\mheight/4)$) to ($(M4.west)+(0,\mheight/4)$);
\draw [myred,very thick,-stealth',shorten >=5pt,shorten <=5pt] ($(M4.east)+(0,\mheight/4)$) to ($(M4.east)+(2,\mheight/4)$);


\draw [myblue,very thick,dashed,-stealth',shorten >=5pt,shorten <=5pt] ($(M1.west)+(-2,-\mheight/4)$) to node [auto,pos=.5] {job $\wB$} ($(M1.west)+(0,-\mheight/4)$);
\draw [myblue,very thick,dashed,-stealth',shorten >=5pt,shorten <=5pt] ($(M1.east)+(0,-\mheight/4)$) to ($(M2.west)+(0,-\mheight/4)$);
\draw [myblue,very thick,dashed,-stealth',shorten >=5pt,shorten <=5pt] ($(M2.east)+(0,-\mheight/4)$) to ($(M3.west)+(0,-\mheight/4)$);
\draw [myblue,very thick,dashed,-stealth',shorten >=5pt,shorten <=5pt] ($(M3.east)+(0,-\mheight/4)$) to ($(M4.west)+(0,-\mheight/4)$);
\draw [myblue,very thick,dashed,-stealth',shorten >=5pt,shorten <=5pt] ($(M4.east)+(0,-\mheight/4)$) to ($(M4.east)+(2,-\mheight/4)$);

\end{tikzpicture}

}
    \caption{Illustration of a flow shop with 4 machines and jobs of two different types, labeled $\wA$ and $\wB$.}\label{fi:flow_shop}
\end{figure}

Time-window constraints arise in different settings.
Typical examples from manufacturing are the bakery and the semi-conductor industries, where some processes (e.g., yeast fermentation in the former case, metal deposition in the latter) must be executed under strict temporal constraints in order to obtain the desired quality of the final product~\cite{hecker2014application,manier2003classification}.
In transportation, time windows can be used to specify admissible pickup and delivery times for customers~\cite{solomon1988survey}.
From the pharmaceutical domain, high-throughput screening systems are another example where, due to chemical requirements, the time allowed to elapse between two operations is restricted by lower and/or upper bounds~\cite{mayer2004time}.

P-time event graphs (P-TEGs) are event graphs\footnote{We recall that event graphs are Petri nets in which all arcs have unitary weight and \zor{all} places have exactly one upstream and one downstream transition.} in which tokens are forced to sojourn in places in predefined time windows.
Given their ability to model time-window constraints, they have been applied to solve scheduling and control problems in various types of systems including bakeries, electroplating lines, and cluster tools~\cite{declerck2021critical,spacek1999control,becha2017model,kim2003scheduling}.
The signal describing firing times of transitions in P-TEGs evolves -- non-deterministically -- according to max-plus linear-dual inequalities (LDIs), i.e., dynamical inequalities that are linear in the primal operations and in the dual operations of the max-plus algebra (see~\eqref{eq:dynamics_LDIs} in Section~\ref{su:LDIs}).

\zor{
Since temporal upper bound constraints can be considered as specifications that a system driven exclusively by lower time bounds (i.e., a max-plus linear system) needs to satisfy, several authors have addressed the problem of limiting the sojourn time of tokens in places from a control point of view.
This resulted in a variety of techniques; we mention for instance~\cite{katz2007maxplus,maia2011super,maia2011control,declerck2016compromise}, in which the authors find sufficient conditions for the control of max-plus linear systems subject to constraints using, respectively, geometric control theory, max-plus spectral theory, residuation theory, and model predictive control.
Other sufficient conditions were discovered in~\cite{amari2012maxplus}.
Note that, in its most general formulation, this control problem is still open, as no necessary and sufficient condition for its solution has been found.
The reason for this can be traced to the fact that a more fundamental problem, namely, checking the existence of solutions of LDIs, has not yet been fully understood.
On the other hand, the steady-state version of the problem was solved for a large class of instances in~\cite{goncalves2017max}, exploiting the results of~\cite{katz2007maxplus}.
}
Moreover, the simple nature of the dynamics of P-TEGs has led to a number of theoretical results regarding the analysis of their cycle time \cite{declerck2007cycle,lee2014steady,vspavcek2017analysis,zorzenon2021periodic}, which is defined as the temporal difference between the occurrence of two repetitions of the same event in the system, assuming events occur in a periodic manner (for a formal definition, see Section~\ref{su:structural_PTEGs}).

Such a simple characterization comes\zor{, however,} at the cost of limited modeling power.
To illustrate this, consider the example of flow shops, which are manufacturing systems consisting of a sequence of machines of unitary capacity where all the jobs (i.e., parts to be processed) must visit all the machines in the same order (see Figure~\ref{fi:flow_shop}).
P-TEGs are the ideal tool for representing flow shops where all the jobs are identical.
Suppose instead that jobs of different \zor{types} are present, each of which requires different processing times in machines, and that the entrance order of jobs is periodic and repeats every $V\in\nat$ jobs.
An example of periodic entrance order of period $V = 2$ for the flow shop of Figure~\ref{fi:flow_shop} is $\wA\wB\wA\wB\wA\wB\ldots$
In this case, P-TEGs can still be adopted to model the system (this is shown formally in Proposition~\ref{pr:SLDI-P-TEGs}), but:
\begin{itemize}
    \item the number of transitions in the resulting P-TEGs increases (linearly) with $V$, which leads to a rather high computational complexity for the cycle time analysis where large values of $V$ are considered;
    \item for a different entrance order of jobs, a new P-TEG must be built.
\end{itemize}
Moreover, when the number of jobs to be processed is infinite\footnote{Considering infinitely many jobs to be processed is a typical assumption in the context of continuous production systems.} and the entrance order is not periodic, no P-TEG can model the manufacturing system.

With the aim of overcoming these limitations, in this paper we introduce a new class of dynamical systems called \textit{switched max-plus linear-dual inequalities} (SLDIs).
SLDIs extend the modeling power of P-TEGs by allowing to switch among different modes of operation, each corresponding to a system of LDIs.
Let us take the example of the flow shop again.
By assigning a mode of operation to each job type, we can model the manufacturing system by SLDIs in which different job entrance orders simply correspond to different schedules, i.e., sequences of modes.
A first advantage of SLDIs is thus that, with a single dynamical system, one can represent all possible entrance orders of jobs in the flow shop, including non-periodic ones.

It is worth noting that this property is not exclusive to SLDIs, as also P-time Petri nets can be utilized to model flow shops under any job order~\cite{khansa1996p,bonhomme2013towards}.
In fact, there exists a strong relation between the two models (which is briefly discussed in Section~\ref{se:conclusions}).
However, \zor{in contrast to} P-time Petri nets, the dynamics of SLDIs possesses another appealing feature: a switched-linear formulation in the max-plus algebra.
By exploiting the abundance of existing results in this algebraic framework (see, e.g., \cite{cuninghame2012minimax,baccelli1992synchronization,butkovivc2010max,hardouin2018control}), in this paper we will derive low-complexity algorithms for the cycle time computation, considering two types of schedules: periodic and intermittently periodic schedules.
In the example of the flow shop, the first type of schedules corresponds to periodic arrivals of jobs of different types in the manufacturing system (as in schedule $\wA\wB\wA\wB\wA\wB\ldots$), whereas the second one is a generalization in which the entrance order of jobs alternates among periodic and non-periodic regimes. 
An example of a schedule of the second type for the flow shop of Figure~\ref{fi:flow_shop} is $\wA\wA\wB\wA\wB\wA\wB\ldots\wA\wB\wB\wB\wA\wB\wA\wB\wA\ldots$, which consists of an initial transient $\wA$ followed by periodic regime $\wA\wB\wA\wB\ldots$, an intermediate transient $\wB$, and the final periodic subschedule $\wB\wA\wB\wA\ldots$\footnote{\zor{In this schedule example, different decompositions in periodic and transient regimes are possible. For more details on this, we refer to Section~\ref{su:piecewise_schedules}.}}
In the cycle time analysis for this kind of schedules, one is interested in finding the cycle times that can be achieved in each periodic regime.
The motivation for studying intermittently periodic schedules is two-fold:
\begin{enumerate}
    \item they are ubiquitous in applications (as discussed in Section~\ref{su:piecewise_schedules}), and
    \item differently from systems described by only lower time-bounds, the cycle time analysis in this class of schedules can produce non-trivial results in the presence of time-window constraints.
\end{enumerate}
Here, by "non-trivial" we mean that the cycle times of the periodic subschedules in an intermittently periodic schedule may be different from those obtained by studying the periodic subschedules independently.

The present article enhances and extends the recent conference paper~\cite{ZORZENON2022196} in \zor{several ways}.
Besides simplifying the preliminaries in Section~\ref{se:preliminaries} and adding a number of simple and illustrative examples, an important contribution of this extended paper is to systematically investigate the initial conditions of P-TEGs (in Section~\ref{se:P-time_event_graphs}) and their relation to SLDIs.
Two types of initial conditions are presented, loose and strict, and it is proven that P-TEGs with strict initial conditions can be represented by SLDIs but not by pure LDIs (Section~\ref{su:SLDIS_PTEGs}).
In Section~\ref{se:switched_max-plus_linear-dual_systems}, after formally presenting SLDIs, the cycle time analysis in periodic and intermittently periodic trajectories is discussed; intermittently periodic trajectories are introduced for the first time in this extended version.
In~\cite{ZORZENON2022196}, the correctness of the low-complexity method for computing the cycle time in SLDIs under periodic schedules was proven using tools from automata and regular languages theory; in this paper, the proof is entirely based on algebraic arguments.
Although the initial inspiration for the algorithm was drawn from analogies between multi-precedence graphs and automata theory, the new proof relies on more established results, which hopefully makes it less arduous \zor{to read}.

In Section~\ref{se:example}, the analysis of the case study considered in~\cite{ZORZENON2022196} has been deepened further.
The case study consists of \zor{a robotic job shop example} derived from~\cite{KATS20081196}, where parts of different type \zor{are required} to visit a sequence of processing stations in different order, and are transported by a single robot.
The authors of~\cite{KATS20081196} proved that the cycle time analysis in this class of systems can be performed in strongly polynomial time complexity $\pazocal{O}(V^4n^4)$, where $V$ is the period of the entrance order of different types of parts in the system, and $n$ is the number of processing stations.
In this paper, we show that the complexity can be reduced to $\pazocal{O}(Vn^3 + n^4)$ using SLDIs.
Computational tests show that the advantage is not only theoretical, but translates into a tangibly faster cycle time analysis.
This makes our algorithm an appealing optimization subroutine for the solution of (NP-hard) cyclic scheduling problems in which the goal is to find the optimal path of the robot.
Additionally, in this paper we show how to derive a complete (intermittently periodic) trajectory of the system, consisting of a start-up transient (where parts are initially introduced into the system), a periodic regime, and a shut-down transient (in which all parts are removed from the stations).

Finally, Section~\ref{se:conclusions} provides concluding remarks, comparisons with related classes of dynamical systems, and suggestions for future work.

\subsection*{Notation}
The set of positive, respectively non-negative, integers is denoted by $\nat$, respectively $\nato$.
The set of non-negative real numbers is denoted by $\R_{\geq 0}$.
Moreover, $\Rmax \coloneqq \R \cup \{-\infty\}$, $\Rmin \coloneqq \R\cup\{\infty\}$, and $\Rbar \coloneqq \Rmax \cup \{\infty\}=\Rmin\cup\{-\infty\}$.
If $A\in\Rbar^{n\times n}$, we will use notation $A^\sharp$ to indicate $-A^\top$.
Given $a, b \in \mathbb{Z}$ with $b \geq a$, $\dint{a, b}$ denotes the discrete interval $\{a, a + 1, a + 2, \ldots, b\}$.

\section{Preliminaries}\label{se:preliminaries}

In this section, some basic concepts of max-plus algebra (Section~\ref{su:idempotent_semirings}) and precedence graphs (Section~\ref{su:precedence_graphs}) are recalled; for a more detailed discussion of those topics, we refer to~\cite{baccelli1992synchronization,butkovivc2010max,hardouin2018control}.
Sections~\ref{su:NCP} and~\ref{su:LDIs} present the non-positive circuit weight problem (introduced in~\cite{zorzenon2021nonpositive}) and max-plus linear-dual inequalities.

\subsection{Max-plus algebra}\label{su:idempotent_semirings}

The max-plus algebra operates on the set of real numbers extended with $-\infty$ and $+\infty$, and is endowed with operations $\oplus$ (addition), $\otimes$ (multiplication), $\splus$ (dual addition), and $\stimes$ (dual multiplication), defined as follows: for all $a,b\in\Rbar$,
\[
\begin{array}{rclcrcl}
    a \oplus b &=& \max(a,b), &\quad& a\otimes b &=& 
    \begin{dcases}
        a+b & \mbox{if } a,b\neq -\infty,\\
        -\infty & \mbox{otherwise,}
    \end{dcases}\\
    a \splus b &=& \min(a,b), &\quad& a\stimes b &=& 
    \begin{dcases}
        a+b & \mbox{if } a,b\neq +\infty,\\
        +\infty & \mbox{otherwise.}
    \end{dcases}
\end{array}
\]
These operations can be extended to matrices; given $A,B\in\Rbar^{m\times n}$, $C\in\Rbar^{n\times p}$, for all $i\in\dint{1,m}$, $j\in\dint{1,n}$, $h\in\dint{1,p}$,
\[
    \begin{array}{rclcrcl}
        (A\oplus B)_{ij} &=& A_{ij}\oplus B_{ij}, &\quad& (A\otimes C)_{ih} &=& \displaystyle\bigoplus_{k=1}^n A_{ik}\otimes C_{kh},\\
        (A\splus B)_{ij} &=& A_{ij}\splus B_{ij}, &\quad& (A\stimes C)_{ih} &=& \displaystyle\bigsplus_{k=1}^n A_{ik}\stimes C_{kh}.
    \end{array}
\]
\zor{When the meaning is clear from the context, we will denote the max-plus multiplication between matrices $A$ and $C$, $A\otimes C$, simply by $A C$.}
The symbols $\pazocal{E}$, $\pazocal{T}$, and $E_\otimes$ denote, respectively, the neutral element for $\oplus$, $\splus$, and $\otimes$, i.e., $\pazocal{E}_{ij} = -\infty$ and $\pazocal{T}_{ij} = +\infty$ for all $i,j$, and $E_\otimes$ is a square matrix with $(E_\otimes)_{ij} = 0$ if $i=j$ and $(E_\otimes)_{ij}=-\infty$ if $i\neq j$.
Given a square matrix $A$, its $r$\textsuperscript{th} power is defined recursively by $A^0 = E_\otimes$ and, for all $r\geq 1$, $A^r = A^{r-1}\otimes A$.
Moreover, the Kleene star of a matrix $A\in\Rbar^{n\times n}$ is
\[
	A^* = \bigoplus_{i = 0}^{+\infty} A^i.
\]
The partial order relation $\preceq$ between two matrices $A$ and $B$ of the same size is defined elementwise: $A\preceq B$ if and only if $A_{ij} \leq B_{ij}$ for all $i,j$.
Analogously to \zor{the} standard algebra, we define the max-plus trace of matrix $A\in\Rbar^{n\times n}$ by $\tr(A) = \bigoplus_{i = 1}^n A_{ii}$.
The product and dual product between scalar $\lambda\in\Rbar$ and matrix $A\in\Rbar^{m\times n}$ are given by
\[
	(\lambda \otimes A)_{ij} = \lambda \otimes A_{ij},\quad (\lambda \stimes A)_{ij} = \lambda \stimes A_{ij}.
\]
If $\lambda\notin\{-\infty,+\infty\}$, the two expressions coincide.
We will therefore simply write $\lambda A$ in place of $\lambda \otimes A$ or $\lambda \stimes A$ when $\lambda\in\R$.
When $\lambda\in\R$, we indicate by $\lambda^{-1}$ the element such that $\lambda^{-1}\otimes \lambda = \lambda \otimes \lambda^{-1} = 0$; thus, in \zor{the} standard algebra, $\lambda^{-1}$ coincides with $-\lambda$.
\zor{When $\lambda\in\{-\infty,+\infty\}$, $\lambda$ does not have a multiplicative inverse; nevertheless, we will use symbol $\lambda^{-1}$ to denote $-\lambda$.}

\subsection{Precedence graphs}\label{su:precedence_graphs}

A \textit{directed graph} is a pair $(\nodes,\arcs)$ where $\nodes$ is a finite set of nodes and $\arcs\subseteq\nodes\times\nodes$ is the set of arcs.
A \textit{weighted directed graph} is a triplet $(\nodes,\arcs,w)$, where $(\nodes,\arcs)$ is a directed graph, and $w:\arcs\rightarrow\R$ is a function that associates a weight $w((i,j))$ to each arc $(i,j)\in\arcs$ of graph $(\nodes,\arcs)$.
A sequence of $r+1$ nodes $\rho=(i_1,i_2,\ldots,i_{r+1})$, $r\geq 0$, such that $(i_j,i_{j+1})\in\arcs$ for all $j\in\dint{1,r}$ is a path of length $r$; a path $\rho$ such that $i_1 = i_{r+1}$ is called a circuit.
The weight of a path is the sum (in conventional algebra) of the weights of the arcs composing it; conventionally, the weight of a path of length $r=0$, i.e., $\rho=(i_1)$, is equal to $0$.

The \textit{precedence graph} associated with matrix $\MA\in\Rmax^{n\times n}$ is the weighted directed graph $\graph(A)=(\nodes,\arcs,w)$, where $\nodes=\dint{1,n}$, and there is an arc $(j,i)\in\arcs$ of weight $w((j,i))=A_{ij}$ if and only if $A_{ij}\neq -\infty$. 
We say that $\graph(A)$ is a \textit{parametric precedence graph} when elements of $A$ are functions of some real parameters $\lambda_1,\ldots,\lambda_p$, i.e., $A = A(\lambda_1,\ldots,\lambda_p)$.

There are important connections between the max-plus algebra and precedence graphs.
For instance, element $(i,j)$ of the $r$\textsuperscript{th} max-plus power of a matrix $A\in\Rmax^{n\times n}$, $(A^r)_{ij}$, corresponds to the maximum weight of all paths in $\graph(A)$ of length $r$ from node $j$ to node $i$.
A direct consequence is that $(A^*)_{ij}$ is equal to the largest weight of all paths (of any length) from node $j$ to $i$.
Observe that a precedence graph $\graph(A)$ does not contain circuits with positive weight if and only if $\tr(A^*) = 0$; in presence of at least one circuit with positive weight, $\tr(A^*) = \infty$.
In the following, we indicate by $\nonegset$ the set of all precedence graphs that do not contain circuits with positive weight, i.e.,
\[
	\nonegset = \{ \graph(A)\ | \ \tr(A^*) = 0 \}.
\]
The following proposition will be used later to verify the existence of, and compute, trajectories \zor{satisfying time-window constraints} in (switched) max-plus linear-dual inequalities. 
\begin{proposition}{\cite{butkovivc2010max,baccelli1992synchronization}}\label{pr:nonegset_inequality}
    Let $A$ be an $n\times n$ matrix with elements in $\Rmax$.
    Inequality $A\otimes x \preceq x$ admits a solution $x\in \R^n$ if and only if $\zor{\graph(A)}\in\nonegset$.
    In this case, any column of $A^*$ solves the inequality, i.e., $A\otimes (A^*)_{\cdot,i} \preceq (A^*)_{\cdot,i}$ for all $i\in\dint{1,n}$.
\end{proposition}

The \textit{maximum circuit mean} of precedence graph $\graph(A)$, denoted by $\mbox{mcm}(A)$, is the maximum weight-over-length ratio of all circuits \zor{of positive length} in the graph; this value coincides with the largest max-plus eigenvalue (the \textit{max-plus spectral radius}) of $A$\zor{, i.e., the largest $\lambda\in\Rmax$ such that $A\otimes x = \lambda x$ for some vector $x$ with elements from $\Rmax$.}
For a matrix of dimension $n\times n$, the maximum circuit mean can be computed through the following formula (\cite{baccelli1992synchronization}):
\[
	\mbox{mcm}(A) = \bigoplus_{k=1}^n \tr(A^k)^{\frac{1}{k}},
\]
where $a^{\frac{1}{k}}$ (corresponding to $\frac{a}{k}$ in \zor{the} standard algebra) is the $k$\textsuperscript{th} max-plus root of $a\in\Rmax$\zor{; a more efficient algorithm that returns the same value in time complexity $\mathcal{O}(n\times m)$ in the worst case, where $m$ is the number of edges in $\graph(A)$, is due to Karp~\cite{KARP1978309}}.

\zor{We recall that it is possible to check whether $\graph(A)\in\nonegset$ and, when $\graph(A)\in\nonegset$, to compute $A^*$ in time $\pazocal{O}(n^3)$ in the worst case, using the Floyd-Warshall algorithm~\cite{HOUGARDY2010279}.
In the case $\graph(A)\notin\nonegset$, computing $A^*$ is an NP-hard problem; fortunately, the algorithms presented in the next sections will never face this issue in practice.}

\subsection{The non-positive circuit weight problem}\label{su:NCP}

Given a parametric precedence graph $\graph(\MA)$, where $\MA = \MA(\lambda_1,\ldots,\lambda_p)$, the \textit{non-positive circuit weight problem} (NCP) consists in characterizing the set $\solNCP{\MA} = \{(\lambda_1,\dots,\lambda_p)\in\R^p\ |\ \graph(\MA)\in\nonegset\}$ of all values of parameters $\lambda_1, \dots,\lambda_p$ for which $\graph(A)$ does not contain circuits with positive weight.
Specific classes of the NCP find applications in the analysis of periodic trajectories in max-plus dynamical systems.
An application example is shown in Section~\ref{su:LDIs}.

\subsubsection{The PIC-NCP}

When matrix $\MA$ has the form 
\[
	\MA(\lambda) = \lambda\MP \oplus \lambda^{-1}\MI \oplus \MC
\]
for arbitrary matrices $\MP,\MI,\MC\in\Rmax^{n\times n}$ (called proportional, inverse\footnote{\zor{In contrast with other notations, symbol $I$ in this paper does \textit{not} indicate the identity matrix in the max-plus algebra, but an arbitrary matrix. The identity matrix is instead denoted by $E_{\otimes}$.}}, and constant matrix, respectively), then the problem is referred to as the \textit{proportional-inverse-constant-NCP}. 
In this case, $\solNCP{\lambda\MP\oplus\lambda^{-1}\MI\oplus\MC}=[\lambda_{\text{min}},\lambda_{\text{max}}]\cap\R$ is an interval, and its extreme points can be found either in weakly polynomial time using linear programming solvers such as the interior-point method, or in strongly polynomial time $\pazocal{O}(n^4)$ using Algorithm~\ref{al:PIC-NCP} (\cite{zorzenon2021nonpositive}).
\zor{The functioning of the latter algorithm is briefly described in the following}.

\zor{We remark that the PIC-NCP represents a generalization of the max-plus subeigenproblem (see~\cite{gaubert1995resource}), i.e., the problem of finding a real $\lambda$ such that the inequality}
\[
 I \otimes x\leq \lambda x
\] 
\zor{admits a solution $x\in\R^{n}$ for a given matrix $I\in\Rmax^{n\times n}$.
Indeed, the above inequality can be rewritten (by multiplying both sides by $\lambda^{-1}$) as}
\[
\lambda^{-1}I \otimes x \leq x,
\] 
\zor{and, from Proposition~\ref{pr:nonegset_inequality}, admits a solution $x\in\R^n$ if and only if the precedence graph $\graph(\lambda^{-1}I)$ does not have circuits with positive weight; the PIC-NCP thus simplifies into the max-plus subeigenproblem when matrices $P$ and $C$ are $\mathcal{E}$.
We recall from~\cite[Lemma 1]{gaubert1995resource} that the least solution of the subeigenproblem is the max-plus spectral radius of matrix $I$, i.e.,}
\[
\solNCP{\lambda^{-1}I} = [\mbox{mcm}(I),+\infty) \cap \R.
\] 
\zor{When matrices $P$ and $C$ are not $\mathcal{E}$, a more sophisticated approach is necessary to solve the PIC-NCP. 
In particular, in Algorithm~\ref{al:PIC-NCP} some pre-computations are first performed (lines 3-4) to simplify the problem into an equivalent PI-NCP (a PIC-NCP where matrix $C$ is $\mathcal{E}$) by redefining $P$ as $C^* \otimes P \otimes C^*$ and $I$ as $C^* \otimes I \otimes C^*$. 
Then, through the for-loop of lines 5-7, a matrix $S$ is constructed such that, in every new iteration of the loop, the spectral radius of matrix $I\otimes S^*$ and the inverse of the spectral radius of matrix $P\otimes S^*$ approximate better and better $\lambda_{\textup{min}}$ and $\lambda_{\textup{max}}$, respectively.
It can be shown (see~\cite{zorzenon2021nonpositive}) that after at most $\left \lfloor{\frac{n}{2}} \right \rfloor$ iterations, the two quantities converge to the desired values (line 10), i.e.,}
\[
    \solNCP{\lambda P \oplus \lambda^{-1}I \oplus C} = [\mbox{mcm}(I\otimes S^*),(\mbox{mcm}(P\otimes S^*))^{-1}] \cap \R.
\]
\zor{If some conditions on matrices $C$ and $S$ do not hold, the algorithm can be terminated prematurely as no $\lambda$ solving the PIC-NCP exists (lines 1-2 and 8-9).}

\begin{algorithm2e}[t]   \DontPrintSemicolon \small
\KwIn{$\MP,\MI,\MC\in \Rmax^{n\times n}$}
\KwOut{$\Lambda_{\text{NCP}}(\lambda\MP\oplus \lambda^{-1}\MI\oplus \MC)$}
 \If{$\graph(\MC)\notin \nonegset$}{%
     \Return $\emptyset$%
 }%
$\MP \leftarrow \MC^* \MP \MC^*$\;
$\MI \leftarrow \MC^* \MI \MC^*$\;
$S \leftarrow E_{\otimes}$\;
\For{$k = 1$ \KwTo $\left \lfloor{\frac{n}{2}} \right \rfloor$}{\vspace*{2pt}
	$S\leftarrow \MP S^2 \MI \oplus \MI S^2 \MP \oplus E_{\otimes}$
}
\If{$\graph(S)\notin \nonegset$}{%
	\Return $\emptyset$%
}%
\Return $[\mbox{mcm}(\MI S^*),(\mbox{mcm}(\MP S^*))^{-1}]\cap\R$
 \caption{$\mathsf{Solve\_NCP}(\MP,\MI,\MC)$ (from~\cite{zorzenon2021nonpositive})}\label{al:PIC-NCP}
\end{algorithm2e}

\subsubsection{The MPIC-NCP}\label{su:MPIC-NCP}

A natural generalization of the PIC-NCP is the \textit{multivariable PIC-NCP} (MPIC-NCP), in which matrix $A$ takes the form
\[
    A(\lambda_1,\dots,\lambda_q) = \bigoplus_{i=1}^q (\lambda_i P_i \oplus \lambda_i^{-1} I_i) \oplus C,
\]
for some matrices $P_i,I_i,C\in \Rmax^{n\times n}$ and parameters $\lambda_i\in\R$, for all $i\in\dint{1,q}$.
In this case, it can be shown that the problem becomes "trivially intractable" in $q$,\footnote{On the other hand, assuming that $q$ is fixed, the problem remains solvable in strongly polynomial time. The proof of this fact is out of the scope of the present article.} in the sense that the set $\solNCP{\MA}$ corresponds\zor{, in the worst case,} to the solution of a system of $3^q-1$ (i.e., exponentially many) non-redundant linear inequalities (in the conventional sense) in variables $\lambda_1,\dots,\lambda_q$; in other words, $\solNCP{\MA}$ is a polytope with at most $3^q-1$ facets living in a $q$-dimensional space.
Consequently, it is unrealistic to expect to efficiently solve the MPIC-NCP, since the solution set $\solNCP{\MA}$ cannot be described in polynomial space.
Nevertheless, it is possible to verify the non-emptiness of $\solNCP{\MA}$ in (weakly) polynomial time from the following observation.
Proposition~\ref{pr:nonegset_inequality} suggests that the parameters $\lambda_i$ such that $\graph(A)\in\nonegset$ are those for which the inequality
\[
	\left(\bigoplus_{i=1}^q (\lambda_i P_i \oplus \lambda_i^{-1} I_i) \oplus C\right) \otimes x \preceq x
\] 
admits a solution $x\in\R^{n}$.
We can rewrite the inequality above in \zor{the} standard algebra as the system
\[
\left\{
\begin{array}{ccl}
    \displaystyle\max_{i\in\dint{1,q},j\in\dint{1,n}} \left( (P_i)_{1j} + \lambda_i + x_j, (I_i)_{1j}-\lambda_i + x_j, C_{1j} + x_j\right) &\leq& x_1 \\
    \vdots&&\\
    \displaystyle\max_{i\in\dint{1,q},j\in\dint{1,n}} \left( (P_i)_{nj} + \lambda_i + x_j, (I_i)_{nj}-\lambda_i + x_j, C_{nj} + x_j\right) &\leq& x_n ,
\end{array}
\right.
\] 
which is equivalent to
\begin{equation}\label{eq:linear_inequalities}
\left\{
\begin{array}{rclcr}
    (P_i)_{1j} + \lambda_i + x_j &\leq& x_1 &&\quad \forall i\in\dint{1,q},j\in\dint{1,n}\\
    (I_i)_{1j}-\lambda_i + x_j &\leq& x_1 &&\quad \forall i\in\dint{1,q},j\in\dint{1,n}\\
    C_{1j} + x_j &\leq& x_1 &&\quad \forall j\in\dint{1,n}\\
        &\vdots&&&\\
    (P_i)_{nj} + \lambda_i + x_j &\leq& x_n &&\quad \forall i\in\dint{1,q},j\in\dint{1,n}\\
    (I_i)_{nj}-\lambda_i + x_j &\leq& x_n &&\quad \forall i\in\dint{1,q},j\in\dint{1,n}\\
    C_{nj} + x_j &\leq& x_n &&\quad \forall j\in\dint{1,n}.
\end{array}
\right.
\end{equation}
The system above consists of at most\footnote{The inequalities corresponding to $-\infty$-elements in matrices $P_i$, $I_i$, and $C$ are automatically satisfied and can be ignored.} $(2q+1)n^2$ linear inequalities in $n+q$ real unknowns $x_1,\dots,x_n,\lambda_1,\dots,\lambda_q$.
Therefore, the non-emptiness of its solution set can be checked in polynomial time using linear programming techniques.

\subsection{Max-plus linear-dual inequalities}\label{su:LDIs}

In the following, we define max-plus linear-dual inequalities (LDIs) from a purely formal point of view.
Their application to describe the dynamics of P-time event graphs is discussed in the next section.
Let $A^0,A^1\in\Rmax^{n\times n}$, $B^0,B^1\in\Rmin^{n\times n}$, and $K\in\nat\cup\{+\infty\}$.
LDIs are systems of $(\oplus,\otimes)$- and $(\splus,\stimes)$-linear dynamical inequalities in the \textit{dater function} $x:\dint{1,K}\rightarrow\R^n$ of the form
\begin{equation}\label{eq:dynamics_LDIs}
	\begin{array}{rrcl}
		\forall k\in\dint{1,K}, &
		A^0\otimes x(k) \preceq & x(k) & \preceq B^0\stimes x(k)\\
		\forall k\in\dint{1,K-1}, &
		\quad\quad A^1\otimes x(k) \preceq & x(k+1) & \preceq B^1\stimes x(k)
	\end{array}
	~.
\end{equation}

A finite (when $K\in\nat$) or infinite (when $K=+\infty$) trajectory $\{x(k)\}_{k\in\dint{1,K}}$ of length $K$ is consistent if it satisfies~\eqref{eq:dynamics_LDIs} for all $k$.
It is often useful in practice to restrict the focus on the simple class of \textit{$1$-periodic trajectories}, which are those of the form $\{\lambda^{k-1}x(1)\}_{k\in\dint{1,K}}$; in \zor{the} standard algebra, 1-periodic trajectories are those that satisfy: $\forall k\in\dint{1,K}$, $i\in\dint{1,n}$, $x_i(k) = (k-1)\times \lambda + x_i(1)$.
The number $\lambda$ is called \textit{period} or \textit{cycle time} of the 1-periodic trajectory.

In the following, we recall how to verify the existence of 1-periodic trajectories for given LDIs.
Substituting in~\eqref{eq:dynamics_LDIs} the formula $x(k) = \lambda^{k-1} x(1)$, we obtain
\[
	\begin{array}{rrcl}
	\forall k\in\dint{1,K}, &
		A^0\otimes \lambda^{k-1} x(1) \preceq & \lambda^{k-1} x(1) & \preceq B^0\stimes \lambda^{k-1} x(1)\\
	\forall k\in\dint{1,K-1}, &
		\quad\quad A^1\otimes \lambda^{k-1} x(1) \preceq & \lambda^k x(1) & \preceq B^1\stimes \lambda^{k-1} x(1)
	\end{array}
	~,
\]
which, after multiplying \zor{left and right hand sides of the above inequalities} by $(\lambda^{k-1})^{-1}$,\footnote{Recall that, in \zor{the} standard algebra, this is equivalent to subtracting everywhere by $(k-1)\times \lambda$.} simplifies to
\begin{equation}\label{eq:aux}
	\begin{array}{rcl}
		A^0\otimes x(1) \preceq & x(1) & \preceq B^0\stimes x(1)\\
		A^1\otimes x(1) \preceq & \lambda x(1) & \preceq B^1\stimes x(1)
	\end{array}
	~.
\end{equation}
We recall the following result.

\begin{proposition}[\cite{cuninghame2012minimax}]\label{pr:max-min}
Let $x,y\in\R^n$, $A,B\in\Rmax^{n\times n}$.
Then\footnote{\zor{The reader familiar with residuation theory~\cite{residuationtheory} may note that function $f_A^\sharp(y) = A^{\sharp}\stimes y$ coincides with the residual of the function "left multiplication by $A$" $f_A(x) = A\otimes x$ (i.e., $f_A^\sharp(y)$ is the greatest $x$ that satisfies $A\otimes x\leq y$).
In the max-plus context, function $f_A^\sharp$ is often called "left division by $A$" and denoted by $f_A^\sharp(y) = A \maxldiv y$~\cite{baccelli1992synchronization}.
Dually, function $g_B^\flat(x) = B^\sharp \otimes x$ coincides with the dual residual of the function "left dual multiplication by $B$" $g_B(y) = B \stimes y$ (i.e., $g_B^\flat(x)$ is the least $y$ that satisfies $x \leq B\stimes y$), and is sometimes denoted by $g_B^\flat(x) = B \minldiv x$~\cite{brunsch2012duality}.}}
\[
	x\preceq A^\sharp \stimes y \ \Leftrightarrow \ A\otimes x\preceq y,
\]
and
\[
    \left\{
    \begin{array}{l}
        A\otimes x\preceq y\\
        B\otimes x\preceq y
    \end{array}
    \right.
    \quad
    \Leftrightarrow
    \quad
    (A\oplus B) \otimes x\preceq y.
\]
\end{proposition}

Because of Proposition~\ref{pr:max-min} we can rewrite~\eqref{eq:aux} as
\[
	(\lambda B^{1\sharp} \oplus \lambda^{-1} A^1 \oplus A^0\oplus B^{0\sharp}) \otimes x(1) \preceq x(1),
\]
which, from Proposition~\ref{pr:nonegset_inequality}, admits solution $x(1)\in\R^n$ if and only if $\graph(\lambda B^{1\sharp} \oplus \lambda^{-1} A^1 \oplus A^0\oplus B^{0\sharp})\in\nonegset$.
Note that we obtained a PIC-NCP with matrices $P = B^{1\sharp}$, $I = A^1$, and $C = A^0 \oplus B^{0\sharp}$.
Therefore, a consistent $1$-periodic trajectory exists if and only if $\solNCP{\lambda B^{1\sharp} \oplus \lambda^{-1}A^1 \oplus A^0 \oplus B^{0\sharp}} = [\lambda_{\textup{min}},\lambda_{\textup{max}}]\cap \R$ is nonempty, where $\lambda_{\textup{min}}$ and $\lambda_{\textup{max}}$ can be found in time $\pazocal{O}(n^4)$ using Algorithm~\ref{al:PIC-NCP}.

%

%
%
%
%

\section{P-time event graphs}\label{se:P-time_event_graphs}

\begin{definition}[From~\cite{CALVEZ19971487}]
An \zor{ordinary (or unweighted)} P-time Petri net (P-TPN) is a 5-tuple $(\places,\transitions,\arcs,m,\iota)$, where $(\places\cup\transitions,E)$ is a directed graph in which the set of nodes is partitioned into the set of places, $\places$, and the set of transitions, $\transitions$, the set of arcs, $\arcs$, is such that $\arcs\subseteq (\places\times \transitions)\cup(\transitions\times\places)$, $m:\places\rightarrow\nato$ is a map such that $m(p)$ represents the number of tokens initially residing in place $p\in\places$ (also called initial marking of $p$), and 
\[
    \iota:\places\rightarrow\{[\tau^-,\tau^+]\ |\ \tau^-\in\R_{\geq 0},\tau^+\in\R_{\geq 0}\cup\{\infty\},\tau^-\leq \tau^+\}
\]
is a map that associates to every place $p\in\places$ a time interval $\iota(p)=[\tau^-_p,\tau^+_p]$.
\end{definition}

In the following, we briefly describe the dynamics of an \zor{ordinary} P-TPN\footnote{Since the functioning of P-TPNs is here reported only to present P-time event graphs, we will not delve into the nuances of weak and strong semantics for P-TPNs.
For a thorough description of different types of semantics, we invite the reader to consult~\cite{boyer2008compared}.}.
A transition $t$ is enabled when either it has no upstream places or each upstream place $p$ of $t$ contains at least one token which has resided in $p$ for a time between $\tau^-_{p}$ and $\tau^+_{p}$ (extremes included).
When transition $t$ is enabled, it may fire; its firing causes one token to be removed instantaneously from each of the upstream places of $t$, and one token to be added, again instantaneously, to each of the downstream places of $t$.
If a token sojourns more than $\tau^+_{p}$ time instants in a place $p$, then the token becomes \textit{dead}.

A P-time event graph (P-TEG) is an \zor{ordinary} P-TPN in which every place has exactly one upstream and one downstream transition.
Let $|\transitions|=n$ be the number of transitions in a P-TEG and let $x:\dint{1,K}\rightarrow\R^n$ be a \textit{dater function} of length $K\in\nat\cup\{+\infty\}$, i.e., a function such that $x_i(k)$ represents the time at which transition $t_i$ fires for the $k$\textsuperscript{th} time.
Since the $(k+1)$\textsuperscript{st} firing of any transition cannot occur before the $k$\textsuperscript{th}, we require the dater to be a non-decreasing function, i.e., $\forall i\in\dint{1,n}$, $x_i(k+1)\geq x_i(k)$.
The evolution of the marking in a P-TEG is entirely described by its corresponding dater trajectory $\{x(k)\}_{k\in\dint{1,K}}$, and if a non-decreasing dater trajectory exists for which no token death occurs, then it is said to be \textit{consistent} for the P-TEG.

It is always possible to transform a P-TEG into an equivalent one whose places have at most $1$ initial token each~\cite{vspavcek2017analysis}. 
Therefore, in the following we will only focus on P-TEGs in which the initial marking $m(p)$ is either $0$ or $1$ for each place $p\in\places$. 
Under this assumption, a consistent trajectory for a given P-TEG must satisfy LDIs as in~\eqref{eq:dynamics_LDIs},
where matrices $A^0,A^1\in\Rmax^{n\times n}$, $B^0,B^1\in\Rmin^{n\times n}$ \zor{are called} \textit{characteristic matrices} of the P-TEG, and are defined as follows.
If there exists a place $p_{ij}$ with initial marking $\mu\in\{0,1\}$, upstream transition $t_j$ and downstream transition $t_i$, then $A^\mu_{ij}=\tau^-_{p_{ij}}$ and $B^\mu_{ij}=\tau^+_{p_{ij}}$; otherwise, $A^\mu_{ij} = -\infty$ and $B^\mu_{ij} = \infty$.

Before introducing some structural properties of P-TEGs, it is useful to clarify the role of initial conditions for these dynamical systems.

\subsection{Initial conditions}\label{su:initial_conditions}

Depending on the P-TEG's intended application, the restrictiveness of initial conditions may vary; in the following, we introduce two alternatives.

\subsubsection{Loose initial conditions}

\zor{Inequalities}~\eqref{eq:dynamics_LDIs} suggest that $x(1)$ can assume any value in $\R^n$, as long as it satisfies
\begin{equation}\label{eq:initial_conditions}
	\begin{array}{c}
		A^0\otimes x(1) \preceq x(1) \preceq B^0\stimes x(1)\\
		A^1\otimes x(1) \preceq B^1\stimes x(1)
	\end{array}
	~.
\end{equation}
Note that~\eqref{eq:initial_conditions} does not restrict the first firing time based on the arrival times of the initial tokens in the Petri net; the arrival times of initial tokens are, in fact, not even defined.
For this reason, we say that the initial conditions of a P-TEG are \textit{loose} if no other restriction other than~\eqref{eq:initial_conditions} applies to $x(1)$.

P-TEGs with loose initial conditions evolve entirely according to LDIs, and are suitable to model manufacturing systems operating in periodic regime, after a transient period has passed (an example is given in Section~\ref{su:cycle_time_analysis}).
Another scenario where loose initial conditions may be convenient is when the time-window constraints need to be fulfilled only after the occurrence of the first events, as shown in the following example.

\begin{example}[Heat treatment line]\label{ex:heat_treatment_loose}
This example is adapted from~\cite{zorzenon2020bounded}.
Consider a simple heat treatment line, schematically represented in Figure~\ref{fi:furnace_A}, consisting of a furnace, which performs a heat treatment, and an autonomous guided vehicle that receives processed pieces and transports them to the next stage.
Both the furnace and the vehicle have unitary capacity, i.e., they can process and transport one piece at a time, respectively.
The heat treatment must last between 2 and 3 time units, and the autonomous vehicles takes (at least) 0.5 time units both to transport a processed piece to the next stage and to travel back to the furnace.
Customers' demand imposes that the time difference between subsequent unloadings of processed pieces from the autonomous guided vehicle must not exceed 4 time units; the specification needs to be met for all pieces after the first one.
Moreover, each piece must spend at least 6 time units in the processing line, from the moment it enters the furnace to the one it is removed from the vehicle, in order to synchronize with other processing stages.
Initially, the furnace is empty and the vehicle is waiting for a piece at the furnace.
The P-TEG in Figure~\ref{fi:TEG_furnace} models the described plant; a firing of transitions $t_1$, $t_2$, and $t_3$ represents, respectively, the arrival of an unprocessed piece in the furnace, the loading of a processed piece onto the autonomous guided vehicle, and the unloading of a piece from the vehicle.

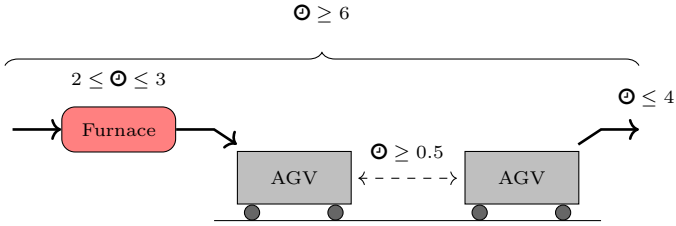
\begin{figure}
	\centering
	\begin{tikzpicture}
\footnotesize 
\node (O) at (0,0) {};
\node[draw,rectangle,minimum height=.6cm,minimum width=1.5cm,fill=red!50!white,rounded corners=5pt,align=center] (F) at (1.5,0) {Furnace};
\draw[line width = 1pt,->] (O) -- (F.west);

\node at ($(F.north) + (0,1.3em)$) {$2 \leq\mbox{\faClockO}\leq 3$};

\newcommand{\carH}{.7}
\newcommand{\carL}{1.5}
\newcommand{\carW}{.2}

\coordinate (A0) at ($(F.east) + (.8,-1)$);

\draw[line width = 1pt,->] (F.east) -- ($(A0) + (-.3,1)$) -- ($(A0) + (0,.8)$);

\node[draw, rectangle, minimum height=\carH cm,minimum width=\carL cm,anchor=south west,align=center,fill=gray!50!white] (C1) at (A0) {AGV};
\node[draw, circle, minimum size = \carW cm, inner sep = 0, anchor = north,fill=black!60!white] (W1) at ($(A0) + (.2,0)$) {};
\node[draw, circle, minimum size = \carW cm, inner sep = 0, anchor = north,fill=black!60!white] (W2) at ($(A0) + (\carL-.2,0)$) {};

\coordinate (A1) at ($(A0) + (\carL+1.5,0)$);

\node[draw, rectangle, minimum height=\carH cm,minimum width=\carL cm,anchor=south west,align=center,fill=gray!50!white] (C2) at (A1) {AGV};
\node[draw, circle, minimum size = \carW cm, inner sep = 0, anchor = north,fill=black!60!white] (W3) at ($(A1) + (.2,0)$) {};
\node[draw, circle, minimum size = \carW cm, inner sep = 0, anchor = north,fill=black!60!white] (W4) at ($(A1) + (\carL-.2,0)$) {};

\draw[<->,dashed] ($(A0) + (\carL+.1,\carH/2)$) -- ($(A1) + (-.1,\carH/2)$)node[midway,yshift=1.2em]{$\mbox{\faClockO}\geq 0.5$};

\coordinate (end) at ($(A1) + (\carL+.8,1)$);
\draw[line width = 1pt,->] ($(A1) + (\carL,.8)$) -- ($(A1) + (\carL+.3,1)$) -> (end);

\draw ($(W1.south) - (.5,0)$) -- ($(W4.south) + (.5,0)$);

\draw [decorate,decoration={brace,amplitude=5pt,raise=8ex}]
  ($(O)+(0,0)$) -- ($(end)+(0,0)$) node[midway,yshift=5.5em]{$\mbox{\faClockO}\geq 6$};

\node at ($(end) + (0.1,1.5em)$) {$\mbox{\faClockO}\leq 4$};
\end{tikzpicture}
	\caption{Illustration of the heat treatment line of Example~\ref{ex:heat_treatment_loose}.}\label{fi:furnace_A}
\end{figure}
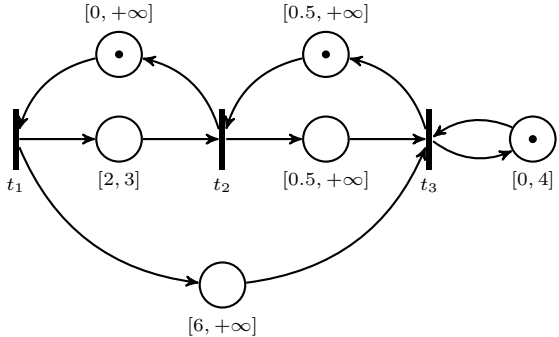
\begin{figure}
	\centering
	\begin{tikzpicture}[node distance=.5cm and 1cm,>=stealth',bend angle=45,thick]
\tikzstyle{place}=[circle,thick,draw=black,minimum size=6mm]
\tikzstyle{transitionV}=[rectangle,thick,fill=black,minimum height=8mm,inner xsep=1pt]

\footnotesize

\node [transitionV,label=below:$t_1$] (x1) {};
\node [place,tokens=0,label=below:{$[2,3]$}] (p21) [right= of x1] {};
\node [place,tokens=1,label=above:{$[0,+\infty]$}] (p12) [above= of p21] {};
\node [transitionV,label=below:$t_2$] (x2) [right=of p21] {};
\node (p22) [below= of x2] {};
\node [place,tokens=0,label=below:{$[0.5,+\infty]$}] (p32) [right= of x2] {};
\node [place,tokens=1,label=above:{$[0.5,+\infty]$}] (p23) [above= of p32] {};
\node [transitionV,label=below:$t_3$] (x3) [right=of p32] {};
\node [place,tokens=1,label=below:{$[0,4]$}] (p33) [right= of x3] {};
\node [place,tokens=0,label=below:{$[6,+\infty]$}] (p31) [below= of p22] {};

\draw (x1) edge[->] (p21);
\draw (p21) edge[->] (x2);
\draw (x2) edge[->] (p32);
\draw (p32) edge[->] (x3);
\draw (x1) edge[bend right=30,->] (p31);
\draw (p31) edge[bend right=30,->] (x3);
\draw (x2) edge[bend right=30,->] (p12);
\draw (p12) edge[bend right=30,->] (x1);
\draw (x3) edge[bend right=30,->] (p23);
\draw (p23) edge[bend right=30,->] (x2);
\draw (x3) edge[bend right=30,->] (p33);
\draw (p33) edge[bend right=30,->] (x3);
\end{tikzpicture}
	\caption{P-TEG representing the heat treatment line.}\label{fi:TEG_furnace}
\end{figure}

The characteristic matrices of the P-TEG are 
\[
	A^0 = \begin{bmatrix}
		-\infty & -\infty & -\infty \\
		2 & -\infty & -\infty \\
		6 & 0.5 & -\infty
	\end{bmatrix},\quad
	A^1 = \begin{bmatrix}
		-\infty & 0 & -\infty \\
		-\infty & -\infty & 0.5 \\
		-\infty & -\infty & 0
	\end{bmatrix},
\]
\[
	B^0 = \begin{bmatrix}
		+\infty & +\infty & +\infty \\
		3 & +\infty & +\infty \\
		+\infty & +\infty & +\infty
	\end{bmatrix},\quad
	B^1 = \begin{bmatrix}
		+\infty & +\infty & +\infty \\
		+\infty & +\infty & +\infty \\
		+\infty & +\infty & 4
	\end{bmatrix}.
\]
It is possible to verify that 
\[
x(1) = \begin{bmatrix}
	0\\3\\6
\end{bmatrix},\quad \forall k\geq 1,\ x(k+1) = 3.5 x(k)
\]
is a consistent trajectory for the P-TEG under loose initial conditions.
Observe that the first firing time of transition $t_3$ does not violate the upper bound associated to place $p_{33}$ (i.e., the place that is upstream and downstream of transition $t_3$), even though $x_3(1) = 6 > 4 = B_{33}^1$.
Indeed, the sojourn time of the initial token in place $p_{33}$ does not restrict the dynamics of the P-TEG. 
This is convenient from a practical point of view, as the constraint on the processing rate must be enforced only after the first piece leaves the plant.
\end{example}

\subsubsection{Strict initial conditions}\label{su:strict}

For the considered application, it may be necessary to impose further restrictions on the initial conditions.
Here we take into account the amount of time that initial tokens have resided in places prior to the initial time $t_0\in\R$; we call this value the \textit{time tag} of the token\footnote{The concept of time tags is analogous to that of \textit{lag times}, defined as in~\cite{baccelli1992synchronization}.}.
Time tags can be useful, for instance, in manufacturing, to specify that some machines have already been processing a part since time $t_0-\tau$, or in transportation, to indicate that a vehicle has left a station at time $t_0-\tau$, for $\tau\geq 0$.

Let $\rho$ be a function that associates a time tag to every place with an initial token in a P-TEG.
Formally, if there is a place $p_{ij}$ with marking $m(p_{ij}) = 1$, upstream transition $t_j$, and downstream transition $t_i$, then we denote $\rho(p_{ij}) = \rho_{ij}\in\R_{\geq 0}$, otherwise $\rho(p_{ij})$ is not defined.
Then, in addition to~\eqref{eq:dynamics_LDIs}, the first firing time of the transitions of the P-TEG must satisfy, for all $i,j\in\dint{1,n}$,
\[
	A^1_{ij} + t_0 - \rho_{ij} \leq x_i(1) \leq B^1_{ij} + t_0 - \rho_{ij};
\]
for each $i,j$, the inequality specifies that the first firing time of transition $t_i$ must not violate the time-window constraint $[A^1_{ij},B^1_{ij}]$ associated with place $p_{ij}$, considering that the initial token of this place arrived at time $t_0 - \rho_{ij}$.
In the max-plus algebra, the latter inequalities can be expressed as
\begin{equation}\label{eq:time_tags}
	\underline{\Delta} \otimes t_0\tilde{e} \leq x(1) \leq \overline{\Delta} \stimes t_0\tilde{e},
\end{equation}
where
\[
	\underline{\Delta}_{ij} = \begin{dcases}
		A^1_{ij} -\rho_{ij} & \mbox{if } A^1_{ij} \neq -\infty,\\
		-\infty & \mbox{otherwise},
	\end{dcases}
	\quad
	\overline{\Delta}_{ij} = \begin{dcases}
		B^1_{ij} -\rho_{ij} & \mbox{if } B^1_{ij} \neq +\infty,\\
		+\infty & \mbox{otherwise},
	\end{dcases}
\]
and
\[
	\tilde{e} = \begin{bmatrix}0\\0\\\vdots\\0\end{bmatrix}\in\R^n.
\]

Note that other possible definitions for $\underline{\Delta}$ and $\overline{\Delta}$, corresponding to different requirements for the first firings of transitions, may be considered.
In general, given any $\underline{\Delta}\in\Rmax^{n\times n}$, $\overline{\Delta}\in\Rmin^{n\times n}$ such that $(\underline{\Delta},\overline{\Delta}) \neq (\pazocal{E},\pazocal{T})$, inequality~\eqref{eq:time_tags} restricts the set of consistent trajectories for a P-TEG; hence, we say that the initial conditions of a P-TEG are \textit{strict} if $x(1)$ is required to satisfy them for some $(\underline{\Delta},\overline{\Delta})\neq(\pazocal{E},\pazocal{T})$. 
We will refer to consistent trajectories with either loose or strict initial conditions depending if they satisfy only~\eqref{eq:dynamics_LDIs} or also~\eqref{eq:time_tags}.
Note that, without loss of generality, we can assume that $t_0=0$, as P-TEGs (with either loose or strict initial conditions) are time-invariant systems, i.e., if $\{x(k)\}_{k\in\dint{1,K}}$ is a consistent trajectory, then $\{t_0 \otimes x(k)\}_{k\in\dint{1,K}}$ is consistent as well for any $t_0\in\R$. 
In other words, the choice of the initial time $t_0$ does not affect the dynamics of P-TEGs.

\begin{example}[Heat treatment line, cont.]
Consider again the P-TEG of Figure~\ref{fi:TEG_furnace}.
It is not difficult to see that, if we assign a time tag to each initial token of the P-TEG, no consistent trajectory that satisfies strict initial conditions can be found; indeed, it is not possible to fire transition $t_3$ before time $4-\rho_{33}$, for any time tag $\rho_{33}\in \R_{\geq 0}$. 
So, let us modify the configuration of the initial tokens as in Figure~\ref{fi:TEG_furnace_time_tags}; time tags are indicated in the figure.

\begin{figure}
	\centering
	\begin{tikzpicture}[node distance=.5cm and 1cm,>=stealth',bend angle=45,thick]
\tikzstyle{place}=[circle,thick,draw=black,minimum size=6mm]
\tikzstyle{transitionV}=[rectangle,thick,fill=black,minimum height=8mm,inner xsep=1pt]

\footnotesize

\node [transitionV,label=below:$t_1$] (x1) {};
\node [place,tokens=0,label=below:{$[2,3]$}] (p21) [right= of x1] {};
\node [place,tokens=1,label=above:{$[0,+\infty]$},label=below:{$\rho_{12} = 0.5$}] (p12) [above= of p21] {};
\node [transitionV,label=below:$t_2$] (x2) [right=of p21] {};
\node (p22) [below= of x2] {};
\node [place,tokens=1,label=above:{$[0.5,+\infty]$},label=below:{$\rho_{32}=0.5$}] (p32) [right= of x2] {};
\node [place,tokens=0,label=above:{$[0.5,+\infty]$}] (p23) [above= of p32] {};
\node [transitionV,label=below:$t_3$] (x3) [right=of p32] {};
\node [place,tokens=1,label=above:{$[0,4]$},label=below:{$\rho_{33}=1$}] (p33) [right= of x3] {};
\node [place,tokens=1,label=above:{$[6,+\infty]$},label=below:{$\rho_{31}=3$}] (p31) [below= of p22] {};

\draw (x1) edge[->] (p21);
\draw (p21) edge[->] (x2);
\draw (x2) edge[->] (p32);
\draw (p32) edge[->] (x3);
\draw (x1) edge[bend right=30,->] (p31);
\draw (p31) edge[bend right=30,->] (x3);
\draw (x2) edge[bend right=30,->] (p12);
\draw (p12) edge[bend right=30,->] (x1);
\draw (x3) edge[bend right=30,->] (p23);
\draw (p23) edge[bend right=30,->] (x2);
\draw (x3) edge[bend right=30,->] (p33);
\draw (p33) edge[bend right=30,->] (x3);
\end{tikzpicture}
	\caption{P-TEG representing the heat treatment line in another initial configuration.}\label{fi:TEG_furnace_time_tags}
\end{figure}
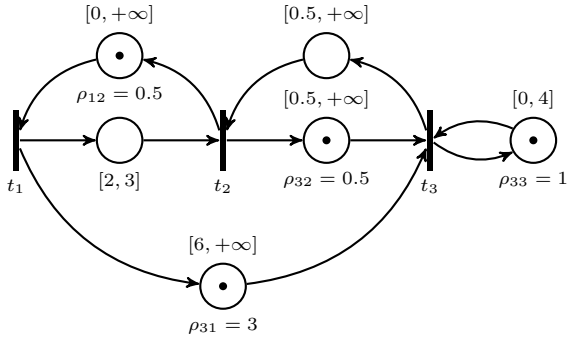

The interpretation is that:
\begin{itemize}
	\item a piece is inside the heat treatment line since time $t_0 - 3$, as $\rho_{31} = 3$,
	\item an autonomous guided vehicle is at the unloading location with a processed piece from time $t_0$, as $\rho_{32} = 0.5 = A^1_{32}$,
	\item the furnace has completed the last heat treatment at time $t_0 - 0.5$, as $\rho_{12} = 0.5$, and
	\item the first processed piece is required to leave the heat treatment plant before time $t_0 + 3$, as $\rho_{33} = 1$.
\end{itemize}

The characteristic matrices for this example are 
\[
	A^0 = \begin{bmatrix}
		-\infty & -\infty & -\infty \\
		2 & -\infty & 0.5 \\
		-\infty & -\infty & -\infty
	\end{bmatrix},\quad
	A^1 = \begin{bmatrix}
		-\infty & 0 & -\infty \\
		-\infty & -\infty & -\infty \\
		6 & 0.5 & 0
	\end{bmatrix},
\]
\[
	B^0 = \begin{bmatrix}
		+\infty & +\infty & +\infty \\
		3 & +\infty & +\infty \\
		+\infty & +\infty & +\infty
	\end{bmatrix},\quad
	B^1 = \begin{bmatrix}
		+\infty & +\infty & +\infty \\
		+\infty & +\infty & +\infty \\
		+\infty & +\infty & 4
	\end{bmatrix},
\]
and the matrices $\underline{\Delta}$ and $\overline{\Delta}$ are
\[
	\underline{\Delta} = \begin{bmatrix}
		-\infty & -0.5 & -\infty \\
		-\infty & -\infty & -\infty \\
		3 & 0 & -1
	\end{bmatrix},\quad
	\overline{\Delta} = \begin{bmatrix}
		+\infty & +\infty & +\infty \\
		+\infty & +\infty & +\infty \\
		+\infty & +\infty & 3
	\end{bmatrix}.
\]

Assuming that $t_0 = 0$, the following is a consistent trajectory for the P-TEG under strict initial conditions:
\[
x(1) = \begin{bmatrix}
	1\\4\\3
\end{bmatrix},\quad \forall k\geq 1,\ x(k+1) = 4 x(k).
\]
\end{example}

Despite their usefulness in applications, strict initial conditions present an additional complexity: the inequalities describing the dynamics of of P-TEGs with strict initial conditions are not (pure) LDIs.
This means that mathematical results in LDIs can be directly applied to P-TEGs with loose but not with strict initial conditions; this will be made evident in the following section.
In Section~\ref{su:SLDIS_PTEGs}, we will see that the dynamics of P-TEGs with strict initial conditions falls in the category of switched LDIs.

%
%

\subsection{Structural properties}\label{su:structural_PTEGs}

In this section we recall the definition of some structural properties of P-TEGs.
These properties can be equivalently stated for P-TEGs under loose or strict initial conditions.

A P-TEG is said to be \textit{consistent} if it admits a consistent, non-decreasing trajectory $\{x(k)\}_{k\in\nat}$ of infinite length.
The non-decreasingness of the dater trajectory, equivalent to having $x(k+1)\succeq x(k)$ for all $k\in\dint{1,K-1}$, is a natural requirement, as the $(k+1)$\textsuperscript{st} firing of a transition cannot occur before the $k$\textsuperscript{th} one; because of Proposition~\ref{pr:max-min}, to restrict the evolution of the dater trajectory such that this restriction is always satisfied, it is sufficient to modify the definition of matrix $A^1$ into $A^1 \oplus E_\otimes$.

We say that a trajectory $\{x(k)\}_{k\in\nat}$ is \textit{delay-bounded} if there exists a positive real number $M$ such that, for all $i,j\in\dint{1,n}$ and for all $k\in\nat$, $x_i(k)-x_j(k)<M$; a P-TEG admitting a consistent delay-bounded trajectory of the dater function is said to be \textit{boundedly consistent}.
Although in consistent P-TEGs it is possible to find a marking evolution such that no time-window constraint is violated, if the stronger property of bounded consistency does not hold, any consistent, infinite trajectory will accumulate unbounded delay between the firing times of two distinct transitions.
This phenomenon is usually not desirable in manufacturing systems represented by P-TEG, where the firings of transitions represent the start or end of processes, and the $k$\textsuperscript{th} product entering the system is finished when all transitions fire for the $k$\textsuperscript{th} time.
Indeed, in this context it implies that the total time the $k$\textsuperscript{th} product spends in the manufacturing system increases without bounds with $k$.

Analogously to LDIs, in P-TEGs we say that dater trajectories of the form $\{\lambda^{k-1}x(1)\}_{k\in\dint{1,K}}$ are 1-periodic with period $\lambda\in\R_{\geq 0}$.
Clearly, in P-TEGs with loose initial conditions, 1-periodic trajectories can be found in time complexity $\pazocal{O}(n^4)$ using Algorithm~\ref{al:PIC-NCP}, as their evolution satisfies LDIs.
To our knowledge, no algorithm that checks whether a P-TEG is consistent has been found until now; on the other hand, bounded consistency of P-TEGs with loose initial conditions can be verified in time $\pazocal{O}(n^4)$.
This fact comes from the following result.

\begin{theorem}[\cite{zorzenon2020bounded}]\label{th:bounded_consistency}
	A P-TEG with loose initial conditions is boundedly consistent if and only if it admits a consistent 1-periodic trajectory.
\end{theorem}

On the other hand, \zor{an} analogous result for the case with strict initial conditions \zor{does not hold}: boundedly consistent P-TEGs with strict initial conditions may admit no 1-periodic trajectory, as shown in the following example. 

\begin{example}
\begin{figure}
	\centering
		\begin{tikzpicture}[node distance=1.2cm and 1.5cm,>=stealth',bend angle=30,thick,on grid]
\footnotesize
\node[transitionV,label=below:{$t_1$}] (t1) {};
\node[place,tokens=1,right=of t1,label=above:{$[0,+\infty]$},label=below:{$\rho_{21}=0$}] (p12) {};
\node[place,tokens=1,below=of p12,label=above:{$[0,+\infty]$},label=below:{$\rho_{12}=0$}] (p21) {};
\node[transitionV,right=of p12,label=below:{$t_2$}] (t2) {};
\node[place,tokens=1,left=of t1,label=above:{$[10,10]$},label=below:{$\rho_{11} = 10$}] (p11) {};
\node[place,tokens=1,right=of t2,label=above:{$[10,10]$},label=below:{$\rho_{32}=9$}] (p23) {};
\node[place,tokens=1,below right=of t2,label=above:{$[10,10]$},label=below:{$\rho_{23}=10$}] (p32) {};
\node[transitionV,right=of p23,label=below:{$t_3$}] (t3) {};

\draw (t1) edge[->] (p12);
\draw (p12) edge[->] (t2);
\draw (t2.-90-15) edge[->,bend left=15] (p21);
\draw (p21) edge[->,bend left=15] (t1.-90+15);
\draw (t1.180-15) edge[bend right,->] (p11);
\draw (p11) edge[bend right,->] (t1.180+15);
\draw (t2) edge[->] (p23);
\draw (p23) edge[->] (t3);
\draw (t3.-90-15) edge[->,bend left=15] (p32);
\draw (p32) edge[->,bend left=15] (t2.-90+15);

\end{tikzpicture}
	\caption{Example of boundedly consistent P-TEG with strict initial conditions that admits no 1-periodic trajectory.}
	\label{fig:P-TEG_example_2_periodic}
\end{figure}
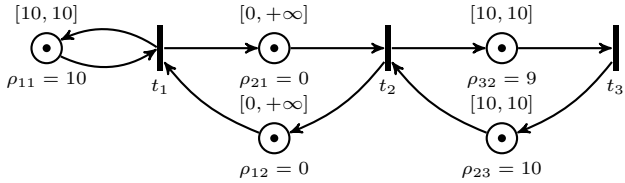

Using an algorithm that will be presented in Section~\ref{su:SLDIS_PTEGs}, it can be shown that the P-TEG with strict initial conditions in Figure~\ref{fig:P-TEG_example_2_periodic} admits no 1-periodic trajectory.
However, assuming that $t_0 = 0$, it admits the following delay-bounded (2-periodic) trajectory: 
\[
	x(1) = \begin{bmatrix}
		0\\0\\1
	\end{bmatrix},\quad x(2) = \begin{bmatrix}
		10\\11\\10
	\end{bmatrix}, \quad \forall k\geq 1,\ x(k+2) = 20 x(k).
\]
Therefore, it is boundedly consistent.

\end{example}

The following example illustrates the discussed properties in the case of P-TEGs with loose initial conditions.

\begin{example}\label{ex:P-TEGs}
Consider the P-TEG represented in~\Cref{fig:P-TEG_example}, in which time windows are parametrized with respect to label $\wZ$; in~\Cref{tab:P-TEG_parameters}, values of time windows are given for $\wZ\in\{\wA,\wB,\wC\}$.
The matrices characterizing the P-TEG labeled $\wZ$ are:
\[
	A^0_\wZ = 
	\begin{bmatrix}
		-\infty & -\infty \\
		0 & -\infty
	\end{bmatrix},\quad
	A^1_\wZ = 
	\begin{bmatrix}
		\alpha_\wZ & -\infty \\
		-\infty & \beta_\wZ
	\end{bmatrix},
\]\[
	B^0_\wZ = 
	\begin{bmatrix}
		\infty & \infty \\
		\infty & \infty
	\end{bmatrix},\quad
	B^1_\wZ = 
	\begin{bmatrix}
		\alpha_\wZ & \infty \\
		\infty & \beta_\wZ
	\end{bmatrix}.
\]
We analyze structural properties of the P-TEGs under loose initial conditions.

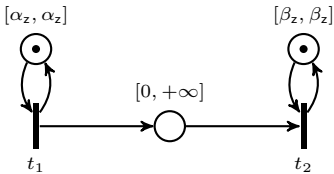
\begin{minipage}{\linewidth}
\begin{minipage}[b]{.4\linewidth}
\centering
\begin{tikzpicture}[node distance=.5cm and 1.5cm,>=stealth',bend angle=30,thick]
\footnotesize
\node[transitionV,label=below:{$t_1$}] (t1) {};
\node[place,right=of t1,label=above:{$[0,+\infty]$}] (p12) {};
\node[transitionV,right=of p12,label=below:{$t_2$}] (t2) {};
\node[place,tokens=1,above=of t1,label=above:{$[\alpha_\wZ,\alpha_\wZ]$}] (p11) {};
\node[place,tokens=1,above=of t2,label=above:{$[\beta_\wZ,\beta_\wZ]$}] (p22) {};

\draw (t1) edge[->] (p12);
\draw (p12) edge[->] (t2);
\draw (t1.90-15) edge[bend right,->] (p11);
\draw (p11) edge[bend right,->] (t1.90+15);
\draw (t2.90-15) edge[bend right,->] (p22);
\draw (p22) edge[bend right,->] (t2.90+15);

\end{tikzpicture}
\captionof{figure}{Example of P-TEG.}
	\label{fig:P-TEG_example}
\end{minipage}
\hspace{30pt}
\begin{minipage}[b]{.4\linewidth}
	\centering
		\begin{tabular}{ccc}
			$\wZ$ & $\alpha_\wZ$ & $\beta_\wZ$ \\\hline
			$\wA$ & 2 & 1 \\
			$\wB$ & 1 & 2 \\ 
			$\wC$ & 1 & 1
		\end{tabular}
\captionof{table}{Parameters for the P-TEG of~\Cref{fig:P-TEG_example}.}
\label{tab:P-TEG_parameters}
	\end{minipage}
\end{minipage}

	Since lower and upper bounds for the sojourn times of the two places with an initial token coincide, once the vector of first firing times $x_\wZ(1)$ is chosen (such that the first inequality in~\eqref{eq:dynamics_LDIs} is satisfied for $k=1$, i.e., $x_{\wZ,2}(1)\geq x_{\wZ,1}(1)$), the only infinite trajectory $\{x_\wZ(k)\}_{k\in\nat}$ that is a candidate to be consistent for the P-TEG labeled $\wZ$ is deterministically given by
	\[
		\forall k\in\nat,\quad x_\wZ(k+1) = \begin{bmatrix}\alpha_\wZ + x_{\wZ,1}(k)\\\beta_\wZ + x_{\wZ,2}(k)\end{bmatrix}.
	\]
	However, for the case $\wZ=\wA$ it is easy to see that, for any valid choice of the vector of first firing times, the candidate trajectory $\{x_\wA(k)\}_{k\in\nat}$ is not consistent (as for a sufficiently large $k$, $x_{\wA,2}(k) < x_{\wA,1}(k)$, i.e., the first inequality of~\eqref{eq:dynamics_LDIs} is violated
	).
	For $\wZ=\wB$, candidate trajectories $\{x_\wB(k)\}_{k\in\nat}$, despite being consistent, are not delay-bounded and result in the infinite accumulation of tokens in the place between $t_1$ and $t_2$ for $k\rightarrow\infty$.
	On the other hand, $\{x_\wC(k)\}_{k\in\nat}$ is consistent and delay-bounded (in fact, it is 1-periodic with period $1$).
	Thus we can conclude that the P-TEG labeled $\wA$ is not consistent, the one labeled $\wB$ is consistent but not boundedly consistent, and the one labeled $\wC$ is boundedly consistent.
	Of course, we would have reached the same conclusion regarding bounded consistency by using Theorem~\ref{th:bounded_consistency}.
\end{example}

\section{Switched max-plus linear-dual inequalities}\label{se:switched_max-plus_linear-dual_systems}

This section introduces the class of dynamical systems called switched max-plus linear-dual inequalities (SLDIs), and demonstrates its usefulness by means of simple examples.
In Section~\ref{su:SLDIS_PTEGs}, the relationship between SLDIs and P-TEGs with strict initial conditions is examined.
Methods to efficiently verify the existence of specific trajectories are then presented in Sections~\ref{su:analysis_of_fixed_schedules} and~\ref{su:piecewise_schedules}.

\subsection{Mathematical description}\label{su:general_description}

We start by defining switched LDIs (SLDIs) as the natural extension of LDIs in which matrices $A^0,A^1,B^0,B^1$ may be different for all $k$.
Formally, SLDIs are a 5-tuple $\pazocal{S}=(\Sigma,A^0,A^1,B^0,B^1)$, where $\Sigma=\{\wA_1,\ldots,\wA_m\}$ is a finite alphabet whose symbols are called \textit{modes}, and $A^0,A^1:\Sigma \rightarrow \Rmax^{n\times n}$, $B^0,B^1:\Sigma\rightarrow \Rmin^{n\times n}$ are functions that associate a matrix to each mode of $\Sigma$; for the sake of simplicity, given a mode $\wZ\in\Sigma$, we will write $A^0_\wZ,A^1_\wZ,B^0_\wZ,B^1_\wZ$ in place of $A^0(\wZ),A^1(\wZ),B^0(\wZ),B^1(\wZ)$, respectively.
We denote by $\Sigma^*$ and $\Sigma^\omega$ the sets of finite and infinite concatenations of modes from $\Sigma$, respectively.
A \textit{schedule} $w$ is an element of $\Sigma^* \cup \Sigma^\omega$, i.e., it is either a finite or an infinite sequence of modes $w = w_1w_2\ldots w_{K}$ with $w_k\in\Sigma$ for all $k\in\dint{1,K}$, where $K\in\nat\cup\{+\infty\}$ denotes the length of schedule $w$.

The dynamics of SLDIs $\pazocal{S}$ under schedule $w\in\Sigma^*\cup\Sigma^\omega$ is expressed by the following system of inequalities:
\begin{equation}\label{eq:dynamics}
	\begin{array}{rrcl}
		\mbox{for all } k\in\dint{1,K},&
		A^0_{w_k}\otimes x(k) \preceq & x(k) & \preceq B^0_{w_k}\stimes x(k),\\
		\mbox{for all } k\in\dint{1,K-1},&
		\quad\quad A^1_{w_k}\otimes x(k) \preceq & x(k+1) & \preceq B^1_{w_k}\stimes x(k),
	\end{array}
\end{equation}
where function $x:\dint{1,K}\rightarrow \R^n$ is called \textit{dater} of $\pazocal{S}$ associated with schedule $w$.
Term $x_i(k)$ represents the occurrence time of event $i$ associated with mode\footnote{Strictly speaking, $x_i$ is the function that associates to sequence $w_1w_2\cdots w_k$ the time of the $k$\textsuperscript{th} occurrence of event $i$ (in mode $w_k$). For notational simplicity, however, we prefer to write $x(k)$ in place of $x(w_1w_2\cdots w_k)$, relying on the fact that the meaning is clear from the context.} $w_{k}$. 
Similar to P-TEGs, it is natural to assume the following non-decreasingness condition for the dater of SLDIs: for all $k,h\in\dint{1,K}$, $h>k$, such that $w_k = w_h$, $x(h) \succeq x(k)$.
The implication is that events occurring during a later execution of a mode cannot occur before events that took place during an earlier execution of that mode.
Note that this does not imply that $x(k)$ is non-decreasing over $k$\zor{; this stronger condition would indeed unnecessarily limit the modeling expressiveness of SLDIs, as illustrated by the dater trajectory in Example~\ref{ex:starving_philosophers} at page~\pageref{ex:starving_philosophers}}.

For convenience, given a finite sequence of modes $v=v_1v_2\ldots v_V\in\Sigma^*$ of length $V\in\nat$ and a number $K\in\nat$, in the remainder of the paper we will denote by $v^K\in\Sigma^*$ the sequence of length $V \cdot K$ formed by concatenating sequence $v$ with itself $K$ times, i.e.,
\[
	v^1 = v,\quad \forall K\geq 2,\quad v^{K} = vv^{K-1};
\]
congruently, if $K = +\infty$, $v^K\in\Sigma^\omega$ denotes an infinite concatenation of sequence $v$.

%

We now show possible applications of SLDIs with two simple examples.

\begin{example}[Heat treatment line, cont.]\label{ex:heat_treatment_AB}
Consider again the heat treatment line of Example~\ref{ex:heat_treatment_loose}.
Now, suppose that two types of parts can be processed in the system: part $\wA$ and part $\wB$; in this example, a schedule $w\in\Sigma^*\cup\Sigma^\omega$ represents the entrance order of parts in the heat treatment line.
As illustrated in Figure~\ref{fi:furnace_AB}, the two parts require different heating times; pieces of type $\wA$ must be heated in the furnace for a time between 2 and 3 time units (as in Example~\ref{ex:heat_treatment_loose}), whereas pieces of type $\wB$ between 3 and 4 time units.
Moreover, the processing rate requirement changes depending on the type $w_k\in\{\wA,\wB\}$ of the $k$\textsuperscript{th} part entering the plant: the $(k+1)$\textsuperscript{st} part must be unloaded from the autonomous guided vehicle at most 4 time units after the $k$\textsuperscript{th} one if $w_k=\wA$, and at most 5 time units later if $w_k=\wB$.
\zor{As in Example~\ref{ex:heat_treatment_loose}, we consider loose initial conditions, i.e., we suppose that the processing rate requirement needs to hold for all pieces after the first one.}

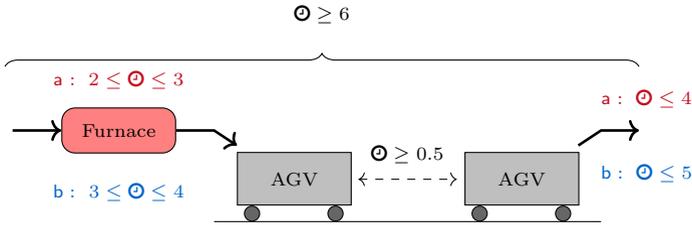
\begin{figure}
	\centering
	\begin{tikzpicture}
\footnotesize 
\node (O) at (0,0) {};
\node[draw,rectangle,minimum height=.6cm,minimum width=1.5cm,fill=red!50!white,rounded corners=5pt,align=center] (F) at (1.5,0) {Furnace};
\draw[line width = 1pt,->] (O) -- (F.west);

\node at ($(F.north) + (0,1.3em)$) {\textcolor{myred}{$\wA:\ 2 \leq\mbox{\faClockO}\leq 3$}};
\node at ($(F.north) + (0,-4em)$) {\textcolor{myblue}{$\wB:\ 3 \leq\mbox{\faClockO}\leq 4$}};

\newcommand{\carH}{.7}
\newcommand{\carL}{1.5}
\newcommand{\carW}{.2}

\coordinate (A0) at ($(F.east) + (.8,-1)$);

\draw[line width = 1pt,->] (F.east) -- ($(A0) + (-.3,1)$) -- ($(A0) + (0,.8)$);

\node[draw, rectangle, minimum height=\carH cm,minimum width=\carL cm,anchor=south west,align=center,fill=gray!50!white] (C1) at (A0) {AGV};
\node[draw, circle, minimum size = \carW cm, inner sep = 0, anchor = north,fill=black!60!white] (W1) at ($(A0) + (.2,0)$) {};
\node[draw, circle, minimum size = \carW cm, inner sep = 0, anchor = north,fill=black!60!white] (W2) at ($(A0) + (\carL-.2,0)$) {};

\coordinate (A1) at ($(A0) + (\carL+1.5,0)$);

\node[draw, rectangle, minimum height=\carH cm,minimum width=\carL cm,anchor=south west,align=center,fill=gray!50!white] (C2) at (A1) {AGV};
\node[draw, circle, minimum size = \carW cm, inner sep = 0, anchor = north,fill=black!60!white] (W3) at ($(A1) + (.2,0)$) {};
\node[draw, circle, minimum size = \carW cm, inner sep = 0, anchor = north,fill=black!60!white] (W4) at ($(A1) + (\carL-.2,0)$) {};

\draw[<->,dashed] ($(A0) + (\carL+.1,\carH/2)$) -- ($(A1) + (-.1,\carH/2)$)node[midway,yshift=1.2em]{$\mbox{\faClockO}\geq 0.5$};

\coordinate (end) at ($(A1) + (\carL+.8,1)$);
\draw[line width = 1pt,->] ($(A1) + (\carL,.8)$) -- ($(A1) + (\carL+.3,1)$) -> (end);

\draw ($(W1.south) - (.5,0)$) -- ($(W4.south) + (.5,0)$);

\draw [decorate,decoration={brace,amplitude=5pt,raise=8ex}]
  ($(O)+(0,0)$) -- ($(end)+(0,0)$) node[midway,yshift=5.5em]{$\mbox{\faClockO}\geq 6$};

\node at ($(end) + (0.1,1.5em)$) {\textcolor{myred}{$\wA:\ \mbox{\faClockO}\leq 4$}};
\node at ($(end) + (0.1,-2em)$) {\textcolor{myblue}{$\wB:\ \mbox{\faClockO}\leq 5$}};
\end{tikzpicture}
	\caption{Illustration of the multi-product heat treatment line of Example~\ref{ex:heat_treatment_AB}.}\label{fi:furnace_AB}
\end{figure}
\begin{figure}
	\centering
	\begin{tikzpicture}[node distance=.5cm and 1cm,>=stealth',bend angle=45,thick]
\tikzstyle{place}=[circle,thick,draw=black,minimum size=6mm]
\tikzstyle{transitionV}=[rectangle,thick,fill=black,minimum height=8mm,inner xsep=1pt]

\footnotesize

\node [transitionV,label=below:$t_1$] (x1) {};
\node [place,tokens=0,label=below:{$[3,4]$}] (p21) [right= of x1] {};
\node [place,tokens=1,label=above:{$[0,+\infty]$}] (p12) [above= of p21] {};
\node [transitionV,label=below:$t_2$] (x2) [right=of p21] {};
\node (p22) [below= of x2] {};
\node [place,tokens=0,label=below:{$[0.5,+\infty]$}] (p32) [right= of x2] {};
\node [place,tokens=1,label=above:{$[0.5,+\infty]$}] (p23) [above= of p32] {};
\node [transitionV,label=below:$t_3$] (x3) [right=of p32] {};
\node [place,tokens=1,label=below:{$[0,5]$}] (p33) [right= of x3] {};
\node [place,tokens=0,label=below:{$[6,+\infty]$}] (p31) [below= of p22] {};

\draw (x1) edge[->] (p21);
\draw (p21) edge[->] (x2);
\draw (x2) edge[->] (p32);
\draw (p32) edge[->] (x3);
\draw (x1) edge[bend right=30,->] (p31);
\draw (p31) edge[bend right=30,->] (x3);
\draw (x2) edge[bend right=30,->] (p12);
\draw (p12) edge[bend right=30,->] (x1);
\draw (x3) edge[bend right=30,->] (p23);
\draw (p23) edge[bend right=30,->] (x2);
\draw (x3) edge[bend right=30,->] (p33);
\draw (p33) edge[bend right=30,->] (x3);
\end{tikzpicture}
	\caption{P-TEG representing the heat treatment line if only parts of type $\wB$ were to be processed. \zor{The case where only parts of type $\wA$ are to be processed is shown in Figure~\ref{fi:TEG_furnace}.}}\label{fi:TEG_furnace_B}
\end{figure}
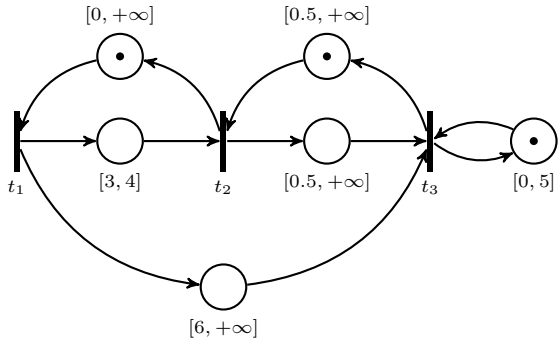

This system can be modeled by SLDIs $\pazocal{S} = (\Sigma,A^0,A^1,B^0,B^1)$, where $\Sigma = \{\wA,\wB\}$,
\[
	A^0_\wA = \begin{bmatrix}
		-\infty & -\infty & -\infty \\
		2 & -\infty & -\infty \\
		6 & 0.5 & -\infty
	\end{bmatrix},\quad
	A^0_\wB = \begin{bmatrix}
		-\infty & -\infty & -\infty \\
		3 & -\infty & -\infty \\
		6 & 0.5 & -\infty
	\end{bmatrix},\quad
	A^1_\wA = A^1_\wB =  \begin{bmatrix}
		-\infty & 0 & -\infty \\
		-\infty & -\infty & 0.5 \\
		-\infty & -\infty & 0
	\end{bmatrix},
\]
\[
	B^0_\wA = \begin{bmatrix}
		+\infty & +\infty & +\infty \\
		3 & +\infty & +\infty \\
		+\infty & +\infty & +\infty
	\end{bmatrix},\quad
	B^0_\wB = \begin{bmatrix}
		+\infty & +\infty & +\infty \\
		4 & +\infty & +\infty \\
		+\infty & +\infty & +\infty
	\end{bmatrix},
\]
\[
	B^1_\wA = \begin{bmatrix}
		+\infty & +\infty & +\infty \\
		+\infty & +\infty & +\infty \\
		+\infty & +\infty & 4
	\end{bmatrix},\quad
	B^1_\wB = \begin{bmatrix}
		+\infty & +\infty & +\infty \\
		+\infty & +\infty & +\infty \\
		+\infty & +\infty & 5
	\end{bmatrix}.
\]
Clearly, if $w = \wA\wA\wA\ldots = \wA^K$ is a concatenation of only mode $\wA$, the dynamics of the system can be described by the P-TEG of Figure~\ref{fi:TEG_furnace}; similarly, if $w = \wB^K$, then the SLDIs are equivalent to the dynamics of the P-TEG of Figure~\ref{fi:TEG_furnace_B}.

In this simple example, the following trajectory satisfies all the time-window constraints, for any schedule $w\in\{\wA,\wB\}^*\cup\{\wA,\wB\}^\omega$:
\[
x(1) = 
\begin{bmatrix}
	0\\3\\6
\end{bmatrix},\quad \forall k\geq 1,\ x(k+1) = \begin{dcases}
	3.5 x(k) & \mbox{if } w_k = \wA,\\
	4 x(k) & \mbox{if } w_k = \wB.
\end{dcases}
\]
The example shows that SLDIs are capable of modeling flow shops with time-window constraints, i.e., manufacturing systems where different jobs (in this case, parts) are processed in each machine (the furnace and the autonomous guided vehicle) of the system in the same order.
\end{example}

\begin{example}[Starving philosophers problem]\label{ex:starving_philosophers}
We present a variant of the famous dining philosophers problem, which we call the \textit{starving philosophers problem}.
There are $p\in\nat$ philosophers sitting at a table eating spaghetti; on the table, there are $p$ chopsticks, and each philosopher $i\in\dint{1,p-1}$ needs both the $i$\textsuperscript{th} and the $(i+1)$\textsuperscript{st} chopstick to start eating, whereas the $p$\textsuperscript{th} philosopher needs the $p$\textsuperscript{th} and the $1$\textsuperscript{st} chopstick.
The $i$\textsuperscript{th} philosopher takes $c_{ij}$ time units to grab the $j$\textsuperscript{th} chopstick and $e_i$ time units to eat; philosophers are allowed to grab two chopsticks at the same time and are not forced to stop eating after $e_i$ time units.
After eating, a philosopher instantaneously puts the chopsticks on the table, so that they can be used by other philosophers.
In our version of the problem, dining can take a "macabre" turn: if the $i$\textsuperscript{th} philosopher does not eat for more than $s_i$ time units, s/he will starve.
The objective of the problem is to find a dining order such that no philosopher starves.

The problem is a metaphor for an issue encountered in concurrent programming, namely, resource starvation.
Consider $p$ critical processes (the philosophers) that need to access some shared resources (the chopsticks); for safety reasons, it might be desirable to prevent some processes from not receiving the requested resources for too long.
Thus, a scheduling plan should guarantee safe operation by granting each process the permission to access the resources at the right time.

Here we suppose that there are $p = 4$ philosophers at the table, and that the following periodic dining order is imposed: initially, \zor{the second and the fourth philosophers eat (they can do so concurrently, as they do not need to share chopsticks); after that, the first philosopher eats once while, in the meantime, the third philosopher eats twice in a row; finally, the eating order repeats from the beginning.}
The order in which philosophers eat before repeating the periodic sequence is referred to as the dining cycle.

Since the considered dining order is periodic, it is possible to describe all trajectories that are valid for the system by means of a P-TEG (with time tags, if we suppose that the $i$\textsuperscript{th} philosopher should start eating for the first time before time $t_0+s_i$ for all $i\in\dint{1,p}$).
The P-TEG for this example is shown in Figure~\ref{fi:P-TEG_philosophers}, where a firing of transitions $t_{s,i}$ and $t_{f,i}$ represents that the $i$\textsuperscript{th} philosopher has started and finished eating for the first time in a dining cycle, and a firing of transitions $t_{s,3}'$ and $t_{f,3}'$ indicate that the \zor{third} philosopher has started and finished eating for the second time in a dining cycle, respectively.
Note that the dimension of the dater function for this problem increases not only with the number of philosophers, but also with the amount of times a philosopher eats in a dining cycle; furthermore, observe that P-TEGs can represent infinite eating orders only if they are periodic.

The same system can be represented more compactly by SLDIs.
Let $\Sigma = \{\winit,\wP_1,\wP_2,\wP_3,\wP_4\}$ be the alphabet associated with the SLDIs, where $\winit$ is the auxiliary initial mode, which will be used to impose strict initial conditions on the system (strict initial conditions are analyzed in more depth in Section~\ref{su:SLDIS_PTEGs}), and $\wP_i$ corresponds to the $i$\textsuperscript{th} philosopher. 
A meaningful schedule for this system is any sequence $w\in\Sigma^{K}$ such that $w_1 = \winit$ and, for all $k\in\dint{2,K}$, $w_k = \wP_{i_k}$ for some $i_k\in\dint{1,p}$.
Whereas the first mode $w_1 = \winit$ does not have physical interpretation \zor{besides mathematically characterizing} the initial conditions for the system, for all $k\in\dint{2,K}$, $w_k$ describes the eating order of philosophers.
The interpretation of schedule $w$ is as follows; consider $k\in\dint{2,K}$, $w_k = \wP_i$ and $w_{k+1} = \wP_j$:
\begin{itemize}
	\item if the $i$\textsuperscript{th} and $j$\textsuperscript{th} philosophers require access to the same chopstick, then philosopher $i$ will eat once before philosopher $j$;
	\item otherwise, the $i$\textsuperscript{th} and $j$\textsuperscript{th} philosophers will eat independently of each other, i.e., philosopher $i$ will start (and finish) eating either before or after or at the same time as philosopher $j$.
	In this case schedules $u w_k w_{k+1} v$ and $u w_{k+1} w_k v$ are representative of the same behavior of the system\footnote{In concurrent systems theory, one would say that $u w_k w_{k+1} v$ and $u w_{k+1} w_k v$ represent the same \textit{trace}~\cite{mazurkiewicz1987trace}.}, for $u\in\Sigma^*$ and $v\in\Sigma^*\cup\Sigma^\omega$.
\end{itemize}
A possible schedule corresponding to the chosen dining order is then given by
\[
	w = \winit \zor{(\wP_2\wP_4\wP_1\wP_3\wP_3)^{\frac{K-1}{5}}}.
\]
We will design $A^0,A^1,B^0,B^1$ such that any dater function $x(k)\in\R^{p+1} = \R^5$ satisfying~\eqref{eq:dynamics} assumes the following interpretation, from which the evolution of the system can be obtained: for all $k\in\dint{2,K}$, if $w_k = \wP_i$, then $x_i(k)$ and $x_5(k)$ represent, respectively, the time at which the $i$\textsuperscript{th} philosopher starts and finishes eating; therefore, assuming $w_{k+1} = \wP_j$, if both the $i$\textsuperscript{th} and the $j$\textsuperscript{th} philosophers require access to the $h$\textsuperscript{th} chopstick, then $x_i(k)+e_i+c_{jh} \leq x_j(k+1)$ (sequential behavior), otherwise, if they do not need to share chopsticks, $x_i(k)$ could also be greater than $x_j(k+1)$ (concurrent behavior). 
For all philosophers $i\in\dint{1,4}$ such that $w_k\neq \wP_i$, $x_i(k)$ is an auxiliary variable that stores the time at which the $i$\textsuperscript{th} philosopher will eat next.
For all $i\in\dint{1,5}$, element $x_i(1)$ will be assigned to the initial time $t_0$, in order to manage the initial conditions (for more details, see Section~\ref{su:SLDIS_PTEGs}).

\newlength{\yuckkyhack}
\settowidth{\yuckkyhack}{$\displaystyle x_4(k+1)$}

For example, consider a value of $k\in\dint{2,K-1}$ such that $w_k = \wP_1$; with the above interpretation, in order to represent the dynamics of the system, the dater function must satisfy
\begin{subequations}
	\begin{alignat}{4}
	x_1(k) + e_1 \leq &\ & \makebox[\yuckkyhack]{$x_5(k)$}  &\ \label{eq:subeq1}\\
	x_5(k) + \max(c_{11},c_{12}) \leq & & x_1(k+1)  & \leq x_5(k) + s_1\label{eq:subeq2}\\
	x_5(k) + c_{22}  \leq  &\ &x_2(k+1)  &\ \label{eq:subeq3}\\
	x_5(k) + c_{41}  \leq &\ & x_4(k+1)  &\ \label{eq:subeq4}\\
	x_2(k)  \leq  &\ &x_2(k+1)  & \leq x_2(k)\label{eq:subeq6}\\
	x_3(k)  \leq  &\ &x_3(k+1)  & \leq x_3(k)\label{eq:subeq7}\\
	x_4(k)  \leq  &\ &x_4(k+1)  & \leq x_4(k),\label{eq:subeq8}
	\end{alignat}
\end{subequations}
\zor{where~\eqref{eq:subeq1} imposes the time between starting and finishing eating for the 1\textsuperscript{st} philosopher, \eqref{eq:subeq2} is used to force the 1\textsuperscript{st} philosopher to start eating again only after s/he has grabbed once more the 1\textsuperscript{st} and 2\textsuperscript{nd} chopsticks but before starving, \eqref{eq:subeq3} and \eqref{eq:subeq4} impose the 2\textsuperscript{nd} and 4\textsuperscript{th} philosophers to start eating only after grabbing the chopsticks left by the 1\textsuperscript{st} philosopher, 
and (\ref{eq:subeq6}--\ref{eq:subeq8}) are auxiliary constraints to impose that $x_i(k+1) = x_i(k)$ for all philosophers $i\in\dint{2,4}$.}
Finally, for the initial condition, we want to impose that
\[
	\begin{array}{rcl}
	\max(c_{11},c_{12}) + t_0 \leq &x_1(2) & \leq s_1 + t_0\\
	\max(c_{22},c_{23}) + t_0 \leq &x_2(2) & \leq s_2 + t_0\\
	\max(c_{33},c_{34}) + t_0 \leq &x_3(2) & \leq s_3 + t_0\\
	\max(c_{44},c_{41}) + t_0 \leq &x_4(2) & \leq s_4 + t_0\\
	\end{array}
\]
to make sure that philosophers start eating for the first time after grabbing the chopsticks and before starving.

In order to get the above interpretation for the dater function, the matrices for the initial mode $\winit$ can be defined as 
\[
	(A^0_\winit)_{ij} = (B^0_\winit)_{ij} = 0\ \forall i,j, 
\]
\[
	A^1_\winit = 
	\begin{bsmallmatrix}
		-\infty & c_{12} & -\infty & c_{11} & -\infty\\
		c_{22} & -\infty & c_{23} & -\infty & -\infty\\
		-\infty & c_{33} & -\infty & c_{34} & -\infty\\
		c_{41} & -\infty & c_{44} & -\infty & -\infty\\
		-\infty & -\infty & -\infty & -\infty & -\infty
	\end{bsmallmatrix},\quad 
	B^1_\winit = 
	\begin{bsmallmatrix}
		+\infty & +\infty & +\infty & +\infty & s_1 \\
		+\infty & +\infty & +\infty & +\infty & s_2 \\
		+\infty & +\infty & +\infty & +\infty & s_3 \\
		+\infty & +\infty & +\infty & +\infty & s_4 \\
		+\infty & +\infty & +\infty & +\infty & +\infty
	\end{bsmallmatrix}.
\]
For mode $\wP_1$ we define
\[
	A^0_{\wP_1} = 
	\begin{bsmallmatrix}
		-\infty & -\infty & -\infty & -\infty & -\infty \\
		-\infty & -\infty & -\infty & -\infty & -\infty \\
		-\infty & -\infty & -\infty & -\infty & -\infty \\
		-\infty & -\infty & -\infty & -\infty & -\infty \\
		e_1 & -\infty & -\infty & -\infty & -\infty
	\end{bsmallmatrix},\quad
	A^1_{\wP_1} = 
	\begin{bsmallmatrix}
		0 & -\infty & -\infty & -\infty & \max(c_{11},c_{12}) \\
		-\infty & 0 & -\infty & -\infty & c_{22} \\
		-\infty & -\infty & 0 & -\infty & -\infty \\
		-\infty & -\infty & -\infty & 0 & c_{41} \\
		-\infty & -\infty & -\infty & -\infty & -\infty
	\end{bsmallmatrix}
\]
\[
	B^0_{\wP_1} = \pazocal{T}, \quad
	B^1_{\wP_1} = 
	\begin{bsmallmatrix}
		+\infty & +\infty & +\infty & +\infty & s_1 \\ 
		+\infty & 0 & +\infty & +\infty & +\infty \\ 
		+\infty & +\infty & 0 & +\infty & +\infty \\ 
		+\infty & +\infty & +\infty & 0 & +\infty \\ 
		+\infty & +\infty & +\infty & +\infty & +\infty
	\end{bsmallmatrix};
\]
for the sake of brevity, we leave it to the reader to derive the matrices for modes $\wP_2,\wP_3,\wP_4$.

\begin{figure}
	\centering
	\resizebox{1\textwidth}{!}{
	\begin{tikzpicture}[node distance=2.6cm and 1.5cm,>=stealth',bend angle=45,thick,on grid]
\tikzstyle{place}=[circle,thick,draw=black,minimum size=5mm]
\tikzstyle{transitionV}=[rectangle,thick,fill=black,minimum height=8mm,inner xsep=1pt]

\small

\node[transitionV,label=above:{$t_{s,4}$}] (ts4) {};
\node[place,right=of ts4,label=below:{$[e_4,\infty]$}] (pe4) {};
\node[place,above=1cm of pe4,label=below:{$[0,s_4]$},tokens=1,myblue] (ps4) {};
\node[transitionV,right=of pe4,label=below:{$t_{f,4}$}] (tf4) {};

\node[place,myred,tokens=0,right=of tf4,label=below:{$[c_{34},\infty]$}] (pc34) {};
\node[place,myred,tokens=0,above=1cm of pc34,label=below:{$[c_{11},\infty]$}] (pc11) {};

\node[transitionV,right=of pc34,label=below:{$t_{s,3}$}] (ts3) {};
\node[place,right=of ts3,label=below:{$[e_3,\infty]$}] (pe3) {};
\node[transitionV,right=of pe3,label=above:{$t_{f,3}$}] (tf3) {};
\node[place,myred,tokens=0,right=of tf3,label=below:{$[\max(c_{33},c_{34}),s_3]$}] (pc3) {};
\node[circle,minimum size=4.4mm,draw,myblue] at (pc3) {};
\node[place,above=1cm of pc3,label=below:{$[0,s_3]$},tokens=1,myblue] (ps3) {};
\node[transitionV,right=of pc3,label=above:{$t_{s,3}'$}] (ts3p) {};
\node[place,right=of ts3p,label=below:{$[e_3,\infty]$}] (pe3p) {};
\node[transitionV,right=of pe3p,label=above:{$t_{f,3}'$}] (tf3p) {};

\node[transitionV,above=of ts4,label=below:{$t_{s,2}$}] (ts2) {};
\node[place,right=of ts2,label=above:{$[e_2,\infty]$}] (pe2) {};
\node[place,below=1cm of pe2,label=above:{$[0,s_2]$},tokens=1,myblue] (ps2) {};
\node[transitionV,right=of pe2,label=above:{$t_{f,2}$}] (tf2) {};

\node[place,myred,tokens=0,right=of tf2,label=above:{$[c_{12},\infty]$}] (pc12) {};
\node[place,myred,tokens=0,below=1cm of pc12,label=above:{$[c_{33},\infty]$}] (pc33) {};

\node[transitionV,label=above:{$t_{s,1}$},above=of ts3] (ts1) {};
\node[place,above=of pc3,label=above:{$[e_1,\infty]$}] (pe1) {};
\node[place,below=1cm of pe1,label=above:{$[0,s_1]$},tokens=1,myblue] (ps1) {};
\node[transitionV,above=of tf3p,label=below:{$t_{f,1}$}] (tf1) {};

\node[place,above=3.9cm of pe3,tokens=1,myred,label=below:{$[c_{22},\infty]$}] (pc22) {};
\node[place,above=1cm of pc22,tokens=1,myred,label=below:{$[c_{41},\infty]$}] (pc41) {};
\node[place,below=1.3cm of pe3,tokens=1,myred,label=below:{$[c_{44},\infty]$}] (pc44) {};
\node[place,below=1cm of pc44,tokens=1,myred,label=below:{$[c_{23},\infty]$}] (pc23) {};

\useasboundingbox ($(ts4)+(-1,-3)$) rectangle ($(tf1)+(1,3)$);

\draw (ts4) edge[-stealth'] (pe4);
\draw (pe4) edge[-stealth'] (tf4);
\draw (tf4) edge[-stealth'] (pc34);
\draw (pc34) edge[-stealth'] (ts3);
\draw (ts3) edge[-stealth'] (pe3);
\draw (pe3) edge[-stealth'] (tf3);
\draw (tf3) edge[-stealth'] (pc3);
\draw (pc3) edge[-stealth'] (ts3p);
\draw (ts3p) edge[-stealth'] (pe3p);
\draw (pe3p) edge[-stealth'] (tf3p);

\draw (tf4.90+10) edge[-stealth',bend right=10] (ps4);
\draw (ps4) edge[-stealth',bend right=10] (ts4.90-10);
\draw (tf3p.90+10) edge[-stealth',bend right=10] (ps3);
\draw (ps3) edge[-stealth',bend right=10] (ts3.90-10);
\draw (tf4.90-10) edge[-stealth',bend left=10] (pc11);
\draw (pc11) edge[-stealth',bend right=10] (ts1.-90-5);

\draw (ts2) edge[-stealth'] (pe2);
\draw (pe2) edge[-stealth'] (tf2);
\draw (tf2) edge[-stealth'] (pc12);
\draw (pc12) edge[-stealth'] (ts1);
\draw (ts1) edge[-stealth'] (pe1);
\draw (pe1) edge[-stealth'] (tf1);

\draw (tf2.-90-10) edge[-stealth',bend left=10] (ps2);
\draw (ps2) edge[-stealth',bend left=10] (ts2.-90+10);
\draw (tf2.-90+10) edge[-stealth',bend right=10] (pc33);
\draw (pc33) edge[-stealth',bend left=10] (ts3.90+5);
\draw (tf1.-90-10) edge[-stealth',bend left=10] (ps1);
\draw (ps1) edge[-stealth',bend left=10] (ts1.-90+10);

\draw [-stealth'] (tf1.90-15) .. controls ($(tf1.90-15)+(.7,0)$) and ($(tf1.90-15)+(.7,1)$) .. (pc22);
\draw [-stealth'] (pc22) .. controls ($(ts2.90+15)+(-.7,1)$) and ($(ts2.90+15)+(-.7,0)$) .. (ts2.90+15);

\draw [-stealth'] (tf1.-90+15) .. controls ($(tf1.-90+15)+(1.3,0)$) and ($(tf1.-90+15)+(1.3,2.4)$) .. (pc41);
\draw [-stealth'] (pc41.180) .. controls ($(ts4.90+15)+(-3.3,4.8)$) and ($(ts4.90+15)+(-1.3,0)$) .. (ts4.90+15);

\draw [-stealth'] (tf3p.-90+15) .. controls ($(tf3p.-90+15)+(.7,0)$) and ($(tf3p.-90+15)+(.7,-1)$) .. (pc44);
\draw [-stealth'] (pc44) .. controls ($(ts4.-90-15)+(-.7,-1)$) and ($(ts4.-90-15)+(-.7,0)$) .. (ts4.-90-15);

\draw [-stealth'] (tf3p.90-15) .. controls ($(tf3p.90-15)+(1.3,0)$) and ($(tf3p.90-15)+(1.3,-2.4)$) .. (pc23);
\draw [-stealth'] (pc23.180) .. controls ($(ts2.-90-15)+(-3.3,-4.8)$) and ($(ts2.-90-15)+(-1.3,0)$) .. (ts2.-90-15);

\end{tikzpicture}}
	\caption{P-TEG for the starving philosophers problem.
	Tokens inside a place colored \textbf{black}, \textcolor{myblue}{\textbf{blue}}, and \textcolor{myred}{\textbf{red}} represent, respectively, a philosopher eating, a philosopher waiting to eat, and a chopstick being grabbed.
	The time tag associated to each place with initial token is $0$.
	\zor{Note that a token in the place with upstream transition $t_{f,3}$ and downstream transition $t_{s,3}'$ indicates that the third philosopher is both grabbing the third and the fourth chopstick and waiting to eat, hence the double color of this place.}}
	\label{fi:P-TEG_philosophers}
\end{figure}
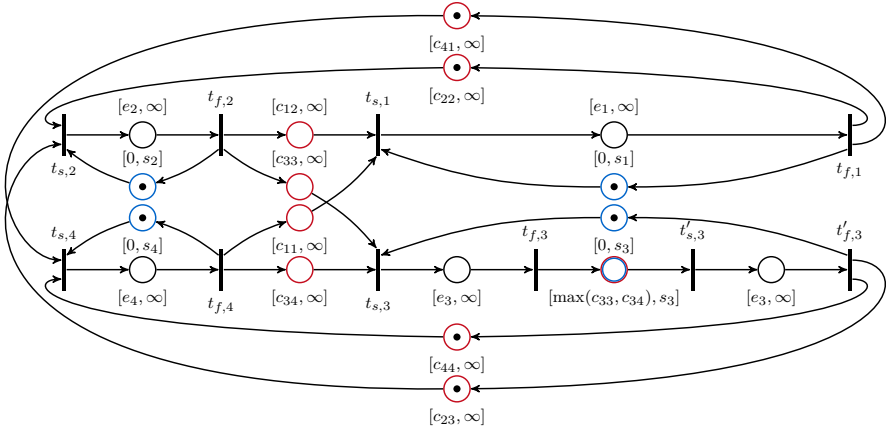

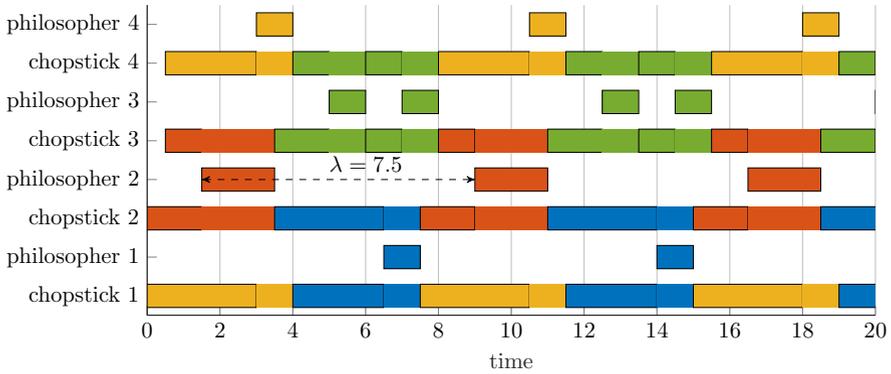
\begin{figure}
	\centering
	\resizebox{1\textwidth}{!}{
%
%
\definecolor{mycolor1}{rgb}{0.85000,0.32500,0.09800}%
\definecolor{mycolor2}{rgb}{0.92900,0.69400,0.12500}%
\definecolor{mycolor3}{rgb}{0.00000,0.44700,0.74100}%
\definecolor{mycolor4}{rgb}{0.46600,0.67400,0.18800}%
\begin{tikzpicture}

\begin{axis}[%
width=4.44in,
height=1.912in,
at={(0.839in,0.409in)},
scale only axis,
xmin=0,
xmax=20,
xlabel style={font=\color{white!15!black}},
xlabel={time},
ymin=0,
ymax=8,
ytick={0.5,1.5,2.5,3.5,4.5,5.5,6.5,7.5,8.5},
yticklabels={{chopstick 1},{philosopher 1},{chopstick 2},{philosopher 2},{chopstick 3},{philosopher 3},{chopstick 4},{philosopher 4},{}},
axis background/.style={fill=white},
axis x line*=bottom,
axis y line*=left,
legend style={legend cell align=left, align=left, draw=white!15!black},
xmajorgrids=true
]
\draw[fill=mycolor1, draw=black] (axis cs:0,2.2) rectangle (axis cs:1.5,2.8);
\draw[opacity=.3,draw=none, fill=mycolor1] (axis cs:1.5,2.2) rectangle (axis cs:3.5,2.8);
\draw[fill=mycolor1, draw=black] (axis cs:1.5,3.2) rectangle (axis cs:3.5,3.8);
\draw[fill=mycolor1, draw=black] (axis cs:0.5,4.2) rectangle (axis cs:1.5,4.8);
\draw[opacity=.3,draw=none, fill=mycolor1] (axis cs:1.5,4.2) rectangle (axis cs:3.5,4.8);
\draw[fill=mycolor2, draw=black] (axis cs:0.5,6.2) rectangle (axis cs:3,6.8);
\draw[opacity=.3,draw=none, fill=mycolor2] (axis cs:3,6.2) rectangle (axis cs:4,6.8);
\draw[fill=mycolor2, draw=black] (axis cs:3,7.2) rectangle (axis cs:4,7.8);
\draw[fill=mycolor2, draw=black] (axis cs:0,0.2) rectangle (axis cs:3,0.8);
\draw[opacity=.3,draw=none, fill=mycolor2] (axis cs:3,0.2) rectangle (axis cs:4,0.8);
\draw[fill=mycolor3, draw=black] (axis cs:4,0.2) rectangle (axis cs:6.5,0.8);
\draw[opacity=.3,draw=none, fill=mycolor3] (axis cs:6.5,0.2) rectangle (axis cs:7.5,0.8);
\draw[fill=mycolor3, draw=black] (axis cs:6.5,1.2) rectangle (axis cs:7.5,1.8);
\draw[fill=mycolor3, draw=black] (axis cs:3.5,2.2) rectangle (axis cs:6.5,2.8);
\draw[opacity=.3,draw=none, fill=mycolor3] (axis cs:6.5,2.2) rectangle (axis cs:7.5,2.8);
\draw[fill=mycolor4, draw=black] (axis cs:3.5,4.2) rectangle (axis cs:5,4.8);
\draw[opacity=.3,draw=none, fill=mycolor4] (axis cs:5,4.2) rectangle (axis cs:6,4.8);
\draw[fill=mycolor4, draw=black] (axis cs:5,5.2) rectangle (axis cs:6,5.8);
\draw[fill=mycolor4, draw=black] (axis cs:4,6.2) rectangle (axis cs:5,6.8);
\draw[opacity=.3,draw=none, fill=mycolor4] (axis cs:5,6.2) rectangle (axis cs:6,6.8);
\draw[fill=mycolor4, draw=black] (axis cs:6,4.2) rectangle (axis cs:7,4.8);
\draw[opacity=.3,draw=none, fill=mycolor4] (axis cs:7,4.2) rectangle (axis cs:8,4.8);
\draw[fill=mycolor4, draw=black] (axis cs:7,5.2) rectangle (axis cs:8,5.8);
\draw[fill=mycolor4, draw=black] (axis cs:6,6.2) rectangle (axis cs:7,6.8);
\draw[opacity=.3,draw=none, fill=mycolor4] (axis cs:7,6.2) rectangle (axis cs:8,6.8);
\draw[fill=mycolor1, draw=black] (axis cs:7.5,2.2) rectangle (axis cs:9,2.8);
\draw[opacity=.3,draw=none, fill=mycolor1] (axis cs:9,2.2) rectangle (axis cs:11,2.8);
\draw[fill=mycolor1, draw=black] (axis cs:9,3.2) rectangle (axis cs:11,3.8);
\draw[fill=mycolor1, draw=black] (axis cs:8,4.2) rectangle (axis cs:9,4.8);
\draw[opacity=.3,draw=none, fill=mycolor1] (axis cs:9,4.2) rectangle (axis cs:11,4.8);
\draw[fill=mycolor2, draw=black] (axis cs:8,6.2) rectangle (axis cs:10.5,6.8);
\draw[opacity=.3,draw=none, fill=mycolor2] (axis cs:10.5,6.2) rectangle (axis cs:11.5,6.8);
\draw[fill=mycolor2, draw=black] (axis cs:10.5,7.2) rectangle (axis cs:11.5,7.8);
\draw[fill=mycolor2, draw=black] (axis cs:7.5,0.2) rectangle (axis cs:10.5,0.8);
\draw[opacity=.3,draw=none, fill=mycolor2] (axis cs:10.5,0.2) rectangle (axis cs:11.5,0.8);
\draw[fill=mycolor3, draw=black] (axis cs:11.5,0.2) rectangle (axis cs:14,0.8);
\draw[opacity=.3,draw=none, fill=mycolor3] (axis cs:14,0.2) rectangle (axis cs:15,0.8);
\draw[fill=mycolor3, draw=black] (axis cs:14,1.2) rectangle (axis cs:15,1.8);
\draw[fill=mycolor3, draw=black] (axis cs:11,2.2) rectangle (axis cs:14,2.8);
\draw[opacity=.3,draw=none, fill=mycolor3] (axis cs:14,2.2) rectangle (axis cs:15,2.8);
\draw[fill=mycolor4, draw=black] (axis cs:11,4.2) rectangle (axis cs:12.5,4.8);
\draw[opacity=.3,draw=none, fill=mycolor4] (axis cs:12.5,4.2) rectangle (axis cs:13.5,4.8);
\draw[fill=mycolor4, draw=black] (axis cs:12.5,5.2) rectangle (axis cs:13.5,5.8);
\draw[fill=mycolor4, draw=black] (axis cs:11.5,6.2) rectangle (axis cs:12.5,6.8);
\draw[opacity=.3,draw=none, fill=mycolor4] (axis cs:12.5,6.2) rectangle (axis cs:13.5,6.8);
\draw[fill=mycolor4, draw=black] (axis cs:13.5,4.2) rectangle (axis cs:14.5,4.8);
\draw[opacity=.3,draw=none, fill=mycolor4] (axis cs:14.5,4.2) rectangle (axis cs:15.5,4.8);
\draw[fill=mycolor4, draw=black] (axis cs:14.5,5.2) rectangle (axis cs:15.5,5.8);
\draw[fill=mycolor4, draw=black] (axis cs:13.5,6.2) rectangle (axis cs:14.5,6.8);
\draw[opacity=.3,draw=none, fill=mycolor4] (axis cs:14.5,6.2) rectangle (axis cs:15.5,6.8);
\draw[fill=mycolor1, draw=black] (axis cs:15,2.2) rectangle (axis cs:16.5,2.8);
\draw[opacity=.3,draw=none, fill=mycolor1] (axis cs:16.5,2.2) rectangle (axis cs:18.5,2.8);
\draw[fill=mycolor1, draw=black] (axis cs:16.5,3.2) rectangle (axis cs:18.5,3.8);
\draw[fill=mycolor1, draw=black] (axis cs:15.5,4.2) rectangle (axis cs:16.5,4.8);
\draw[opacity=.3,draw=none, fill=mycolor1] (axis cs:16.5,4.2) rectangle (axis cs:18.5,4.8);
\draw[fill=mycolor2, draw=black] (axis cs:15.5,6.2) rectangle (axis cs:18,6.8);
\draw[opacity=.3,draw=none, fill=mycolor2] (axis cs:18,6.2) rectangle (axis cs:19,6.8);
\draw[fill=mycolor2, draw=black] (axis cs:18,7.2) rectangle (axis cs:19,7.8);
\draw[fill=mycolor2, draw=black] (axis cs:15,0.2) rectangle (axis cs:18,0.8);
\draw[opacity=.3,draw=none, fill=mycolor2] (axis cs:18,0.2) rectangle (axis cs:19,0.8);
\draw[fill=mycolor3, draw=black] (axis cs:19,0.2) rectangle (axis cs:21.5,0.8);
\draw[opacity=.3,draw=none, fill=mycolor3] (axis cs:21.5,0.2) rectangle (axis cs:22.5,0.8);
\draw[fill=mycolor3, draw=black] (axis cs:21.5,1.2) rectangle (axis cs:22.5,1.8);
\draw[fill=mycolor3, draw=black] (axis cs:18.5,2.2) rectangle (axis cs:21.5,2.8);
\draw[opacity=.3,draw=none, fill=mycolor3] (axis cs:21.5,2.2) rectangle (axis cs:22.5,2.8);
\draw[fill=mycolor4, draw=black] (axis cs:18.5,4.2) rectangle (axis cs:20,4.8);
\draw[opacity=.3,draw=none, fill=mycolor4] (axis cs:20,4.2) rectangle (axis cs:21,4.8);
\draw[fill=mycolor4, draw=black] (axis cs:20,5.2) rectangle (axis cs:21,5.8);
\draw[fill=mycolor4, draw=black] (axis cs:19,6.2) rectangle (axis cs:20,6.8);
\draw[opacity=.3,draw=none, fill=mycolor4] (axis cs:20,6.2) rectangle (axis cs:21,6.8);
\draw[fill=mycolor4, draw=black] (axis cs:21,4.2) rectangle (axis cs:22,4.8);
\draw[opacity=.3,draw=none, fill=mycolor4] (axis cs:22,4.2) rectangle (axis cs:23,4.8);
\draw[fill=mycolor4, draw=black] (axis cs:22,5.2) rectangle (axis cs:23,5.8);
\draw[fill=mycolor4, draw=black] (axis cs:21,6.2) rectangle (axis cs:22,6.8);
\draw[opacity=.3,draw=none, fill=mycolor4] (axis cs:22,6.2) rectangle (axis cs:23,6.8);

\draw[stealth'-stealth',dashed] (axis cs:1.5,3.5) to node [above,pos=.6] {\normalsize$\lambda = 7.5$} (axis cs:9,3.5);
\end{axis}

\end{tikzpicture}%
}
	\caption{Gantt chart of a possible trajectory for the starving philosophers problem. Opaque bars indicate either that a chopstick is being grabbed or that a philosopher is eating; transparent bars represent chopsticks being used by a philosopher to eat. Different colors correspond to different philosophers. The dashed line indicates the period of the trajectory.}
	\label{fi:gantt_philosophers}
\end{figure}

We consider the following numerical parameters:
\[
	c_{11} = 2,\ c_{12} = 3,\ e_1 = 1,\ s_1 = 10,\ c_{22} = 1,\ c_{23} = 1,\ e_2 = 2,\ s_2 = 10,
\]
\[
	c_{33} = 1,\ c_{34} = 1,\ e_3 = 1,\ s_3 = 15,\ c_{44} = 2,\ c_{41} = 3,\ e_4 = 1,\ s_4 = 12.
\]
The Gantt chart of Figure~\ref{fi:gantt_philosophers} represents a valid trajectory for the P-TEG\footnote{Note that the trajectory is 1-periodic for the P-TEG.} of Figure~\ref{fi:P-TEG_philosophers} and for the SLDIs defined above, supposing that the initial time $t_0$ is equal to $0$.

The first 5 elements of the dater trajectory for the SLDIs are:
\[
	x(1) = \begin{bsmallmatrix}
		0\\0\\0\\0\\0
	\end{bsmallmatrix},\ 
	x(2) = \begin{bsmallmatrix}
		6.5\\1.5\\5\\3\\3.5
	\end{bsmallmatrix},\ 
	x(3) = \begin{bsmallmatrix}
		6.5\\9\\5\\3\\4
	\end{bsmallmatrix},\
	x(4) = \begin{bsmallmatrix}
		6.5\\9\\5\\10.5\\7.5
	\end{bsmallmatrix},\ 
	x(5) = \begin{bsmallmatrix}
		14\\9\\5\\10.5\\6
	\end{bsmallmatrix}.
\]
It is worth noting that, differently from P-TEGs, the dater function of SLDIs does not need to be non-decreasing: for instance, in this example we have
\[
	x_5(4) = 7.5 \geq 6 = x_5(5),
\]
as the time in which the first philosopher ($w_4 = \wP_1$) stops eating for the first time is after the time in which the third philosopher ($w_5 = \wP_3$) stops eating for the first time.

The three main advantages of using SLDIs instead of P-TEGs for this problem are the following:
\begin{enumerate}
	\item higher computational efficiency: the dater function for the SLDIs has smaller dimension compared to the P-TEG.
	As we shall see, this corresponds to lower computational complexity for analyzing trajectories of the system;
	\item lower modeling effort: the P-TEG in Figure~\ref{fi:P-TEG_philosophers} can only represent the dining order specified above; to analyze a different dining order, a new P-TEG needs to be provided.
	On the other hand, for the SLDIs different dining orders simply correspond to different schedules $w$;
	\item larger modeling expressiveness: only SLDIs are able to represent dining orders that are not periodic (with a dater function of finite dimension).
\end{enumerate}
	
\end{example}

When schedule $w$ is fixed, we can extend the definition of some properties of P-TEGs to SLDIs in a natural way.
For instance, if there exists a trajectory of the dater $\{x(k)\}_{k\in\dint{1,K}}$ that satisfies~\eqref{eq:dynamics}, then the trajectory is consistent for the SLDIs under schedule $w$, and we say that $w$ is a consistent schedule for the SLDIs (or that the SLDIs are consistent under schedule $w$).
Note that, different from the simple case of Example~\ref{ex:heat_treatment_AB}, there are SLDIs for which not all schedules admit consistent trajectories.

The definitions of delay-bounded trajectory and bounded consistency are generalized to SLDIs in a similar fashion.
The interpretation of bounded consistency of a schedule $w$ is analogous to the one of P-TEGs; when a process consisting of several tasks (the start and end of which are represented by events) is modeled by SLDIs under a schedule $w$ that is not boundedly consistent, then either the execution of every possible sequence of tasks following $w$ 
will lead to the violation of some time window constraints (if $w$ is not even consistent), or we will certainly observe an infinite accumulation of delay between the start or end of some tasks (if the only consistent trajectories are not delay-bounded).

\subsection{SLDIs and P-TEGs with strict initial conditions}\label{su:SLDIS_PTEGs}

As discussed in Subsection~\ref{su:strict}, the dynamics of P-TEGs with strict initial conditions are not pure LDIs.
In this subsection, we prove that they can be expressed by means of SLDIs under specific types of schedules, with the immediate consequence that any property of SLDIs also holds for P-TEGs with strict initial conditions.

We want to prove that inequalities
\begin{equation}\label{eq:dynamics_PTEGs_strict}
	\begin{array}{rrcl}
		& \underline{\Delta}\otimes t_0\tilde{e} \preceq & x(1) & \preceq \overline{\Delta} \stimes t_0\tilde{e},\\
	\forall k\in\dint{1,K}, &
		A^0\otimes x(k) \preceq & x(k) & \preceq B^0\stimes x(k),\\
	\forall k\in\dint{1,K-1}, &
		\quad\quad A^1\otimes x(k) \preceq & x(k+1) & \preceq B^1\stimes x(k)
	\end{array}
\end{equation}
can be written as SLDIs.
For this aim, let us define an auxiliary variable $x(0)\in\R^n$.
Note that the first inequality of~\eqref{eq:dynamics_PTEGs_strict} is equivalent to 
\[
	\begin{array}{rcl}
		\mathbb{E}\otimes x(0) \preceq & x(0) & \preceq \mathbb{E} \stimes x(0),\\
		\underline{\Delta}\otimes x(0) \preceq & x(1) & \preceq \overline{\Delta} \stimes x(0),
	\end{array}
\]
where $\mathbb{E}_{ij} = 0$ for all $i,j\in\dint{1,n}$.
Indeed, the first of the latter inequalities can be rewritten as
\[
	\forall i\in\dint{1,n},\quad \max(x_1(0),\ldots,x_n(0)) \preceq x_i(0) \preceq \min(x_1(0),\ldots,x_n(0)),
\]
which admits as solution all $x(0)$ that satisfy $x_i(0) = x_j(0)$ for all $i,j\in\dint{1,n}$; therefore, solutions can be parametrized in $t_0\in\R$ as $x(0) = t_0 \tilde{e}$.
Recalling that P-TEGs (as well as SLDIs) are time-invariant systems (see Subsection~\ref{su:strict}), this proves that~\eqref{eq:dynamics_PTEGs_strict} is equivalent to SLDIs $\pazocal{S} = (\{\winit,\wA\},A^0,A^1,B^0,B^1)$ under schedule $w = \winit\wA\wA\wA\wA\ldots = \winit\wA^{K}$ (i.e., mode $\winit$ followed by a sequence of $K$ modes $\wA$), where
\[
	A^0_{\winit} = B^0_{\winit}=\mathbb{E},\quad A^1_{\winit} = \underline{\Delta},\quad B^1_{\winit} = \overline{\Delta},
\] 
\[
	A^0_{\wA} = A^0,\quad A^1_{\wA} = A^1,\quad B^0_{\wA} =B^0, \quad B^1_{\wA} = B^1.
\] 
The following result is an immediate consequence of this fact.

\begin{proposition}\label{pr:SLDI_strict}
The P-TEG characterized by matrices $A^0,A^1,B^0,B^1$ under strict initial conditions determined by matrices $\underline{\Delta}$ and $\overline{\Delta}$ is (boundedly) consistent if and only if schedule $w = \winit\wA^{+\infty}$ is (boundedly) consistent for the SLDIs $\pazocal{S}$ defined as above.
\end{proposition}

Although Proposition~\ref{pr:SLDI_strict} alone does not answer the question about how to verify (bounded) consistency of P-TEGs with strict initial conditions\footnote{The question seems to be related to the topic of transience bounds in max-plus linear systems (see, e.g.,~\cite{nowak2014overview,abate2020computation}).}, it suggests that any result regarding SLDIs under schedules of the form $w = \winit\wA^{K}$ holds automatically for P-TEGs with strict initial conditions.

In the following, we provide an example for the application of this fact.
Let us study the existence of 1-periodic trajectories for P-TEGs with strict initial conditions.
From the above discussion, such trajectories correspond to "ultimately" 1-periodic trajectories of the form
\[
	x(1),x(2)\in\R^n,\quad \forall k\in\dint{1,K-1}\quad x(k+2) = \lambda^kx(2)
\]
for SLDIs $\pazocal{S} = (\{\winit,\wA\},A^0,A^1,B^0,B^1)$ under schedule $w = \winit\wA^K$.
We proceed with a strategy similar to the one seen in Subsection~\ref{su:LDIs}: substituting formula $x(k+2) = \lambda^{k}x(2)$ for all $k\in\dint{1,K-1}$ into~\eqref{eq:dynamics}, we get
\[
	\begin{array}{rrcl}
		& A^0_{\winit}\otimes x(1) \preceq & x(1) & \preceq B^0_\winit \stimes x(1),\\
		& A^1_{\winit}\otimes x(1) \preceq & x(2) & \preceq B^1_\winit \stimes x(1),\\
	\forall k\in\dint{1,K}, &
		A^0_\wA\otimes \lambda^{k-1} x(2) \preceq & \lambda^{k-1} x(2) & \preceq B^0_\wA\stimes \lambda^{k-1} x(2),\\
	\forall k\in\dint{1,K-1}, &
		\quad\quad A^1_\wA\otimes \lambda^{k-1} x(2) \preceq & \lambda^k x(2) & \preceq B^1_\wA\stimes \lambda^{k-1} x(2);
	\end{array}
\] 
multiplying by $\lambda^{-k+1}$ in the third and fourth inequalities, we obtain
\[
	\begin{array}{rcl}
		A^0_{\winit}\otimes x(1) \preceq & x(1) & \preceq B^0_\winit \stimes x(1),\\
		A^1_{\winit}\otimes x(1) \preceq & x(2) & \preceq B^1_\winit \stimes x(1),\\
		A^0_\wA\otimes  x(2) \preceq &  x(2) & \preceq B^0_\wA\stimes  x(2),\\
		\quad\quad A^1_\wA\otimes  x(2) \preceq & \lambda x(2) & \preceq B^1_\wA\stimes x(2).
	\end{array}
\] 
By defining the extended dater vector $\tilde{x} = [x^\top (1)\ \ x^\top(2)]^\top$ and using Proposition~\ref{pr:max-min}, the inequalities can be rewritten in terms of $\tilde{x}$ as
\begin{equation}\label{eq:PIC-NCP_strict}
	(\lambda P \oplus \lambda^{-1} I \oplus C)\otimes \tilde{x} \preceq\tilde{x},
\end{equation}
where
\[
	P = \begin{bmatrix}
		\pazocal{E} & \pazocal{E}\\
		\pazocal{E} & B^{1\sharp}_{\wA}
	\end{bmatrix},\quad 
	I = \begin{bmatrix}
		\pazocal{E} & \pazocal{E}\\
		\pazocal{E} & A^1_{\wA}
	\end{bmatrix},\quad 
	C = \begin{bmatrix}
		A_\winit^0\oplus B_\winit^{0\sharp} & B_\winit^{1\sharp}\\
		A^1_\winit & A_\wA^0\oplus B_\wA^{0\sharp}
	\end{bmatrix}. 
\]
Clearly,~\eqref{eq:PIC-NCP_strict} defines a PIC-NCP, whose solution can be found in time $\pazocal{O}(n^4)$ using Algorithm~\ref{al:PIC-NCP}.
The conclusion is that periods of consistent 1-periodic trajectories for P-TEGs under strict initial conditions can be obtained in the same strongly polynomial time complexity as for P-TEGs under loose initial conditions.

\subsection{Analysis of periodic schedules}\label{su:analysis_of_fixed_schedules}

In this subsection, we analyze bounded consistency and cycle times of SLDIs when schedule $w\in\Sigma^*\cup\Sigma^\omega$ is periodic, i.e., when it can be written as $w = v^K$\zor{, $K\in\nat\cup\{+\infty\}$,} and $v=v_1\cdots v_V\in\Sigma^*$ is a finite \textit{subschedule} of length $V$.
We define $v$-periodic trajectories of period $\lambda\in\R_{\geq 0}$ for SLDIs under schedule $w=v^K$ as those dater trajectories that, for all $k\in\dint{1,K-1}$, $h\in\dint{1,V}$, satisfy 
\[
	x(Vk + h) = \lambda x(V(k-1) + h);
\]
$\solSLDI$ denotes the set of all periods (or cycle times) $\lambda$ for which there exists a consistent $v$-periodic trajectory.
Their relationship with 1-periodic trajectories in P-TEGs is \zor{illustrated in} the following example.

\begin{example}\label{ex:2}
	Let us analyze the SLDIs $\pazocal{S}$, with $\Sigma=\{\wA,\wB,\wC\}$, and $A^0_\wZ,A^1_\wZ,B^0_\wZ,B^1_\wZ$ defined as in Example~\ref{ex:P-TEGs}; now label $\wZ\in\Sigma$ is to be interpreted as a mode. 
	Thus, for each event $k$, the dynamics of the SLDIs may switch among those specified by the P-TEGs labeled $\wA$, $\wB$, and $\wC$.
	We consider periodic schedules $(\wA\wC)^K$ and $(\wA\wB)^K$; observe that for $w=v^K$, with $v\in\{\wA\wC,\wA\wB\}$ (i.e., $v_1=\wA$ and $v_2=\wC$ or $v_2=\wB$), the SLDIs following $w$ can be written as: 
	\begin{equation}\label{eq:example_1}
		\begin{array}{lrcl}
			\forall k\in\dint{1,K},&A^0_{v_1}\otimes x(2(k-1)+1) \preceq & x(2(k-1)+1) & \preceq B^0_{v_1}\stimes x(2(k-1)+1),\\
			\forall k\in\dint{1,K},&A^1_{v_1}\otimes x(2(k-1)+1) \preceq & x(2(k-1)+2) & \preceq B^1_{v_1}\stimes x(2(k-1)+1),\\
			\forall k\in\dint{1,K},&A^0_{v_2}\otimes x(2(k-1)+2) \preceq & x(2(k-1)+2) & \preceq B^0_{v_2}\stimes x(2(k-1)+2),\\
			\forall k\in\dint{1,K-1},&A^1_{v_2}\otimes x(2(k-1)+2) \preceq & x(2k+1) & \preceq B^1_{v_2}\stimes x(2(k-1)+2).
		\end{array}
	\end{equation}
	By defining $\tilde{x}(k) = [x^\top(2(k-1)+1),x^\top(2(k-1)+2)]^\top$, the above set of inequalities can be rewritten as LDIs: 
	\begin{subequations}\label{eq:example_2}
		\begin{empheq}{align}
			\forall k\in\dint{1,K}, \quad 
			&A^0_{v}
	\otimes \tilde{x}(k) \preceq \,\,\,\,\, \tilde{x}(k) \,\,\,\,\, \preceq B^0_v \stimes \tilde{x}(k) \label{eq:example_2-1}\\
			\forall k\in\dint{1,K-1},\quad
			&A^1_v \otimes \tilde{x}(k) \preceq  \tilde{x}(k+1) \preceq B^1_v \stimes \tilde{x}(k)\label{eq:example_2-2}
		\end{empheq}
	\end{subequations}
	where
	\[
		A^0_v = \begin{bmatrix}
			A^0_{v_1} & \pazocal{E} \\
			A^1_{v_1} & A^0_{v_2}
		\end{bmatrix}, \quad
		A^1_v = \begin{bmatrix}
			\pazocal{E} & A^1_{v_2} \\
			\pazocal{E} & \pazocal{E}
		\end{bmatrix},
	\]
	\[
		B^0_v = \begin{bmatrix}
			B^0_{v_1} & \pazocal{T} \\
			B^1_{v_1} & B^0_{v_2}
		\end{bmatrix},\quad 
		B^1_v = \begin{bmatrix}
			\pazocal{T} & B^1_{v_2} \\
			\pazocal{T} & \pazocal{T}
		\end{bmatrix}.
	\]
	To see the equivalence of~\eqref{eq:example_1} and~\eqref{eq:example_2}, observe that the second block of~\eqref{eq:example_2-1} reads
	\[
		\begin{array}{rcl}
		A_{v_1}^1 \otimes x(2(k-1)+1) \oplus \!\! & A_{v_2}^0 \otimes x(2(k-1)+2)  & \\ & \preceq x(2(k-1)+2) \preceq & \\ & B_{v_1}^1 \stimes x(2(k-1)+1) \splus & \!\! B_{v_2}^0 \stimes x(2(k-1)+1).
		\end{array}
	\]
	From this transformation, we can easily conclude that $\pazocal{S}$ is boundedly consistent under $v^{+\infty}$ if and only if the LDIs with characteristic matrices $A^0_v,A^1_v,B^0_v,B^1_v$ are boundedly consistent, and that all consistent $v$-periodic trajectories of $\pazocal{S}$ coincide with consistent 1-periodic trajectories of the LDIs; hence, using Algorithm~\ref{al:PIC-NCP} we obtain
	\[
		\solSLDI[\wA\wC] = 
		\solNCP{\lambda B^{1\sharp}_{\wA\wC}\oplus \lambda^{-1} A^1_{\wA\wC}\oplus A^0_{\wA\wC}\oplus B^{0\sharp}_{\wA\wC}} = \emptyset,
	\]\[
		\solSLDI[\wA\wB] = 
		\solNCP{\lambda B^{1\sharp}_{\wA\wB} \oplus \lambda^{-1}A^1_{\wA\wB} \oplus A^0_{\wA\wB}\oplus B^{0\sharp}_{\wA\wB}} = [3,3].
	\]

	It is worth noting that, although P-TEGs labeled $\wA$ and $\wB$ are not boundedly consistent, the SLDIs under schedule $(\wA\wB)^{+\infty}$ are.
	Thus, in general it is not possible to infer bounded consistency of SLDIs under a fixed schedule $w$ solely based on the analysis of each mode appearing in $w$.
\end{example}

By generalizing the procedure shown in Example~\ref{ex:2}, we can derive the following proposition through some algebraic manipulations (to set up equivalent LDIs) and applying Theorem~\ref{th:bounded_consistency} \zor{(for a formal proof, see Appendix~\ref{ap:Proposition5})}.

\begin{proposition}\label{pr:SLDI-P-TEGs}
SLDIs $\pazocal{S}$ are boundedly consistent under schedule $w=v^{+\infty}$ if and only if they admit a $v$-periodic trajectory.
Moreover, set $\solSLDI$ coincides with $\solPIC[\MP_v][\MI_v][\MC_v]$, where
\[
	\lambda\MP_v \oplus \lambda^{-1}\MI_v \oplus \MC_v = 
	\begin{bmatrix}
		C_1 & P_{1} & \pazocal{E} & \pazocal{E} & \pazocal{E} & \cdots & \pazocal{E} & \lambda^{-1}I_{V}\\ 
		I_{1} & C_{2} & P_{2} & \pazocal{E} & \pazocal{E} & \cdots & \pazocal{E} & \pazocal{E}\\ 
		\pazocal{E} & I_{2} & C_{3} & P_{3} & \pazocal{E} & \cdots & \pazocal{E} & \pazocal{E}\\ 
		\pazocal{E} & \pazocal{E} & I_{3} & C_{4} & P_{4} & \cdots & \pazocal{E} & \pazocal{E}\\ 
		\pazocal{E} & \pazocal{E} & \pazocal{E} & I_{4} & C_{5} & \cdots & \pazocal{E} & \pazocal{E}\\ 
		\vdots & \vdots & \vdots & \vdots & \vdots & \ddots & \vdots & \vdots \\ 
		\pazocal{E} & \pazocal{E} & \pazocal{E} & \pazocal{E} & \pazocal{E} & \cdots & C_{{V-1}} & P_{{V-1}}\\ 
		\lambda P_{V} & \pazocal{E} & \pazocal{E} & \pazocal{E} & \pazocal{E} & \cdots & I_{{V-1}} & C_{V}
	\end{bmatrix} \in\Rmax^{Vn\times Vn},
\]
$P_h = B_{v_h}^{1\sharp}$, $I_h = A_{v_h}^1$, and $C_h = A_{v_h}^0\oplus B_{v_h}^{0\sharp}$ for all $h\in\dint{1,V}$.
\end{proposition}

Proposition~\ref{pr:SLDI-P-TEGs} directly provides an algorithm to compute the minimum and maximum cycle times of SLDIs under a fixed periodic schedule.
Indeed, these values come from solving the NCP for the parametric precedence graph $\graph(\lambda\MP_v\oplus\lambda^{-1}\MI_v\oplus\MC_v)$.
However, this approach results in a slow (although strongly polynomial time) algorithm when the length of subschedule $v$ is large; indeed, its time complexity is $\pazocal{O}((Vn)^4) = \pazocal{O}(V^4n^4)$, as the considered precedence graph has $Vn$ nodes.

To speed up the computation of $\solSLDI$, we may note the following fact: the longer subschedule $v$ is, the larger is the number of $-\infty$'s compared to real numbers in $\lambda\MP_v\oplus\lambda^{-1}\MI_v\oplus\MC_v$.
In other words, the matrix becomes sparser and sparser with larger values of $V$.
Moreover, the real entries of the matrix have a recognizable pattern.
The following theorem, proven in Appendix~\ref{ap:improved_algorithm}, exploits this observation, achieving time complexity $\pazocal{O}(Vn^3+n^4)$ for computing the set $\solSLDI$.
The resulting complexity is thus linear in the length of subschedule $v$.

\begin{theorem}\label{th:improved_algorithm}
Precedence graph $\graph(\lambda \MP_v \oplus \lambda^{-1}\MI_v \oplus \MC_v)$ does not contain circuits with positive weight if and only if the following conditions hold:
\begin{enumerate}
	\item\label{en:1} for all $h\in\dint{1,V}$, $\graph(C_{h})\in\nonegset$;
	\item\label{en:2} for all $h\in\dint{1,V-1}$, $\graph(\mathbb{C}^{P}_h)\in\nonegset$ and $\graph(\mathbb{C}^{I}_{h+1})\in\nonegset$, where
\[\def\arraystretch{1.5}
	\begin{array}{ll}
	\forall h\in\dint{1,V-1}, \quad & \mathbb{C}^{P}_h = \mathbb{P}_h(\mathbb{P}_{h+1}(\cdots(\mathbb{P}_{V-1}\mathbb{I}_{V-1})^*\cdots)^*\mathbb{I}_{h+1})^*\mathbb{I}_{h},\\
	\forall h\in\dint{2,V}, &\mathbb{C}^{I}_h = \mathbb{I}_h(\mathbb{I}_{h-1}(\cdots(\mathbb{I}_{2}\mathbb{P}_{2})^*\cdots)^*\mathbb{P}_{h-1})^*\mathbb{P}_{h},\\
	\forall h\in\dint{1,V},& \mathbb{P}_h = C_{h}^*P_{h}C_{{h+1}}^*, \quad\quad \mathbb{I}_h = C_{{h+1}}^*I_{h}C_{h}^*,\\
	& C_{V+1} = C_1;
\end{array}
\]
	\item\label{en:3} $\lambda\in\solPIC[\mathbb{M}^P][\mathbb{M}^I][\mathbb{M}^C]$, where
\[\def\arraystretch{1.5}
\begin{array}{rcl}
	\mathbb{M}^P &=& \mathbb{P}_1(\mathbb{C}^{P}_2)^*\mathbb{P}_{2}(\mathbb{C}^{P}_{3})^*\cdots(\mathbb{C}^{P}_{V-1})^*\mathbb{P}_{V-1} \mathbb{P}_V,\\
	\mathbb{M}^I &=&\mathbb{I}_V(\mathbb{C}^{I}_{V-1})^*\mathbb{I}_{V-1}(\mathbb{C}^{I}_{V-2})^*\cdots(\mathbb{C}^{I}_{2})^*\mathbb{I}_2\mathbb{I}_1,\\
	\mathbb{M}^C &=& \mathbb{C}^{P}_1\oplus \mathbb{C}^{I}_V.
\end{array}
\]
\end{enumerate}

\end{theorem}

The time complexity is analyzed as follows.
Item~\ref{en:1} from Theorem~\ref{th:improved_algorithm} requires to check the existence of circuits with positive weight in $V$ precedence graphs consisting of $n$ nodes each; since each verification takes time $\pazocal{O}(n^3)$, all operations in item~\ref{en:1} are computed in time $\pazocal{O}(Vn^3)$.
As for item~\ref{en:2}, since $\mathbb{C}^P_h = \mathbb{P}_h(\mathbb{C}^P_{h+1})^* \mathbb{I}_h$ and $\mathbb{C}^I_h = \mathbb{I}_h (\mathbb{C}^I_{h-1})^* \mathbb{P}_h$, there are $\pazocal{O}(V)$ multiplications and Kleene star operations to be performed, each of which requires $\pazocal{O}(n^3)$ operations; the total computational time is thus $\pazocal{O}(Vn^3)$.
In item~\ref{en:3}, $\mathbb{M}^P$, $\mathbb{M}^I$, and $\mathbb{M}^C$ can be obtained performing $\pazocal{O}(V)$ multiplications and Kleene star operations and $\pazocal{O}(1)$ additions on $n\times n$ matrices; finally, set $\solPIC[\mathbb{M}^P][\mathbb{M}^I][\mathbb{M}^C]$ is computable in time $\pazocal{O}(n^4)$ using Algorithm~\ref{al:PIC-NCP}.

The formulas of Theorem~\ref{th:improved_algorithm} generalize those found in~\cite[Theorem 5.2]{gaubert1999modeling} for the throughput evaluation in max-plus automata (or heap models).
The formula in~\cite[Theorem 5.2]{gaubert1999modeling} can indeed be recovered from Theorem~\ref{th:improved_algorithm} considering the special case when $P_h = C_h = \pazocal{E}$ for all $h$, i.e., when no upper bound constraints or relations between $x_i(k)$ and $x_j(k)$ for all $i,j,k$ exist; in order to do so, observe that in this particular case Algorithm~\ref{al:PIC-NCP} simplifies significantly, see~\cite[Remark 3]{zorzenon2021nonpositive}.

\subsection{Analysis of intermittently periodic schedules}\label{su:piecewise_schedules}

We conclude this section with the analysis of particular trajectories of SLDIs with \textit{intermittently periodic schedules}, i.e., schedules \zor{that can be factorized in} the form 
\begin{equation}\label{eq:definition_intermittently_periodic_schedule}
	w = u_0 v_1^{m_1} u_1 v_2^{m_2} u_2 \cdots v_q^{m_q}u_{q},
\end{equation} 
where 
\begin{itemize}
	\item $u_0,\ldots,u_{q},v_1,\ldots,v_{q}$ are finite subschedules of lengths $U_0,\ldots,U_q\in\nato$ and $V_1,\ldots,V_{q}\in\nat$, respectively,
	\item $2\leq m_1,\ldots,m_{q-1}<+\infty$,
	\item $m_q$ is either an element of $\nat$ or $+\infty$; in the second case, $u_q$ is the empty string (of length $U_q=0$) and the schedule is called \textit{ultimately periodic}. 
\end{itemize}
\zor{
We call $u_0,\ldots,u_q$ the \textit{transient subschedules} and $v_1,\ldots,v_q$ the \textit{periodic subschedules} of the schedule $w$.
Observe that the factorization of an intermittently periodic schedule into transient and periodic subschedules may be not unique; for example, schedule $w = \wA\wA\wB\wA\wB\wA$ can be factorized into $w = \wA (\wA\wB)^2 \wA$ ($u_0=\wA$, $v_1=\wA\wB$, $u_1 = \wA$, $m_1=2$), into $w = \wA \wA (\wB\wA)^2$ ($u_0 = \wA\wA$, $v_1 = \wB\wA$, $m_1 = 2$), or even into $w = (\wA)^2 (\wB\wA)^2$ ($U_0 = U_1 = U_2 = 0$, $v_1 = \wA$, $v_2 = \wB\wA$, $m_1=m_2=2$).
In the reminder of the paper, to unequivocally indicate the intended schedule factorization into periodic and transient subschedules, we adopt the practice of writing explicitly exponents $m_h$ that elevate periodic subschedules $v_h$, and of writing extensively transient subschedules $u_h$ (even when $u_h$ could be written as a concatenation of a sequence of modes).}

The interpretation of intermittently periodic schedules in systems modeled by SLDIs is as follows: after the start-up of the system ($u_0$), a number of operations are executed cyclically ($v_1^{m_1}$), after which the system is re-initialized ($u_1$) before starting a new sequence of cyclical tasks ($v_2^{m_2}$), and so on; finally, after a finite number ($q$) of alternations between transient and cyclic working regimes, the system is either shut down ($u_q$) if $m_q\in\nat$ or works in periodic regime forever if $m_q = +\infty$.

\zor{In the reminder of the paper, we let $K_h = U_0+\sum_{j = 1}^{h} \left(m_jV_j+U_j\right)$ for all $h\in\dint{0,q}$.}
The objective of this section is to show that trajectories with important practical relevance under intermittently periodic schedules can be efficiently analyzed in SLDIs.
Namely, we study the existence of \textit{intermittently periodic trajectories}, that is, trajectories of the dater function \zor{$\{x(k)\}_{k\in\dint{1,K_q}}$} such that 
\[
	\zor{\{x(k)\}_{k\in\dint{K_{h-1}+1,K_{h-1} + m_{h}V_{h}}}}
\]
are $v_h$-periodic trajectories of period $\lambda_h$ for all $h\in\dint{1,q}$.
\zor{In other words, for all $h\in\dint{1,q}$, $j\in\dint{1,m_h-1}$, $k\in\dint{1,V_h}$, an intermittently periodic trajectory satisfies, for some $\lambda_h\in\mathbb{R}_{\geq 0}$,}
\begin{equation}\label{eq:intermittently_periodic}
	\zor{x\left(K_{h-1}+ j V_h + k\right) = \lambda_{h}x\left(K_{h-1}+(j-1)V_h+k\right).}
\end{equation}
Intermittently periodic trajectories in which $m_q = +\infty$ are referred to as \textit{ultimately periodic trajectories}.
Intermittently periodic trajectories generalize $v$-periodic trajectories, as they are $v_h$-periodic in each sequence of cyclical tasks $v_h^{m_h}$.

\zor{
Note that the definition of intermittently periodic trajectory depends on the specific factorization of schedule $w$ into transient ($u_h$) and periodic ($v_h$) subschedules.
For instance, let $w = \wA\wA\wB\wA\wB\wA$. 
A trajectory $\{x(k)\}_{k\in\dint{1,6}}$ for schedule $w$ factorized as $\wA(\wA\wB)^2\wA$ is intermittently periodic if it satisfies, for some $\lambda_1\in\R_{\geq 0}$,}
\[
	\zor{x(4) = \lambda_1 x(2),\quad x(5) = \lambda_1 x(3).}
\]
\zor{Considering the factorization $w = \wA \wA (\wB\wA)^2$, the trajectory is intermittently periodic if, for some $\lambda_1\in\R_{\geq 0}$,}
\[
	\zor{x(5) = \lambda_1 x(3),\quad x(6) = \lambda_1 x(4).}
\]
\zor{According to factorization $w = (\wA)^2 (\wB\wA)^2$, instead, the trajectory is intermittently periodic if, for some $\lambda_1,\lambda_2 \in\R_{\geq 0}$,}
\[
	\zor{x(2) = \lambda_1 x(1),\quad x(5) = \lambda_2 x(3),\quad x(6) = \lambda_2 x(4).}
\]

Trajectories of this type appear frequently in applications.
Typical examples are batch manufacturing systems, where each batch is processed in a periodic workflow, but switching between different batches leads to pauses and irregular transients (see, e.g., \cite{lee2012scheduling,froehlich2011transient}).
Urban railway systems also operate on a similar principle, with trains arriving at stations in a periodic manner during peak and off-peak hours, although the period may vary based on the time of day.
Moreover, note that, as shown in Subsection~\ref{su:SLDIS_PTEGs}, $1$-periodic trajectories of P-TEGs with strict initial conditions correspond to a particular case of ultimately periodic trajectories of SLDIs, in which $q=1$ and $U_0 = 1$.

\begin{example}
	\zor{This example shows how to transform the periods of finding intermittently periodic trajectories into an MPIC-NCP.
	We consider again the SLDIs of Example~\ref{ex:2} under the intermittently periodic schedule}
	\[
		w = \wA (\wC)^{m_1} \wB
	\]
	\zor{for some $m_1\geq 2$.
	The inequalities corresponding to schedule $w$ are: for all $k\in\dint{1,m_1}$,}
	\begin{equation}\label{eq:example_intermittently}
		\begin{array}{rcl}
		A_{\wA}^0\otimes x(1) \leq & x(1) & \leq B_{\wA}^0\stimes x(1)\\
		A_{\wA}^1\otimes x(1) \leq & x(2) & \leq B_{\wA}^1\stimes x(1)\\
		A_{\wC}^0\otimes x(k+1) \leq & x(k+1) & \leq B_{\wC}^0\stimes x(k+1)\\
		A_{\wC}^1\otimes x(k+1) \leq & x(k+2) & \leq B_{\wC}^1\stimes x(k+1)\\
		A_{\wB}^0\otimes x(m_1+2) \leq & x(m_1+2) & \leq B_{\wB}^0\stimes x(m_1+2).
		\end{array}
	\end{equation}
	\zor{To analyze the existence of intermittently periodic trajectories, we substitute in~\eqref{eq:example_intermittently}, for all $k\in\dint{1,m_1-1}$, $x(k+2) = \lambda_1^k x(2)$, obtaining: for all $k\in\dint{1,m_1}$,}
	\begin{equation}\label{eq:example_intermittently_trajectory}
		\begin{array}{rcl}
		A_{\wA}^0\otimes x(1) \leq & x(1) & \leq B_{\wA}^0\stimes x(1)\\
		A_{\wA}^1\otimes x(1) \leq & x(2) & \leq B_{\wA}^1\stimes x(1)\\
		A_{\wC}^0\otimes \cancel{\lambda_1^{k-1}} x(2) \leq & \cancel{\lambda_1^{k-1}} x(2) & \leq B_{\wC}^0\stimes \cancel{\lambda_1^{k-1}} x(2)\\
		A_{\wC}^1\otimes \cancel{\lambda_1^{k-1}} x(2) \leq &\lambda_1^{\cancel{k}}  x(2) & \leq B_{\wC}^1\stimes\cancel{\lambda_1^{k-1}}  x(2)\\
		A_{\wC}^1\otimes \lambda_1^{m_1-1} x(2) \leq & x(m_1+2) & \leq B_{\wC}^1\stimes\lambda_1^{m_1-1}  x(2)\\
		A_{\wB}^0\otimes x(m_1+2) \leq & x(m_1+2) & \leq B_{\wB}^0\stimes x(m_1+2).
		\end{array}
	\end{equation}
	\zor{It is possible to get rid of the term $\lambda_1^{m_1-1}$ in the penultimate inequalities by performing a change of variable.
	Let $\xi:\dint{1,3}\rightarrow \R^2$ be defined by}
	\[
		\xi(k) = \begin{dcases}
			x(k) & \mbox{if } k\in\{1,2\},\\
			\lambda_1^{-(m_1-1)} x(m_1+2) & \mbox{if } k=3.
		\end{dcases}
	\]
	\zor{By substituting $\xi$ into~\eqref{eq:example_intermittently_trajectory}, we obtain}
	\[
		\begin{array}{rcl}
		A_{\wA}^0\otimes \xi(1) \leq & \xi(1) & \leq B_{\wA}^0\stimes \xi(1)\\
		A_{\wA}^1\otimes \xi(1) \leq & \xi(2) & \leq B_{\wA}^1\stimes \xi(1)\\
		A_{\wC}^0\otimes \xi(2) \leq & \xi(2) & \leq B_{\wC}^0\stimes \xi(2)\\
		A_{\wC}^1\otimes \xi(2) \leq & \lambda_1 \xi(2) & \leq B_{\wC}^1\stimes\xi(2)\\
		A_{\wC}^1\otimes \cancel{\lambda_1^{m_1-1}} \xi(2) \leq & \cancel{\lambda_1^{m_1-1}} \xi(3) & \leq B_{\wC}^1 \stimes \cancel{\lambda_1^{m_1-1}}  \xi(2)\\
		A_{\wB}^0\otimes \cancel{\lambda_1^{m_1-1}}\xi(3) \leq &\cancel{\lambda_1^{m_1-1}} \xi(3) & \leq B_{\wB}^0\stimes\cancel{\lambda_1^{m_1-1}} \xi(3).
		\end{array}
	\]

\zor{Finally, by rewriting the inequalities in terms of the extended vector $[\xi^\top(1),\ \xi^{\top}(2),\ \xi^{\top}(3)]^\top$ and using Proposition~\ref{pr:max-min}, we can conclude that the set of $\lambda_1$'s for which the SLDIs admit an intermittent periodic trajectory under schedule $w$ coincides with the solution set $\solNCP{\lambda_1 P \oplus \lambda_1^{-1} I \oplus C}$ of the PIC-NCP defined by matrices}
	\[
		P = \begin{bmatrix}
			\mathcal{E} & \mathcal{E} & \mathcal{E} \\
			\mathcal{E} & P_{\wC} & \mathcal{E} \\
			\mathcal{E} & \mathcal{E} & \mathcal{E}
		\end{bmatrix},\quad
		I = \begin{bmatrix}
			\mathcal{E} & \mathcal{E} & \mathcal{E} \\
			\mathcal{E} & I_{\wC} & \mathcal{E} \\
			\mathcal{E} & \mathcal{E} & \mathcal{E}
		\end{bmatrix},\quad
		C = \begin{bmatrix}
			C_{\wA} & P_{\wA} & \mathcal{E} \\
			I_{\wA} & C_{\wC} & P_{\wC} \\
			\mathcal{E} & I_{\wC} & C_{\wB}
		\end{bmatrix},
	\]
	\zor{where $P_\wZ = B_{\wZ}^{1\sharp}$, $I_\wZ = A^1_\wZ$, $C_\wZ = A^0_\wZ \oplus B^{0\sharp}_\wZ$ for all $\wZ\in\{\wA,\wB,\wC\}$.
	Using Algorithm~\ref{al:PIC-NCP} we can compute}
	\[
		\solNCP{\lambda_1 P \oplus \lambda_1^{-1} I \oplus C} = [1,1].
	\]
	\zor{In this simple example, the set of admissible periods coincides with $\solSLDI[\wC]$; at a superficial glance one could think that this is a general fact, i.e., that for any intermittently periodic schedule of the form~\eqref{eq:definition_intermittently_periodic_schedule}, the set of valid periods $\lambda_h$ defined as in~\eqref{eq:intermittently_periodic} is simply $\solSLDI[v_h]$, for all $h\in\dint{1,q}$.
	This is however not the case, as shown in the following example.
}
\end{example}

\newcommand{\tikzmark}[1]{\tikz[overlay,remember picture] \node (#1) {};}%
\newcommand{\DrawBox}[3][]{%
    \tikz[overlay,remember picture]{%
    \draw[#1] ($({#2})+(-0.5em,0.7em)$) rectangle ($({#3})+(0.5em,-0.3em)$);}%
}%

\begin{example}[starving philosophers problem, cont.]\label{ex:3}
	We analyze the existence of intermittently and ultimately periodic trajectories for Example~\ref{ex:starving_philosophers} under schedule
	\[
		w = \winit (\wP_1\wP_1\wP_3\wP_2\wP_4)^{m_1} \wP_1\wP_3\wP_2\wP_4 (\wP_2\wP_4\wP_1\wP_3\wP_3)^{m_2},
	\]
	where $m_1\in\nat$ and $m_2 = +\infty$.
	\zor{For the first $m_1$ dining cycles, schedule $w$ forces the first philosopher to eat twice in a row simultaneously with the third philosopher, who eats once, after which it is the turn of philosophers $2$ and $4$ to eat; the subsequent dining cycle is identical to the previous $m_1$ cycles except that the $1$\textsuperscript{st} philosopher eats only once; thereafter the dining order is the one from Exercise~\ref{ex:starving_philosophers}.}
	Through computations analogous to those seen in \zor{the previous example}, we can show that all periods $\lambda_1$ and $\lambda_2$ (corresponding to the two periodic subschedules $v_1=\wP_1\wP_1\wP_3\wP_2\wP_4$ and $v_2=\wP_2\wP_4\wP_1\wP_3\wP_3$ in $w$) of consistent intermittently periodic trajectories are those in $\solNCP{A(\lambda_1,\lambda_2)} = \solNCP{\lambda_1P_1 \oplus \lambda_1^{-1}I_1 \oplus \lambda_2 P_2 \oplus \lambda_2^{-1}I_2 \oplus C}$, where
	\begin{equation}\label{eq:intermittent_periodic_matrix}
		\scriptstyle A(\lambda_1,\lambda_2) = 
		\begin{bsmallmatrix}
			\phantom{\lambda}\tikzmark{leftdot} C_\winit \tikzmark{rightdot}\phantom{\lambda} & P_\winit & \pazocal{E} & \pazocal{E} & \pazocal{E} & \pazocal{E} & \pazocal{E} & \pazocal{E} & \pazocal{E} & \pazocal{E} & \pazocal{E} & \pazocal{E} & \pazocal{E} & \pazocal{E} & \pazocal{E}\\
			I_\winit & \tikzmark{left} C_{\wP_1} & P_{\wP_1} & \pazocal{E} & \pazocal{E} & \lambda_1^{-1}I_{\wP_4} & \pazocal{E} & \pazocal{E} & \pazocal{E} & \pazocal{E} & \pazocal{E} & \pazocal{E} & \pazocal{E} & \pazocal{E} & \pazocal{E}\\
			\pazocal{E} & I_{\wP_1} & C_{\wP_1} & P_{\wP_1} & \pazocal{E} & \pazocal{E} & \pazocal{E} & \pazocal{E} & \pazocal{E} & \pazocal{E} & \pazocal{E} & \pazocal{E} & \pazocal{E} & \pazocal{E}& \pazocal{E} \\
			\pazocal{E} & \pazocal{E} & I_{\wP_1} & C_{\wP_3} & P_{\wP_3} & \pazocal{E} & \pazocal{E} & \pazocal{E} & \pazocal{E} & \pazocal{E} & \pazocal{E} & \pazocal{E} & \pazocal{E} & \pazocal{E}& \pazocal{E} \\
			\pazocal{E} & \pazocal{E} & \pazocal{E} & I_{\wP_3} & C_{\wP_2} & P_{\wP_2} & \pazocal{E} & \pazocal{E} & \pazocal{E} & \pazocal{E} & \pazocal{E} & \pazocal{E} & \pazocal{E} & \pazocal{E}& \pazocal{E} \\
			\pazocal{E} & \lambda_1P_{\wP_4} & \pazocal{E} & \pazocal{E} & I_{\wP_2} & C_{\wP_4} \tikzmark{right} & P_{\wP_4} & \pazocal{E} & \pazocal{E} & \pazocal{E} & \pazocal{E} & \pazocal{E} & \pazocal{E} & \pazocal{E} & \pazocal{E}\\
			\pazocal{E} & \pazocal{E} & \pazocal{E} & \pazocal{E} & \pazocal{E} & I_{\wP_4} & \phantom{\lambda}\tikzmark{leftdot1} C_{\wP_1}\phantom{\lambda} & P_{\wP_1} & \pazocal{E} & \pazocal{E} & \pazocal{E} & \pazocal{E} & \pazocal{E} & \pazocal{E} & \pazocal{E}\\
			\pazocal{E} & \pazocal{E} & \pazocal{E} & \pazocal{E} & \pazocal{E} & \pazocal{E} & I_{\wP_1} & C_{\wP_3} & P_{\wP_3} & \pazocal{E} & \pazocal{E} & \pazocal{E} & \pazocal{E} & \pazocal{E} & \pazocal{E}\\
			\pazocal{E} & \pazocal{E} & \pazocal{E} & \pazocal{E} & \pazocal{E} & \pazocal{E} & \pazocal{E} & I_{\wP_3} & C_{\wP_2} & P_{\wP_2} & \pazocal{E} & \pazocal{E} & \pazocal{E} & \pazocal{E}& \pazocal{E} \\
			\pazocal{E} & \pazocal{E} & \pazocal{E} & \pazocal{E} & \pazocal{E} & \pazocal{E} & \pazocal{E} & \pazocal{E} & I_{\wP_2} &\phantom{\lambda} C_{\wP_4} \tikzmark{rightdot1}\phantom{\lambda} & P_{\wP_4} & \pazocal{E} & \pazocal{E} & \pazocal{E}& \pazocal{E} \\
			\pazocal{E} & \pazocal{E} & \pazocal{E} & \pazocal{E} & \pazocal{E} & \pazocal{E} & \pazocal{E} & \pazocal{E} & \pazocal{E} & I_{\wP_4} & \tikzmark{left1} C_{\wP_2} & P_{\wP_2} & \pazocal{E} & \pazocal{E} & \lambda_2^{-1}I_{\wP_3} \\
			\pazocal{E} & \pazocal{E} & \pazocal{E} & \pazocal{E} & \pazocal{E} & \pazocal{E} & \pazocal{E} & \pazocal{E} & \pazocal{E} & \pazocal{E} & I_{\wP_2} & C_{\wP_4} & P_{\wP_4} & \pazocal{E} & \pazocal{E} \\
			\pazocal{E} & \pazocal{E} & \pazocal{E} & \pazocal{E} & \pazocal{E} & \pazocal{E} & \pazocal{E} & \pazocal{E} & \pazocal{E} & \pazocal{E} & \pazocal{E} & I_{\wP_4} & C_{\wP_1} & P_{\wP_1} & \pazocal{E} \\
			\pazocal{E} & \pazocal{E} & \pazocal{E} & \pazocal{E} & \pazocal{E} & \pazocal{E} & \pazocal{E} & \pazocal{E} & \pazocal{E} & \pazocal{E} & \pazocal{E} & \pazocal{E} & I_{\wP_{1}} & C_{\wP_3} & P_{\wP_3} \\
			\pazocal{E} & \pazocal{E} & \pazocal{E} & \pazocal{E} & \pazocal{E} & \pazocal{E} & \pazocal{E} & \pazocal{E} & \pazocal{E} & \pazocal{E} & \lambda_2 P_{\wP_3} & \pazocal{E} & \pazocal{E} & I_{\wP_{3}} & C_{\wP_3} \tikzmark{right1}
		\end{bsmallmatrix},
	\end{equation}
	\DrawBox[dashed]{left}{right}%
	\DrawBox[dashed]{left1}{right1}%
	\DrawBox[dotted]{leftdot}{rightdot}%
	\DrawBox[dotted]{leftdot1}{rightdot1}%
	$P_\wZ = B_\wZ^{1\sharp}$, $I_\wZ=A_\wZ^1$, and $C_\wZ = A_\wZ^0\oplus B_\wZ^{0\sharp}$ for all $\wZ\in\{\winit,\wP_1,\wP_2,\wP_3,\wP_4\}$.
	Thus, the existence of intermittently periodic trajectories can be checked by solving an MPIC-NCP.
	For instance, an intermittently periodic trajectory that minimizes the sum of periods $\lambda_1$ and $\lambda_2$ can be found by solving the following linear programming problem\footnote{The problem is written in the max-plus algebra for simplicity. To get the standard linear programming formulation, we refer to Section~\ref{su:MPIC-NCP} and recall that $\lambda_1\otimes \lambda_2$ coincides with $\lambda_1 + \lambda_2$.}:
	\begin{equation}\label{eq:LP_philosophers}
		\begin{array}{cl}
		\displaystyle\min_{x\in\R^{75},(\lambda_1,\lambda_2)\in\R^2_{\geq 0}} \quad & \lambda_1 \otimes \lambda_2 \\
		\mbox{subject to} \quad & A(\lambda_1,\lambda_2) \otimes x \preceq x.
		\end{array}
	\end{equation}
	The constraints of the above problem can be expressed as 212 inequalities (with each inequality corresponding to an element of $A(\lambda_1,\lambda_2)$ different from $-\infty$) in 77 variables.
	Clearly, the optimization problem can be solved efficiently by linear programming solvers such as interior point methods or the simplex method.

	However, since the number of optimization variables and constraints increases (linearly) with the length of subschedules $u_0,\ldots,u_q$ and $v_1,\ldots,v_q$, this method can be impractical for larger problems.
	This becomes especially evident when considering hard scheduling problems such as the following one: find the schedule $w$, from a prescribed set of intermittently periodic schedules, which admits the intermittently periodic trajectory with minimal performance index $\bigotimes_{i=1}^q \lambda_i$.
	The intrinsic NP-hardness of the problem requires the use of enumeration methods for its exact solution, and relying on standard solutions approaches for problems like~\eqref{eq:LP_philosophers} to evaluate the performance index can result in a slow optimization procedure.
	Therefore, it is natural to ask whether an alternative, less expensive technique to solve problems of the form~\eqref{eq:LP_philosophers} exploiting the structure of the matrix in~\eqref{eq:intermittent_periodic_matrix} exists, similarly to what seen in Theorem~\ref{th:improved_algorithm} for periodic schedules; an affirmative answer to this question will be provided in Theorem~\ref{th:intermittent_periodic}.

	Returning to~\eqref{eq:LP_philosophers}, its solution reveals that the optimal value of $\lambda_1 \otimes \lambda_2$ is $19$, corresponding to $\lambda_1 = 11$ and $\lambda_2 = 8$.
	The Gantt chart of Figure~\ref{fi:gantt_philosophers_intermittent} shows a trajectory corresponding to a solution of the optimization problem, for $m_1 = 2$.

	Before concluding the example note that, for the SLDIs $\pazocal{S}$ modeling the starving philosophers problem, $\solSLDI[\wP_2\wP_4\wP_1\wP_3\wP_3] = [7.5,16]$ \zor{(the lower bound is the period of the trajectory shown in the Gantt chart of Figure~\ref{fi:gantt_philosophers})} whereas the value of $\lambda_2$ solving~\eqref{eq:LP_philosophers} is $8> 7.5$.
	This shows a remarkable fact that is exclusive to systems with upper bound constraints: the cycle times valid in intermittently periodic trajectories may be different from those obtained by considering periodic subschedules independently.
	The reason for this phenomenon has to be searched in the structure of the precedence graph of matrix $A(\lambda_1,\lambda_2)$ (schematically illustrated in Figure~\ref{fi:intermittent_graph} on page~\pageref{fi:intermittent_graph}).
	Here the critical circuit (i.e., the circuit with largest weight) can be generated from arcs connecting portions of the graphs related to different regimes (periodic or transient) of the schedule.

	\begin{figure}
		\centering
		\resizebox{1\textwidth}{!}{
		\input{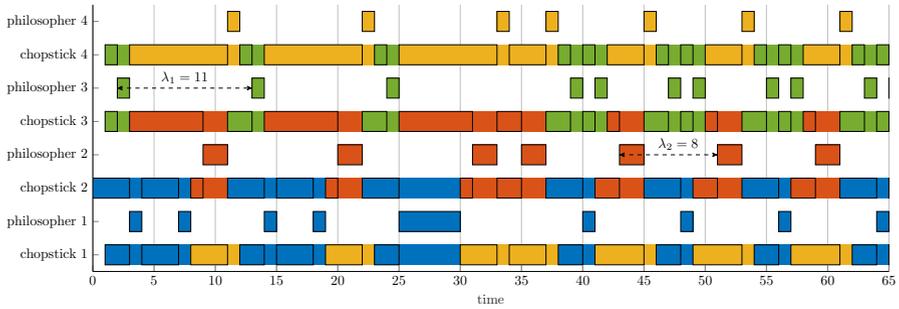}}
		\caption{Gantt chart representing an intermittently periodic trajectory for the starving philosophers problem.}\label{fi:gantt_philosophers_intermittent}
	\end{figure}
\end{example}

The latter two examples suggested the following proposition\zor{, which is proven in Appendix~\ref{ap:Proposition7}}.

\begin{proposition}\label{pr:intermittent_periodic}
The set of periods $(\lambda_1,\ldots,\lambda_q)\in\R_{\geq 0}^q$ of intermittently periodic trajectories that are consistent for a given intermittently periodic schedule coincides with $\solNCP{\bigoplus_{h=1}^q (\lambda_h P_h\oplus \lambda_h^{-1} I_h)\oplus C}$, where $P_h,I_h,C$ are appropriately defined square matrices with $(U_0+\sum_{h=1}^q V_h + U_h)n$ rows and columns.
\end{proposition}

Similarly to Theorem~\ref{th:improved_algorithm}, it is possible to reduce the complexity of the problem by leveraging the sparsity and modularity of matrix $\bigoplus_{h=1}^q (\lambda_h P_h\oplus \lambda_h^{-1} I_h)\oplus C$.
With the term modularity, we mean the recognizable block structure of the matrix; for an illustrative example see the dashed and dotted blocks in matrix $A(\lambda_1,\lambda_2)$ of~\eqref{eq:intermittent_periodic_matrix}.
The main result, which is based on an extensive application of Theorem~\ref{th:improved_algorithm}, is stated below.
Details of this result and its proof is provided for the particular case of Example~\ref{ex:3} in Appendix~\ref{ap:intermittent_periodic_improved}.

\begin{theorem}\label{th:intermittent_periodic}
The MPIC-NCP of Proposition~\ref{pr:intermittent_periodic} can be transformed into an equivalent one where the precedence graph to be studied has $qn$ nodes (instead of $(U_0+\sum_{h=1}^q V_h + U_h)n$).
The reduction requires $\pazocal{O}\left(\left(U_0+\sum_{h=1}^q V_h + U_h\right)n^3\right)$ operations.
\end{theorem}

Observe that performing the reduction takes a number of operations that is linear in the sum of the lengths of subschedules $u_0,\ldots,u_q$ and $v_1,\ldots,v_q$.
As no linear programming solver with linear worst-case complexity has ever been found, the advantage of the reduction becomes more prominent with longer subschedules.

\begin{example}[starving philosophers problem, cont.]

Let $A(\lambda_1,\lambda_2)$ be as in~\eqref{eq:intermittent_periodic_matrix}.
Applying Theorem~\ref{th:intermittent_periodic} as explained in Appendix~\ref{ap:intermittent_periodic_improved}, it can be shown that $\graph(A(\lambda_1,\lambda_2))\in\nonegset$ if and only if $\graph(\widetilde{A}(\lambda_1,\lambda_2))\in\nonegset$, where 
\[
	\widetilde{A}(\lambda_1,\lambda_2) = \lambda_1 P_1 \oplus \lambda_1^{-1}I_1 \oplus \lambda_2 P_2 \oplus \lambda_2^{-1} I_2 \oplus C,
\]
and $P_1,I_1,P_2,I_2,C$ are matrices of dimension $10\times 10$ with coefficients in $\Rmax$.
In particular, the linear programming problem in~\eqref{eq:LP_philosophers} has the same optimal value as the following one:
\begin{equation}\label{eq:LP_philosophers_improved}
	\begin{array}{cl}
	\displaystyle\min_{x\in\R^{10},(\lambda_1,\lambda_2)\in\R^2_{\geq 0}} \quad & \lambda_1 \otimes \lambda_2 \\
	\mbox{subject to} \quad & \widetilde{A}(\lambda_1,\lambda_2) \otimes x \preceq x.
	\end{array}
\end{equation}
It can be verified that the constraints of the above problem can be expressed as 146 inequalities in 12 variables.
Hence, compared with the constraints in~\eqref{eq:LP_philosophers}, we have a $31\%$ reduction in the number of inequalities and $84\%$ reduction in the number of variables.
\end{example}

\section{Practically motivated example}\label{se:example}

The example we present is a multi-product processing network taken from~\cite{KATS20081196}.
Examples of such networks are electroplating lines and cluster tools.
Consider a manufacturing system consisting of 5 processing stations $S_1,\ldots,S_5$ and a robot of capacity one.
The system can treat two types of parts, part $\wA$, which requires to be processed on $S_1$, $S_3$, and $S_5$ (in this order), and part $\wB$, which must follow route $S_2$, $S_1$, $S_4$, $S_5$.
The task of the robot is to transport parts of type $\wA$ and $\wB$ from an input storage $S_0$ to their first processing stations, between the processing stations (in the right order), and finally from the last processing station to an output storage $S_6$.
The time the robot needs to travel from $S_i$ to $S_j$ is $\tau_{ij}$ when it is not carrying any part, and $\tau_{ij}^\wZ$ when it is carrying part $\wZ\in\{\wA,\wB\}$.
Moreover, the processing time for part $\wZ$ in station $S_i$ must be within the interval $\iota_i^\wZ=[L_{i}^\wZ,R_{i}^\wZ]\subset \R_{\geq 0}$.

We consider the following parameters for the processing network: $\tau_{ij} = |i-j|$, $\tau_{ij}^\wA = \tau_{ij}+1$, $\tau_{ij}^\wB=\tau_{ij}+2$, $\iota_1^\wA=[10,15]$, $\iota_3^\wA=[40,140]$, $\iota_5^\wA=[20,30]$, $\iota_2^\wB=[50,150]$, $\iota_1^\wB=[10,20]$, $\iota_4^\wB=[30,150]$, $\iota_5^\wB=[20,30]$.

\subsection{Cycle time analysis}\label{su:cycle_time_analysis}

A classical scheduling problem in this kind of processing networks is to find a periodic robot operation sequence that minimizes the cycle time and avoids time-window constraints violations; in the literature, this is referred to as the \textit{single-robot} (or hoist) \textit{cyclic scheduling problem}~\cite{KATS20081196}.
Such an optimization problem, which is strongly NP-hard, can be divided into two subproblems:
\begin{enumerate}[label=\textbf{P\arabic*}]
	\item\label{en:P1} the cycle time minimization, given a fixed sequence of robot operations,
	\item\label{en:P2} the search for the optimal robot operations sequence.  
\end{enumerate}
In this section we focus on~\ref{en:P1}, which is polynomial-time solvable: clearly, an efficient algorithm for~\ref{en:P1} can be used as subroutine in search procedures to solve the more general scheduling problem (\ref{en:P1}+\ref{en:P2}).

To find the cycle time for a given robot operation sequence, we will firstly assume, as is standard in robotic cyclic scheduling problems~\cite{KATS20081196,levner2010complexity,kats2002cyclic}, that the system is already treating parts in a periodic manner from the initial time; the application of Theorem~\ref{th:improved_algorithm} will then provide the possible periods of such treatment plans.
In particular, we suppose that initially station $S_3$ is processing a part of type $\wA$, and $S_2$, $S_4$ are processing parts of type $\wB$.
Of course, this assumption is not met by real systems at start-up time (when all stations are empty), and will thus be relaxed in the next section.

We denote by $S_i\xrightarrow{\wZ} S_j$ robot operation "unload a part of type $\wZ$ from $S_i$, transport it to and load it into $S_j$" and by $\rightarrow S_j$ operation "travel from the current location to $S_j$ and wait if necessary".
A schedule for this process is a sequence of modes $w\in\{\wA,\wB\}^\omega$, where mode $\wA$ represents the sequence of operations
\[
	\rightarrow S_3, S_3\xrightarrow{\wA} S_5,\rightarrow S_0,S_0 \xrightarrow{\wA} S_1,\rightarrow S_5,S_5 \xrightarrow{\wA} S_6,\rightarrow S_1,S_1 \xrightarrow{\wA} S_3
\]
and mode $\wB$ represents 
\[
	\rightarrow S_4, S_4\! \xrightarrow{\wB} S_5, \rightarrow S_2, S_2\! \xrightarrow{\wB} S_1, \rightarrow S_5, S_5\! \xrightarrow{\wB} S_6, \rightarrow S_0, S_0\! \xrightarrow{\wB} S_2, \rightarrow S_1, S_1\! \xrightarrow{\wB} S_4.
\]
Initially, the robot is positioned at $S_3$ if $w_1=\wA$ or at $S_4$ if $w_1=\wB$.

\begin{figure*}[ht]
	\centering
\subfigure[$\mbox{P-TEG}_\wA$ for parts of type $\wA$.]{
		\centering
		\resizebox{.8\textwidth}{!}{
		\begin{tikzpicture}[node distance=.5cm and .8cm,>=stealth',bend angle=30,thick]
\scriptsize
\node[transitionV,label=below:{$t_0$}] (t0) {};
\node[place,right=of t0,label=above:{$[\tau_{01}^\wA,\infty]$}] (p01) {};
\node[transitionV,right=of p01,label=below:{$t_{1in}$}] (t1in) {};
\node[place,myred,right=of t1in,label=below:{$[L_1^\wA,R_1^\wA]$}] (p1) {};
\node[transitionV,right=of p1,label=below:{$t_{1out}$}] (t1out) {};
\node[place,right=of t1out,label=above:{$[\tau_{13}^\wA,\infty]$}] (p13) {};
\node[transitionV,right=of p13,label=above:{$t_{3in}$}] (t3in) {};
\node[place,myred,tokens=1,right=of t3in,label=above:{$[L_3^\wA,R_3^\wA]$}] (p3) {};
\node[transitionV,right=of p3,label=above:{$t_{3out}$}] (t3out) {};
\node[place,right=of t3out,label=above:{$[\tau_{35}^\wA,\infty]$}] (p35) {};
\node[transitionV,right=of p35,label=below:{$t_{5in}$}] (t5in) {};
\node[place,myred,right=of t5in,label=below:{$[L_5^\wA,R_5^\wA]$}] (p5) {};
\node[transitionV,right=of p5,label=below:{$t_{5out}$}] (t5out) {};
\node[place,right=of t5out,label=above:{$[\tau_{56}^\wA,\infty]$}] (p56) {};
\node[transitionV,right=of p56,label=below:{$t_6$}] (t6) {};

\node[place,myblue,tokens=1,below=.4cm of p3,label=above:{$[\tau_{33},\infty]$}] (p33r) {};
\node[place,myblue,below=1cm of p13,label=above:{$[\tau_{50},\infty]$}] (p50r) {};
\node[place,myblue,above=of p3,label=above:{$[\tau_{15},\infty]$}] (p15r) {};
\node[place,myblue,below=1cm of p35,label=above:{$[\tau_{61},\infty]$}] (p61r) {};

\draw (t0) edge[->] (p01);
\draw (p01) edge[->] (t1in);
\draw (t1in) edge[->] (p1);
\draw (p1) edge[->] (t1out);
\draw (t1out) edge[->] (p13);
\draw (p13) edge[->] (t3in);
\draw (t3in) edge[->] (p3);
\draw (p3) edge[->] (t3out);
\draw (t3out) edge[->] (p35);
\draw (p35) edge[->] (t5in);
\draw (t5in) edge[->] (p5);
\draw (p5) edge[->] (t5out);
\draw (t5out) edge[->] (p56);
\draw (p56) edge[->] (t6);

\draw (t3in.-75) edge[bend right=10,->] (p33r);
\draw (p33r) edge[bend right=10,->] (t3out.180+75);
\draw (t5in.180+75) edge[bend left=10,->] (p50r);
\draw (p50r) edge[bend left=10,->] (t0.-75);
\draw (t1in.75) edge[bend left=10,->] (p15r);
\draw (p15r) edge[bend left=10,->] (t5out.180-75);
\draw (t6.180+75) edge[bend left=10,->] (p61r);
\draw (p61r) edge[bend left=10,->] (t1out.-75);

\end{tikzpicture}
		}
		\label{sub:1}}
	\subfigure[$\mbox{P-TEG}_\wB$ for parts of type $\wB$.]{
		\centering
		\resizebox{1\textwidth}{!}{
		\begin{tikzpicture}[node distance=.5cm and .8cm,>=stealth',bend angle=30,thick]
\scriptsize
\node[transitionV,label=below:{$t_0$}] (t0) {};
\node[place,right=of t0,label=above:{$[\tau_{02}^\wB,\infty]$}] (p02) {};
\node[transitionV,right=of p02,label=below:{$t_{2in}$}] (t2in) {};
\node[place,myred,tokens=1,right=of t2in,label=below:{$[L_2^\wB,R_2^\wB]$}] (p2) {};
\node[transitionV,right=of p2,label=below:{$t_{2out}$}] (t2out) {};
\node[place,right=of t2out,label=above:{$[\tau_{21}^\wB,\infty]$}] (p21) {};
\node[transitionV,right=of p21,label=below:{$t_{1in}$}] (t1in) {};
\node[place,myred,right=of t1in,label=below:{$[L_1^\wB,R_1^\wB]$}] (p1) {};
\node[transitionV,right=of p1,label=below:{$t_{1out}$}] (t1out) {};
\node[place,right=of t1out,label=above:{$[\tau_{14}^\wB,\infty]$}] (p14) {};
\node[transitionV,right=of p14,label=above:{$t_{4in}$}] (t4in) {};
\node[place,myred,tokens=1,right=of t4in,label=above:{$[L_4^\wB,R_4^\wB]$}] (p4) {};
\node[transitionV,right=of p4,label=above:{$t_{4out}$}] (t4out) {};
\node[place,right=of t4out,label=above:{$[\tau_{45}^\wB,\infty]$}] (p45) {};
\node[transitionV,right=of p45,label=below:{$t_{5in}$}] (t5in) {};
\node[place,myred,right=of t5in,label=below:{$[L_5^\wB,R_5^\wB]$}] (p5) {};
\node[transitionV,right=of p5,label=below:{$t_{5out}$}] (t5out) {};
\node[place,right=of t5out,label=above:{$[\tau_{56}^\wB,\infty]$}] (p56) {};
\node[transitionV,right=of p56,label=below:{$t_6$}] (t6) {};

\node[place,myblue,tokens=1,below=.4cm of p4,label=above:{$[\tau_{44},\infty]$}] (p44r) {};
\node[place,myblue,below=1cm of p14,label=above:{$[\tau_{52},\infty]$}] (p52r) {};
\node[place,myblue,above=of p21,label=above:{$[\tau_{21},\infty]$}] (p21r) {};
\node[place,myblue,above=of p4,label=above:{$[\tau_{15},\infty]$}] (p15r) {};
\node[place,myblue,below=.4cm of p52r,label=above:{$[\tau_{60},\infty]$}] (p60r) {};

\draw (t0) edge[->] (p02);
\draw (p02) edge[->] (t2in);
\draw (t2in) edge[->] (p2);
\draw (p2) edge[->] (t2out);
\draw (t2out) edge[->] (p21);
\draw (p21) edge[->] (t1in);
\draw (t1in) edge[->] (p1);
\draw (p1) edge[->] (t1out);
\draw (t1out) edge[->] (p14);
\draw (p14) edge[->] (t4in);
\draw (t4in) edge[->] (p4);
\draw (p4) edge[->] (t4out);
\draw (t4out) edge[->] (p45);
\draw (p45) edge[->] (t5in);
\draw (t5in) edge[->] (p5);
\draw (p5) edge[->] (t5out);
\draw (t5out) edge[->] (p56);
\draw (p56) edge[->] (t6);

\draw (t4in.-75) edge[bend right=10,->] (p44r);
\draw (p44r) edge[bend right=10,->] (t4out.180+75);
\draw (t5in.180+75) edge[bend left=10,->] (p52r);
\draw (p52r) edge[bend left=10,->] (t2out.-75);
\draw (t1in.75) edge[bend left=10,->] (p15r);
\draw (p15r) edge[bend left=10,->] (t5out.180-75);
\draw (t6.180+75) edge[bend left=10,->] (p60r);
\draw (p60r) edge[bend left=10,->] (t0.-75);
\draw (t2in.75) edge[bend left=10,->] (p21r);
\draw (p21r) edge[bend left=10,->] (t1out.180-75);

\end{tikzpicture}
		}
		\label{sub:2}}
	\caption{P-TEGs modeling the processing network considering only types of one part.
	A token in a place colored \textcolor{myred}{\textbf{red}} and \textcolor{black}{\textbf{black}}, respectively \textcolor{myblue}{\textbf{blue}}, represents a part being processed in a station and the robot moving with, respectively without carrying a part.}
	\label{fig:P-TEGs_processing_network}
\end{figure*}

Let us first model the processing network when only part $\wA$, respectively, $\wB$ is considered.
In this way, we obtain two P-TEGs, $\mbox{P-TEG}_\wA$ and $\mbox{P-TEG}_\wB$ (shown in Figure~\ref{fig:P-TEGs_processing_network}), each of which represents the behavior of the system when processing only parts of one type.
Using Algorithm~\ref{al:PIC-NCP}, we can find that the cycle times of the network when processing only parts of type $\wA$, $\wB$ are all values in $[73,+\infty[$, and $[72,192]$, respectively.
Now, from the obtained P-TEGs, we can model the processing network in the case where both part-types are considered as SLDIs $\pazocal{S}=(\{\wA,\wB\},A^0,A^1,B^0,B^1)$.
To do so, we must define matrices $A^0_\wZ,A^1_\wZ,B^0_\wZ,B^1_\wZ\in\Rmax^{n\times n}$ for $\wZ\in\{\wA,\wB\}$ appropriately: we start by adding in $\mbox{P-TEG}_\wA$ (respectively, $\mbox{P-TEG}_\wB$) the missing transitions from $\mbox{P-TEG}_\wB$ (respectively, $\mbox{P-TEG}_\wA$) -- the obtained P-TEGs have both $n = 12$ transitions (in  general, $n=2+2\ \times$ number of processing stations).
For each new transition $t_i$ of $\mbox{P-TEG}_\wZ$, we define $(\MA^1_{\wZ})_{ii}=(\MB^1_{\wZ})_{ii} = 0$; this is done to store in auxiliary variables $x_i(k)$ the last entrance and exit times of parts in stations that are not used in mode $\wZ$.
Moreover, to model the transportation of the robot from $S_3$ to $S_4$ (respectively, from $S_4$ to $S_3$) after each switching of mode from $\wA$ to $\wB$ (respectively, from $\wB$ to $\wA$), we set $(\MA^1_\wA)_{4out,3in}=\tau_{34}$ (respectively, $(\MA^1_\wB)_{3out,4in}=\tau_{43}$).
The other elements of $\MA^0_\wZ,\MA^1_\wZ,\MB^0_\wZ,\MB^1_\wZ$ are taken from the characteristic matrices of $\mbox{P-TEG}_\wZ$, for $\wZ\in\{\wA,\wB\}$.
\setcounter{MaxMatrixCols}{14}
This modeling procedure results in the following matrices for mode $\wA$:
\[
	A^0_\wA \oplus B^{0\sharp}_\wA = 
	\begin{bNiceMatrix}[small,r,first-row,last-col=13]
	0 & 1in & 1out & 2in & 2out & 3in & 3out & 4in & 4out & 5in & 5out & 6 & \\
-\infty & -\infty & -\infty & -\infty & -\infty & -\infty & -\infty & -\infty & -\infty &    5 & -\infty & -\infty & 0\\
2 &  -\infty & -15 & -\infty & -\infty & -\infty & -\infty & -\infty & -\infty & -\infty & -\infty & -\infty & 1in\\
-\infty &  10 & -\infty & -\infty & -\infty & -\infty & -\infty & -\infty & -\infty & -\infty & -\infty &    5 & 1out\\
-\infty & -\infty & -\infty & -\infty & -\infty & -\infty & -\infty & -\infty & -\infty & -\infty & -\infty & -\infty & 2in \\
-\infty & -\infty & -\infty & -\infty & -\infty & -\infty & -\infty & -\infty & -\infty & -\infty & -\infty & -\infty & 2out \\
-\infty & -\infty &    3 & -\infty & -\infty & -\infty & -\infty & -\infty & -\infty & -\infty & -\infty & -\infty & 3in \\
-\infty & -\infty & -\infty & -\infty & -\infty & -\infty & -\infty & -\infty & -\infty & -\infty & -\infty & -\infty & 3out\\
-\infty & -\infty & -\infty & -\infty & -\infty & -\infty & -\infty & -\infty & -\infty & -\infty & -\infty & -\infty & 4in\\
-\infty & -\infty & -\infty & -\infty & -\infty & -\infty & -\infty & -\infty & -\infty & -\infty & -\infty & -\infty & 4out \\
-\infty & -\infty & -\infty & -\infty & -\infty & -\infty &    3 & -\infty & -\infty & -\infty & -30 & -\infty & 5in\\
-\infty &    4 & -\infty & -\infty & -\infty & -\infty & -\infty & -\infty & -\infty &   20 & -\infty & -\infty & 5out\\
-\infty & -\infty & -\infty & -\infty & -\infty & -\infty & -\infty & -\infty & -\infty & -\infty &   2  & -\infty & 6
	\end{bNiceMatrix},
\] 
\[
	A^1_\wA = 
	\begin{bNiceMatrix}[small,r,first-row,last-col=13]
	0 & 1in & 1out & 2in & 2out & 3in & 3out & 4in & 4out & 5in & 5out & 6 & \\
  -\infty & -\infty & -\infty & -\infty & -\infty & -\infty & -\infty & -\infty & -\infty & -\infty & -\infty & -\infty & 0\\
    -\infty & -\infty & -\infty & -\infty & -\infty & -\infty & -\infty & -\infty & -\infty & -\infty & -\infty & -\infty & 1in\\
  -\infty & -\infty & -\infty & -\infty & -\infty & -\infty & -\infty & -\infty & -\infty & -\infty & -\infty & -\infty & 1out\\
  -\infty & -\infty & -\infty &    0 & -\infty & -\infty & -\infty & -\infty & -\infty & -\infty & -\infty & -\infty & 2in\\
  -\infty & -\infty & -\infty & -\infty &    0 & -\infty & -\infty & -\infty & -\infty & -\infty & -\infty & -\infty & 2out\\
  -\infty & -\infty & -\infty & -\infty & -\infty & -\infty & -\infty & -\infty & -\infty & -\infty & -\infty & -\infty & 3in\\
  -\infty & -\infty & -\infty & -\infty & -\infty &   40 & -\infty & -\infty & -\infty & -\infty & -\infty & -\infty & 3out\\
  -\infty & -\infty & -\infty & -\infty & -\infty & -\infty & -\infty &    0 & -\infty & -\infty & -\infty & -\infty & 4in\\
  -\infty & -\infty & -\infty & -\infty & -\infty &    1 & -\infty & -\infty &    0 & -\infty & -\infty & -\infty & 4out\\
  -\infty & -\infty & -\infty & -\infty & -\infty & -\infty & -\infty & -\infty & -\infty & -\infty & -\infty & -\infty & 5in\\
  -\infty & -\infty & -\infty & -\infty & -\infty & -\infty & -\infty & -\infty & -\infty & -\infty & -\infty & -\infty & 5out\\
  -\infty & -\infty & -\infty & -\infty & -\infty & -\infty & -\infty & -\infty & -\infty & -\infty & -\infty & -\infty & 6
	\end{bNiceMatrix},
\] 
\[
	B^{1\sharp}_\wA = 
	\begin{bNiceMatrix}[small,r,first-row,last-col=13]
	0 & 1in & 1out & 2in & 2out & 3in & 3out & 4in & 4out & 5in & 5out & 6 & \\
   -\infty &  -\infty &  -\infty &  -\infty &  -\infty &  -\infty &  -\infty &  -\infty &  -\infty &  -\infty &  -\infty &  -\infty & 0\\
   -\infty &  -\infty &  -\infty &  -\infty &  -\infty &  -\infty &  -\infty &  -\infty &  -\infty &  -\infty &  -\infty &  -\infty & 1in\\
   -\infty &  -\infty &  -\infty &  -\infty &  -\infty &  -\infty &  -\infty &  -\infty &  -\infty &  -\infty &  -\infty &  -\infty & 1out\\
   -\infty &  -\infty &  -\infty &    0 &  -\infty &  -\infty &  -\infty &  -\infty &  -\infty &  -\infty &  -\infty &  -\infty & 2in\\
   -\infty &  -\infty &  -\infty &  -\infty &    0 &  -\infty &  -\infty &  -\infty &  -\infty &  -\infty &  -\infty &  -\infty & 2out\\
   -\infty &  -\infty &  -\infty &  -\infty &  -\infty &  -\infty &  -140 &  -\infty &  -\infty &  -\infty &  -\infty &  -\infty & 3in\\
   -\infty &  -\infty &  -\infty &  -\infty &  -\infty & -\infty &  -\infty &  -\infty &  -\infty &  -\infty &  -\infty &  -\infty & 3out\\
   -\infty &  -\infty &  -\infty &  -\infty &  -\infty &  -\infty &  -\infty &    0 &  -\infty &  -\infty &  -\infty &  -\infty & 4in\\
   -\infty &  -\infty &  -\infty &  -\infty &  -\infty &  -\infty &  -\infty &  -\infty &    0 &  -\infty &  -\infty &  -\infty & 4out\\
   -\infty &  -\infty &  -\infty &  -\infty &  -\infty &  -\infty &  -\infty &  -\infty &  -\infty &  -\infty &  -\infty &  -\infty & 5in\\
   -\infty &  -\infty &  -\infty &  -\infty &  -\infty &  -\infty &  -\infty &  -\infty &  -\infty &  -\infty &  -\infty &  -\infty & 5out\\
   -\infty &  -\infty &  -\infty &  -\infty &  -\infty &  -\infty &  -\infty &  -\infty &  -\infty &  -\infty &  -\infty &  -\infty & 6
	\end{bNiceMatrix}
\]
The modeling effort required to define $\pazocal{S}$ is repaid by the possibility to use Theorem~\ref{th:improved_algorithm} for computing the minimum and maximum cycle times corresponding to a schedule $w=v^K$, for $K\in\nat\cup\{+\infty\}$.
For instance, we get $\solSLDI[\wB\wA] = [77,192]$.
This means that, using schedule $(\wB\wA)^\omega$, the time between subsequent completions of a product of the same type is at least $72$ and at most $192$ time units.

To appreciate the advantage of using the algorithm derived from Theorem~\ref{th:improved_algorithm}, in Figure~\ref{fig:comparison} we compare the computational time to obtain $\solSLDI$ with increasing subschedule length $V$ using Theorem~\ref{th:improved_algorithm} and other methods; specifically, we also consider the algorithm derived from Proposition~\ref{pr:SLDI-P-TEGs} directly, the algorithm developed in~\cite{KATS20081196}, and a linear programming solver.
The implementations were made on Matlab R2019a, and for solving the linear programs we used CPLEX's dual simplex method; the tests were executed on a PC with an Intel i7 processor at 2.20Ghz.
From the results, we can see that the most time-consuming approach is the one using Proposition~\ref{pr:SLDI-P-TEGs} directly, while the algorithm from Theorem~\ref{th:improved_algorithm} achieves the fastest computation.
The advantage becomes more evident with larger subschedule lengths: for instance, when $V = 300$, the dual simplex method takes $11.0\cdot 10^{-2}$ seconds to solve the problem, whereas the algorithm derived from Theorem~\ref{th:improved_algorithm} only $1.25 \cdot 10^{-2}$ seconds.
Such computational time reduction may have considerable impact for the solution of cyclic scheduling problems.

\begin{figure}
\centering
	\resizebox{\linewidth}{!}{
    \input{figures/comparison.tex}
    }
	\caption{Time to compute $\solSLDI$ for increasing values of $V$ using different methods.}
	\label{fig:comparison}
\end{figure}

\subsection{Considering start-up and shut-down transients}

At the start-up, the system cannot follow the periodic trajectories found in the previous section as all stations are initially empty.
Moreover, the periodic workflow must be interrupted for the system shut-down, at the end of which all stations are left empty.
In the following, we also suppose that the initial position of the robot coincides with the location of the input storage $S_0$.

To represent the complete dynamics of the processing network, from the start-up to the shut-down, we introduce additional modes of operations modeling the initial and final transients.
In particular, we add three modes for the start-up: $\winit_{\wB_1}$ corresponds to the sequence of operations
\[
	S_0 \xrightarrow{\wB} S_2, \rightarrow S_2, S_2 \xrightarrow{\wB} S_1, \rightarrow S_0,
\]
mode $\winit_{\wB_2}$ corresponds to
\[
	S_0 \xrightarrow{\wB} S_2, \rightarrow S_1, S_1 \xrightarrow{\wB} S_4, \rightarrow S_0,
\]
and mode $\winit_{\wA}$ is associated to
\[
	S_0 \xrightarrow{\wA} S_1, \rightarrow S_1, S_1 \xrightarrow{\wA} S_3.
\]
The subschedule $\winit_{\wB_1}\winit_{\wB_2}\winit_{\wA}$ represents the initial transient of the processing network, consisting of the transportation inside the system of the first two parts of type $\wB$ and the first part of type $\wA$; thus, at the end of the sequence of operations corresponding to $\winit_{\wB_1}\winit_{\wB_2}\winit_{\wA}$, a part of type $\wB$ is in station $S_2$, another part of the same type is in $S_4$, and a part of type $\wA$ is in $S_3$. 
Similarly, three modes are added for the shut-down: mode $\wfin_{\wB_1}$ corresponds to
\[
	\rightarrow S_4, S_4 \xrightarrow{\wB} S_5, \rightarrow S_2, S_2 \xrightarrow{\wB} S_1, \rightarrow S_5, S_5 \xrightarrow{\wB} S_6, \rightarrow S_1, S_1 \xrightarrow{\wB} S_4,
\]
mode $\wfin_{\wA}$ is associated to
\[
	\rightarrow S_3, S_3 \xrightarrow{\wA} S_5, \rightarrow S_5, S_5 \xrightarrow{\wA} S_6,
\]
and mode $\wfin_{\wB_2}$ corresponds to
\[
	\rightarrow S_4, S_4 \xrightarrow{\wB} S_5, \rightarrow S_5, S_5 \xrightarrow{\wB} S_6.
\]
Similarly to modes $\wA$ and $\wB$, we can derive matrices $A^0_\wZ$, $A^1_\wZ$, $B^0_\wZ$, and $B^1_\wZ$ for the additional modes $\winit_{\wB_1},\winit_{\wB_2},\winit_\wA,\wfin_{\wB_1},\wfin_\wA,\wfin_{\wB_2}$.

An example of complete schedule for the processing network is the intermittently periodic schedule
\[
	w = \winit_{\wB_1}\winit_{\wB_2}\winit_\wA (\wB \wA)^{m} \wfin_{\wB_1} \wfin_\wA \wfin_{\wB_2},
\]
where $m\in\nat$ is the number of repetitions of subschedule $\wB \wA$.
To find the admissible periods $\lambda$ of an intermittently periodic trajectory following schedule $w$, we solve the PIC-NCP on matrix
\[
	\lambda P \oplus \lambda^{-1} I \oplus C = \begin{bsmallmatrix}
		C_{\winit_{\wB_1}} & P_{\winit_{\wB_1}} & \pazocal{E} & \pazocal{E} & \pazocal{E} & \pazocal{E} & \pazocal{E} & \pazocal{E} \\
		I_{\winit_{\wB_1}} & C_{\winit_{\wB_2}} & P_{\winit_{\wB_2}} & \pazocal{E} & \pazocal{E} & \pazocal{E} & \pazocal{E} & \pazocal{E} \\
		\pazocal{E} & I_{\winit_{\wB_2}} & C_{\winit_{\wA}} & P_{\winit_{\wA}} & \pazocal{E} & \pazocal{E} & \pazocal{E} & \pazocal{E} \\
		\pazocal{E} & \pazocal{E} & I_{\winit_{\wA}} & C_{\wB} & P_\wB\oplus \lambda^{-1}I_\wA & \pazocal{E} & \pazocal{E}  & \pazocal{E} \\
		\pazocal{E} & \pazocal{E} & \pazocal{E} & I_\wB\oplus \lambda P_{\wA} & C_\wA & P_\wA & \pazocal{E}  & \pazocal{E} \\
		\pazocal{E} & \pazocal{E} & \pazocal{E} & \pazocal{E} & I_{\wA} & C_{\wfin_{\wB_1}} & P_{\wfin_{\wB_1}} & \pazocal{E} \\
		\pazocal{E} & \pazocal{E} & \pazocal{E} & \pazocal{E} & \pazocal{E} & I_{\wfin_{\wB_1}} & C_{\wfin_{\wA}} & P_{\wfin_\wA} \\
		\pazocal{E} & \pazocal{E} & \pazocal{E} & \pazocal{E} & \pazocal{E} & \pazocal{E} & I_{\wfin_{\wA}} & C_{\wfin_{\wB_2}}
	\end{bsmallmatrix}
\]
using Theorem~\ref{th:intermittent_periodic}; the outcome is that the valid periods are those in interval $[77,192]$.
Finally, a consistent intermittently periodic trajectory can be obtained from any column of $(\lambda P \oplus \lambda^{-1} I \oplus C)^*$ for $\lambda\in [77,192]$ (see Proposition~\ref{pr:nonegset_inequality}).
For instance, the trajectory displayed in Figure~\ref{fi:processing_network_trajectory} is derived through some simple manipulations from the first column of matrix $A=(77 P \oplus 77^{-1} I \oplus C)^*$; in particular, we consider $w = \winit_{\wB_1}\winit_{\wB_2}\winit_\wA \wB \wA \wB \wA \wfin_{\wB_1} \wfin_\wA \wfin_{\wB_2}$ (i.e., $m = 2)$ and impose:
\[
	\forall k\in\dint{1,5},\ x(k) = A_{\dint{12\times (k-1),12\times k},1},\quad x(6) = 77 x(4),\ x(7) = 77 x(5),
	\]\[ \forall h\in\dint{8,10},\ x(h) = 77 A_{\dint{12\times (h-3),12\times (h-2)},1},
\]
where $A_{\dint{i,j},1}$ indicates the vector containing elements $i,i+1,\ldots,j$ of the first column of matrix $A$.

\begin{figure}
	\centering
	\resizebox{1\textwidth}{!}{
%
%
\definecolor{mycolor1}{rgb}{0.85098,0.32549,0.09804}%
\definecolor{mycolor2}{rgb}{0.46667,0.67451,0.18824}%
\definecolor{mycolor3}{rgb}{0.00000,0.44706,0.74118}%
\begin{tikzpicture}

\Large

\begin{axis}[%
width=9.268in,
height=3.015in,
at={(1.555in,0.448in)},
scale only axis,
xmin=0,
xmax=350,
xlabel style={font=\color{white!15!black}},
xlabel={Time},
ymin=0.5,
ymax=6.5,
ytick={0.5,1.5,2.5,3.5,4.5,5.5,6.5},
yticklabels={{${S}_{0}$},{${S}_{1}$},{${S}_{2}$},{${S}_{3}$},{${S}_{4}$},{${S}_{5}$},{${S}_{6}$}},
ylabel style={font=\color{white!15!black}},
ylabel={Stations},
axis background/.style={fill=white},
ymajorgrids,
legend style={at={(0.03,0.97)}, anchor=north west, legend cell align=left, align=left, draw=white!15!black}
]
\addplot [color=mycolor2, line width=1.5pt]
  table[row sep=crcr]{%
0	0.5\\
4	2.5\\
54	2.5\\
57	1.5\\
58	1.5\\
62	1.5\\
67	1.5\\
72	4.5\\
76	4.5\\
78	4.5\\
88	4.5\\
91	4.5\\
101	4.5\\
102	4.5\\
105	5.5\\
112	5.5\\
119	5.5\\
125	5.5\\
128	6.5\\
};
\addlegendentry{Part $\wB$ moves}

\addplot [color=mycolor3, line width=1.5pt,dashed,dash pattern={on 10pt off 5pt}]
  table[row sep=crcr]{%
76	0.5\\
78	1.5\\
88	1.5\\
91	3.5\\
101	3.5\\
102	3.5\\
105	3.5\\
112	3.5\\
119	3.5\\
125	3.5\\
128	3.5\\
134	3.5\\
138	3.5\\
139	3.5\\
144	3.5\\
145	3.5\\
148	5.5\\
153	5.5\\
160	5.5\\
168	5.5\\
170	6.5\\
};
\addlegendentry{Part $\wA$ moves}

\addplot [color=mycolor2, line width=1.5pt, forget plot]
  table[row sep=crcr]{%
58	0.5\\
62	2.5\\
67	2.5\\
72	2.5\\
76	2.5\\
78	2.5\\
88	2.5\\
91	2.5\\
101	2.5\\
102	2.5\\
105	2.5\\
112	2.5\\
119	1.5\\
125	1.5\\
128	1.5\\
134	1.5\\
138	1.5\\
139	1.5\\
144	4.5\\
145	4.5\\
148	4.5\\
153	4.5\\
160	4.5\\
168	4.5\\
170	4.5\\
175	4.5\\
178	4.5\\
179	4.5\\
182	5.5\\
189	5.5\\
196	5.5\\
202	5.5\\
205	6.5\\
};
\addplot [color=mycolor2, line width=1.5pt, forget plot]
  table[row sep=crcr]{%
134	0.5\\
138	2.5\\
139	2.5\\
144	2.5\\
145	2.5\\
148	2.5\\
153	2.5\\
160	2.5\\
168	2.5\\
170	2.5\\
175	2.5\\
178	2.5\\
179	2.5\\
182	2.5\\
189	2.5\\
196	1.5\\
202	1.5\\
205	1.5\\
211	1.5\\
215	1.5\\
216	1.5\\
221	4.5\\
222	4.5\\
225	4.5\\
230	4.5\\
237	4.5\\
245	4.5\\
247	4.5\\
252	4.5\\
255	4.5\\
256	4.5\\
259	5.5\\
266	5.5\\
269	5.5\\
279	5.5\\
282	6.5\\
};
\addplot [color=mycolor3, line width=1.5pt, forget plot,dashed,dash pattern={on 10pt off 5pt}]
  table[row sep=crcr]{%
153	0.5\\
160	1.5\\
168	1.5\\
170	1.5\\
175	1.5\\
178	3.5\\
179	3.5\\
182	3.5\\
189	3.5\\
196	3.5\\
202	3.5\\
205	3.5\\
211	3.5\\
215	3.5\\
216	3.5\\
221	3.5\\
222	3.5\\
225	5.5\\
230	5.5\\
237	5.5\\
245	5.5\\
247	6.5\\
};
\addplot [color=mycolor2, line width=1.5pt, forget plot]
  table[row sep=crcr]{%
211	0.5\\
215	2.5\\
216	2.5\\
221	2.5\\
222	2.5\\
225	2.5\\
230	2.5\\
237	2.5\\
245	2.5\\
247	2.5\\
252	2.5\\
255	2.5\\
256	2.5\\
259	2.5\\
266	2.5\\
269	1.5\\
279	1.5\\
282	1.5\\
287	1.5\\
292	4.5\\
295	4.5\\
298	4.5\\
318	4.5\\
320	4.5\\
322	4.5\\
325	5.5\\
345	5.5\\
348	6.5\\
};
\addplot [color=mycolor3, line width=1.5pt, forget plot,dashed,dash pattern={on 10pt off 5pt}]
  table[row sep=crcr]{%
230	0.5\\
237	1.5\\
245	1.5\\
247	1.5\\
252	1.5\\
255	3.5\\
256	3.5\\
259	3.5\\
266	3.5\\
269	3.5\\
279	3.5\\
282	3.5\\
287	3.5\\
292	3.5\\
295	3.5\\
298	5.5\\
318	5.5\\
320	6.5\\
};

\addplot [color=mycolor1, loosely dotted, line width=2.5pt]
  table[row sep=crcr]{%
0	0.5\\
4	2.5\\
54	2.5\\
57	1.5\\
58	0.5\\
62	2.5\\
67	1.5\\
72	4.5\\
76	0.5\\
78	1.5\\
88	1.5\\
91	3.5\\
101	3.5\\
102	4.5\\
105	5.5\\
112	2.5\\
119	1.5\\
125	5.5\\
128	6.5\\
134	0.5\\
138	2.5\\
139	1.5\\
144	4.5\\
145	3.5\\
148	5.5\\
153	0.5\\
160	1.5\\
168	5.5\\
170	6.5\\
175	1.5\\
178	3.5\\
179	4.5\\
182	5.5\\
189	2.5\\
196	1.5\\
202	5.5\\
205	6.5\\
211	0.5\\
215	2.5\\
216	1.5\\
221	4.5\\
222	3.5\\
225	5.5\\
230	0.5\\
237	1.5\\
245	5.5\\
247	6.5\\
252	1.5\\
255	3.5\\
256	4.5\\
259	5.5\\
266	2.5\\
269	1.5\\
279	5.5\\
282	6.5\\
287	1.5\\
292	4.5\\
295	3.5\\
298	5.5\\
318	5.5\\
320	6.5\\
322	4.5\\
325	5.5\\
345	5.5\\
348	6.5\\
};
\addlegendentry{Robot moves}

\draw[dashed,very thick,stealth'-stealth'] (125,5.5) to node [above,pos=.3] {$\lambda = 77$} (202,5.5);
\end{axis}
\end{tikzpicture}%
}
	\caption{Intermittently periodic trajectory for the multi-product processing network. In this trajectory, 3 parts of type $\wA$ and 4 of type $\wB$ are processed.}\label{fi:processing_network_trajectory}
\end{figure}
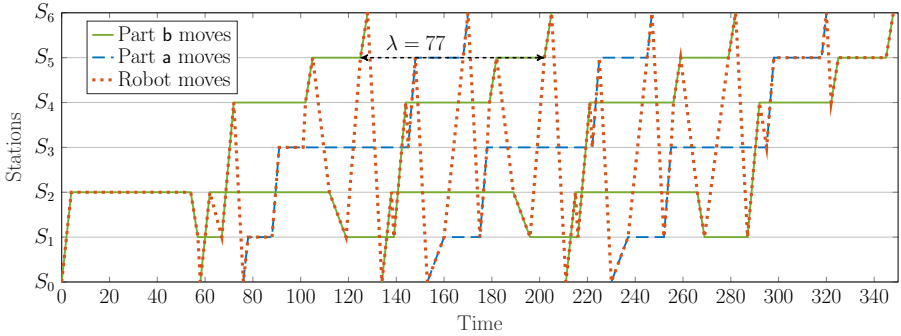

\section{Final remarks}\label{se:conclusions}

A number of results regarding the cycle time analysis in systems with time-window constraints have been presented.
Because of the generality of SLDIs, the formulas in Theorem~\ref{th:improved_algorithm} and the complexity-reduction technique of Theorem~\ref{th:intermittent_periodic} are applicable to a wide range of cyclic scheduling problems.
\zor{However,} plenty of problems of theoretical and practical relevance regarding SLDIs remain open, such as the decidability and complexity of verifying the existence of a schedule $w$ under which the SLDIs is boundedly consistent, or the development of feedback-control techniques for this class of systems.

In the following, we draw comparisons between SLDIs and other related dynamical systems to provide a broader context for this new class of systems and to outline possible research directions inspired from previous work.


SLDIs have strong connections with several other dynamical systems\zor{, with interval weighted automata being the most closely related in terms of modeling expressiveness~\cite{vspacek2010modeling,KOMENDA2020187}.
Interval weighted automata represent the natural extension of max-plus automata to the case of time-window constraints, which force the dater function to satisfy inequalities of the form}
\begin{equation}\label{eq:IWA}A^1_{w_k} \otimes x(k) \leq x(k+1) \leq B^1_{w_k} \stimes x(k).\end{equation}
Expanding the seminal work~\cite{gaubert1999modeling} of Gaubert and Mairesse, in~\cite{KOMENDA2020187} it was shown that safe P-time Petri nets, in which the number of tokens per place cannot exceed 1, are a subclass of interval weighted automata.
\zor{Comparing~\eqref{eq:dynamics} with~\eqref{eq:IWA}}, it is not difficult to see that SLDIs are an even larger class of systems compared to interval weighted automata, implying that also SLDIs can represent the dynamics of safe P-time Petri nets.

\zor{If from~\eqref{eq:IWA} we eliminate the upper bound constraints (by defining $B_{w_k}^1=\mathcal{T}$) and take the least consistent trajectory of the dater function (by substituting the left "$\leq$" sign in~\eqref{eq:IWA} by "$=$"), we get the dynamics of a max-plus automaton~\cite{gaubert1995performance}:}
\begin{equation}\label{eq:MPA}
	A^1_{w_k} \otimes x(k) = x(k+1).
\end{equation}
\zor{This shows that the behavior of any max-plus automaton corresponds to the "fastest" trajectory of specific switched max-plus linear-dual inequalities, in which $A^0_{w_k} = \pazocal{E}$ and $B^0_{w_k} = B^1_{w_k} = \mathcal{T}$.}
From this relationship one can observe that the algorithm derived in this paper for the cycle time computation in periodic schedules (Theorem~\ref{th:improved_algorithm}) generalizes~\cite[Theorem 5.2]{gaubert1999modeling}, which gives a simple formula for the cycle time of safe job shops (without upper bound constraints).
The involvement of more complex formulas featured in Theorem~\ref{th:improved_algorithm} has to be attributed to the greater number of circuits in the precedence graph when accounting for upper-bound constraints (see Figure~\ref{fi:switching_graph}; the case considered in~\cite{gaubert1999modeling} corresponds to the same type of graphs but without arcs labeled $P_i$, $\lambda P_i$, and $C_i$).

\zor{In~\cite{VANDENBOOM20061199}, van den Boom and De Schutter extended the capabilities of max-plus automata by introducing an input signal $u(k)\in\R^m$ and an input-state matrix $D_{w_k}\in\Rmax^{n\times m}$ to~\eqref{eq:MPA}; the resulting dynamical system is referred to as a switching max-plus linear system:}
\begin{equation}\label{eq:switching_max_plus}
	A^1_{w_k} \otimes x(k) \oplus D_{w_k} \otimes u(k+1) = x(k+1).
\end{equation}
\zor{Now the system is not anymore forced to evolve according to the fastest possible trajectory represented by~\eqref{eq:MPA}, but it is possible to construct a controller that, by selecting the appropriate input signal, can delay the occurrence of events with the aim of regulating the behavior of the system.
This feature makes switching max-plus linear systems, in a certain sense, more similar to SLDIs.
We explain this in more detail in the following.
Take a switching max-plus linear system and assume that the input signal is able to delay "directly" the occurrence of each event.
Mathematically, this corresponds to having $D_{w_k} = E_{\otimes}\in\Rmax^{n\times n}$ for all $k$.
Then, recalling that $\oplus$ in~\eqref{eq:switching_max_plus} is the elementwise maximization, by taking $u(k+1)$ sufficiently large (in particular, $A^1_{w_k}\otimes x(k)\leq u(k+1)$), $x(k+1)$ and $u(k+1)$ coincide, and~\eqref{eq:switching_max_plus} simplifies to}
\[
	A^1_{w_k} \otimes x(k) \oplus x(k+1) = x(k+1),
\]
\zor{which is equivalent to}
\[
	A^1_{w_k} \otimes x(k) \leq x(k+1).
\]
\zor{In this way, we have recovered a portion of the inequalities that appear in SLDIs.
Thus, observe that:
\begin{itemize}
	\item on the one hand, unlike SLDIs, switching max-plus linear systems are not able to represent upper-bound constraints,
	\item on the other hand, only systems where all events are directly controllable (i.e., delayable) can be modeled by SLDIs, whereas this does not need to be the case for switching max-plus linear systems.
\end{itemize}
It follows that the modeling expressiveness of switching max-plus linear systems and SLDIs is not comparable.
The generalization of the results of the present paper to switching dynamical systems driven by time-window constraints and where the assumption of direct controllability of all events is relaxed is left as future work.}

Regarding control approaches, it would be valuable to extend to SLDIs techniques already investigated in switching max-plus linear systems and max-plus automata: we mention model predictive control~\cite{VANDENBOOM20061199}, just-in-time control based on residuation theory~\cite{alsaba2006just}, geometric control~\cite{animobono2022model}, and supervisory control~\cite{komenda2009supervisory}.

To conclude, we emphasize that interesting connections can be drawn between the concepts of stability in discrete-time switched linear systems\footnote{\zor{Recall that discrete-time switched linear systems follow $x(k+1) = A^1_{w_k} \cdot x(k)$, where $\cdot$ indicates the standard matrix-vector multiplication.}} (see, e.g., \cite{jungers2009joint}) and bounded consistency in SLDIs.
Considering the case of "unswitched" systems (i.e., discrete-time linear systems and LDIs), the two properties are both verifiable in strongly-polynomial time (using Jury's test for the former property~\cite{aastrom2013computer}, using Theorem~\ref{th:bounded_consistency} for the latter) and are related to the spectral radius in the respective algebras.
In their switching counterparts, stability and bounded consistency are -- unsurprisingly -- interconnected with the notion of joint spectral radius in the standard and max-plus algebras.
Further investigation of this link presents an exciting avenue for future research.

\appendix

\section{\zor{Proof of Proposition~\ref{pr:SLDI-P-TEGs}}}\label{ap:Proposition5}

Let us consider SLDIs $\pazocal{S}$ under schedule $w = v^{+\infty}$, where $v = v_1v_2\ldots v_V$: for all $k\in\mathbb{N}$,
\[
	\left\{
	\begin{array}{rcccl}
	A_{v_1}^0 \otimes x(V(k-1)+1) & \leq & x(V(k-1)+1) &\leq & B_{v_1}^0 \stimes x(V(k-1)+1),\\
	A_{v_1}^1 \otimes x(V(k-1)+1) & \leq& x(V(k-1)+2) &\leq & B_{v_1}^1 \stimes x(V(k-1)+1),\\
	A_{v_2}^0 \otimes x(V(k-1)+2) & \leq& x(V(k-1)+2) &\leq & B_{v_2}^0 \stimes x(V(k-1)+2),\\
	A_{v_2}^1 \otimes x(V(k-1)+2) & \leq& x(V(k-1)+3) &\leq & B_{v_2}^1 \stimes x(V(k-1)+2),\\
	& &\vdots& &\\
	A_{v_V}^0 \otimes x(Vk) & \leq& x(Vk) &\leq & B_{v_V}^0 \stimes x(Vk),\\
	A_{v_V}^1 \otimes x(Vk) & \leq& x(Vk+1) &\leq & B_{v_V}^1 \stimes x(Vk).
	\end{array}
	\right.
\]
By defining 
\[
	\tilde{x}(k) = \begin{bmatrix}
		x(V(k-1)+1) \\ x(V(k-1)+2) \\ \vdots \\ x(Vk)
	\end{bmatrix}
\]
and using Proposition~\ref{pr:max-min}, it is possible to rewrite the SLDIs above as the following LDIs: for all $k\in\mathbb{N}$,
\[
	\begin{array}{rcl}
	A_{v}^0 \otimes \tilde{x}(k) \leq & \tilde{x}(k) &\leq B_{v}^0 \stimes \tilde{x}(k),\\
	A_{v}^1 \otimes \tilde{x}(k) \leq & \tilde{x}(k+1) &\leq B_{v}^1 \stimes \tilde{x}(k),
	\end{array}
\]
where
\[
	A_v^0 = \begin{bsmallmatrix}
		A_{v_1}^0 & \pazocal{E} & \pazocal{E} & \cdots & \pazocal{E} & \pazocal{E}\\
		A_{v_1}^1 & A_{v_2}^0 & \pazocal{E} & \cdots & \pazocal{E} & \pazocal{E}\\
		\pazocal{E} & A_{v_2}^1 & A_{v_3}^0 & \cdots & \pazocal{E} & \pazocal{E}\\
		\pazocal{E} & \pazocal{E} & A_{v_3}^1 & \cdots & \pazocal{E} & \pazocal{E}\\
		\svdots & \svdots &\svdots &\sddots &\svdots &\svdots \\
		\pazocal{E} & \pazocal{E} &\pazocal{E} &\cdots &A_{v_{V-1}}^1 & A_{v_V}^0
	\end{bsmallmatrix},\quad
	B_v^0 = \begin{bsmallmatrix}
		B_{v_1}^0 & \pazocal{T} & \pazocal{T} & \cdots & \pazocal{T} & \pazocal{T}\\
		B_{v_1}^1 & B_{v_2}^0 & \pazocal{T} & \cdots & \pazocal{T} & \pazocal{T}\\
		\pazocal{T} & B_{v_2}^1 & B_{v_3}^0 & \cdots & \pazocal{T} & \pazocal{T}\\
		\pazocal{T} & \pazocal{T} & B_{v_3}^1 & \cdots & \pazocal{T} & \pazocal{T}\\
		\svdots & \svdots &\svdots &\sddots &\svdots &\svdots \\
		\pazocal{T} & \pazocal{T} &\pazocal{T} &\cdots &B_{v_{V-1}}^1 & B_{v_V}^0
	\end{bsmallmatrix},
\]
\[
	A_v^1 = \begin{bsmallmatrix}
		\pazocal{E} & \cdots & \pazocal{E} & A_{v_V}^1\\
		\pazocal{E} & \cdots & \pazocal{E} & \pazocal{E}\\
		\svdots & & \svdots & \svdots\\
		\pazocal{E} & \cdots & \pazocal{E} & \pazocal{E} 
	\end{bsmallmatrix},\quad
	B_v^1 = \begin{bsmallmatrix}
		\pazocal{T} & \cdots & \pazocal{T} & B_{v_V}^1\\
		\pazocal{T} & \cdots & \pazocal{T} & \pazocal{T}\\
		\svdots & & \svdots & \svdots\\
		\pazocal{T} & \cdots & \pazocal{T} & \pazocal{T} 
	\end{bsmallmatrix}.
\]
Hence, since SLDIs under periodic schedule $v^{+\infty}$ are equivalent to the above LDIs, according to Theorem~\ref{th:bounded_consistency} they are boundedly consistent if and only if they admit, for some $\lambda\in\R$, a trajectory of the form
\[
	\forall k\in\nat,\quad \tilde{x}(k) = \lambda^{k-1} \tilde{x}(1);
\]
observe that such a trajectory is $v$-periodic as, in terms of dater $x$, it satisfies the following property:
\[
	\forall k\in\nat,\ \forall h\in\dint{1,V},\quad x(Vk+h) = \lambda x(V(k-1)+h).
\]

To find the set $\solSLDI$ of admissible periods for the SLDIs under schedule $v^{+\infty}$, we proceed as in Section~\ref{su:LDIs}, obtaining:
\[
	\solSLDI = \solNCP{\lambda B_v^{1\sharp}\oplus \lambda^{-1} A_v^1 \oplus A_v^0\oplus B_v^{0\sharp}}.
\]
Therefore, the set $\solSLDI$ coincides with all the $\lambda$'s such that the precedence graph $\graph(\lambda P_v \oplus \lambda^{-1} I_v \oplus C_v )$ does not contain circuits with positive weight, where
\[
	P_v = B_v^{1\sharp} = 
	\begin{bsmallmatrix}
		\pazocal{E} & \pazocal{E} & \cdots & \pazocal{E}\\
		\svdots & \svdots &  & \svdots\\
		\pazocal{E} & \pazocal{E} & \cdots & \pazocal{E}\\
		P_V & \pazocal{E} & \cdots & \pazocal{E} 
	\end{bsmallmatrix},\quad 
	I_v = A_v^1 = 
	\begin{bsmallmatrix}
		\pazocal{E} & \cdots & \pazocal{E} & I_V\\
		\pazocal{E} & \cdots & \pazocal{E} & \pazocal{E}\\
		\svdots & & \svdots & \svdots\\
		\pazocal{E} & \cdots & \pazocal{E} & \pazocal{E} 
	\end{bsmallmatrix},
\]
\[
	C_v = A_v^0\oplus B_v^{0\sharp} = 
	\begin{bsmallmatrix}
		C_1 & P_{1} & \pazocal{E} & \pazocal{E} & \pazocal{E} & \cdots & \pazocal{E} & \pazocal{E}\\ 
		I_{1} & C_{2} & P_{2} & \pazocal{E} & \pazocal{E} & \cdots & \pazocal{E} & \pazocal{E}\\ 
		\pazocal{E} & I_{2} & C_{3} & P_{3} & \pazocal{E} & \cdots & \pazocal{E} & \pazocal{E}\\ 
		\pazocal{E} & \pazocal{E} & I_{3} & C_{4} & P_{4} & \cdots & \pazocal{E} & \pazocal{E}\\ 
		\pazocal{E} & \pazocal{E} & \pazocal{E} & I_{4} & C_{5} & \cdots & \pazocal{E} & \pazocal{E}\\ 
		\svdots & \svdots & \svdots & \svdots & \svdots & \sddots & \svdots & \svdots \\ 
		\pazocal{E} & \pazocal{E} & \pazocal{E} & \pazocal{E} & \pazocal{E} & \cdots & C_{{V-1}} & P_{{V-1}}\\ 
		\pazocal{E} & \pazocal{E} & \pazocal{E} & \pazocal{E} & \pazocal{E} & \cdots & I_{{V-1}} & C_{V}
	\end{bsmallmatrix},
\]
$P_h = B_{v_h}^{1\sharp}$, $I_h = A_{v_h}^1$, and $C_h = A_{v_h}^0\oplus B_{v_h}^{0\sharp}$ for all $h\in\dint{1,V}$.

\section{Proof of Theorem~\ref{th:improved_algorithm}}\label{ap:improved_algorithm}

In this section, we prove Theorem~\ref{th:improved_algorithm}.
All symbols used in the following propositions are as in Proposition~\ref{pr:SLDI-P-TEGs} and Theorem~\ref{th:improved_algorithm}.
To prove the theorem, we first need a number of technical propositions and lemmas, which provide closed formulas for the computation of elements of the Kleene star of particular sparse matrices.
We start from recalling some properties of the Kleene star operation in the max-plus algebra.

\begin{proposition}[\cite{baccelli1992synchronization}]
Given two matrices $a,b\in\Rmax^{n\times n}$, the following properties hold:
\begin{equation}\label{eq:aux1}
	(a\oplus b)^* = a^* (b a^*)^*,
\end{equation}	
\begin{equation}\label{eq:aux2}
	a^*a^* = a^*,
\end{equation}	
\begin{equation}\label{eq:aux3}
	a(ba)^* = (ab)^*a.
\end{equation}	
\end{proposition}

Some of the following intermediate results have already been concisely proven in~\cite{zorzenon2022switched}; here we provide a more detailed proof for the first of them.

\begin{lemma}[\cite{zorzenon2022switched}]\label{le:aux1}
Let $a\in\Rmax^{n_1\times n_1}$, $b\in\Rmax^{n_1\times n_2}$, $c\in\Rmax^{n_2\times n_1}$, and $d\in\Rmax^{n_2\times n_2}$.
Then
\[
	\begin{bmatrix}
		a & b\\
		c & d
	\end{bmatrix}^* = 
	\begin{bmatrix}
		a^*(a^*bd^*ca^*)^*a^* & (a^*bd^*ca^*)^*a^*bd^*\\
		(d^*ca^*bd^*)^*d^*ca^* & d^*(d^*ca^*bd^*)^*d^*
	\end{bmatrix}
\]
\end{lemma}
\begin{proof}
The lemma is proven by applying the formulas from the latter proposition to Algorithm 2 from~\cite{hardouin2018control}, which can be used for the computation of the Kleene star.

Let $M^* = \begin{bsmallmatrix}a & b\\c & d\end{bsmallmatrix}^* = \begin{bsmallmatrix}\pazocal{M}^{11} & \pazocal{M}^{12}\\\pazocal{M}^{21} & \pazocal{M}^{22}\end{bsmallmatrix}$, where each $\pazocal{M}^{ij}$ is a matrix partitioning $M^*$ such that $\pazocal{M}^{11}\in\Rmax^{n_1\times n_1}$; applying Algorithm 2 of~\cite{hardouin2018control}, we easily get
\[
	\begin{array}{rcl}
	\pazocal{M}^{11} & = & a^*b(d\oplus ca^*b)^*ca^*\oplus a^*\\
	&\eqtop{\eqref{eq:aux1}}& a^*bd^*(ca^*bd^*)^*ca^* \oplus a^*\\
	&\eqtop{\eqref{eq:aux2}}& a^*a^*bd^*d^*(ca^*a^*bd^*d^*)^*ca^*a^*\oplus a^*a^*\\
	&\eqtop{\eqref{eq:aux3}}& a^*a^*bd^*(d^*ca^*a^*bd^*)^*d^*ca^*a^*\oplus a^*a^*\\
	&\eqtop{\eqref{eq:aux3}}& a^*(a^*bd^*d^*ca^*)^*a^*bd^*d^*ca^*a^*\oplus a^*a^*\\
	&=& a^*(a^*bd^*d^*ca^*)^+a^*\oplus a^*E_\otimes a^*\\
	&=& a^*((a^*bd^*d^*ca^*)^+\oplus E_\otimes)a^*\\
	&=& a^*(a^*bd^*d^*ca^*)^*a^*\\
	&\eqtop{\eqref{eq:aux2}}& a^*(a^*bd^*ca^*)^*a^*,
	\end{array}
\]
where $a^+ = a a^*$ and, consequently, $a^* = a^+ \oplus E_\otimes$.
The formula for matrix $\pazocal{M}^{12}$ can be obtained in a similar way starting from
\[
	\pazocal{M}^{12} = a^*b(d\oplus ca^*b)^*,
\]
which comes from Algorithm 2 of~\cite{hardouin2018control}.
Formulas for $\pazocal{M}^{21}$ and $\pazocal{M}^{22}$ are directly derived from reasoning about the symmetry of the Kleene star operation.
\end{proof}

Let us now divide matrix $\lambda P_v\oplus \lambda^{-1}I_v\oplus C_v$ into four blocks, such that the top-left block is $C_{1}\in\Rmax^{n\times n}$.
We will first focus on formulas to compute specific elements of the Kleene star of the bottom-right block of $\lambda P_v\oplus \lambda^{-1}I_v \oplus C_v$, which we will denote by $M$.
We use the following notation:
\[
	M^* = \begin{bsmallmatrix}
		C_{2} & P_{2} & \pazocal{E} & \cdots & \pazocal{E}\\
		I_{2} & C_{3} & P_{3} &  \cdots & \pazocal{E}\\
		\pazocal{E} & I_{3} & C_{4} &  \cdots & \pazocal{E}\\
		\svdots & \svdots & \svdots & \sddots & \svdots\\
		\pazocal{E} & \pazocal{E} & \pazocal{E} & \cdots & C_V
	\end{bsmallmatrix}^* 
	= 
	\begin{bsmallmatrix}
		\pazocal{M}^{22} & \pazocal{M}^{23} & \pazocal{M}^{24} & \cdots & \pazocal{M}^{2V}\\
		\pazocal{M}^{32} & \pazocal{M}^{33} & \pazocal{M}^{34} & \cdots & \pazocal{M}^{3V}\\
		\pazocal{M}^{42} & \pazocal{M}^{43} & \pazocal{M}^{44} & \cdots & \pazocal{M}^{4V}\\
		\svdots & \svdots & \svdots & \sddots & \svdots\\
		\pazocal{M}^{V2} & \pazocal{M}^{V3} & \pazocal{M}^{V4} & \cdots & \pazocal{M}^{VV}
	\end{bsmallmatrix},
\]
where each $\pazocal{M}^{ij}$ is an $n\times n$ matrix.
Note that matrix $M$ is a block tridiagonal matrix, i.e., a matrix whose blocks are all $\pazocal{E}$ except for those in the main diagonal and in the first diagonals above and below the main diagonal.
We refer to~\cite{tavakolipour2018tropical,watanabe2018min} for some work related to tridiagonal matrices in the tropical (i.e., either max-plus or min-plus) algebra.

\begin{lemma}[\cite{zorzenon2022switched}]\label{le:aux22}
$\pazocal{M}^{22} = C_2^* (\mathbb{C}_2^{P})^* C_2^*$.	
\end{lemma}

\begin{lemma}\label{le:auxVV}
$\pazocal{M}^{VV} = C_{V}^* (\mathbb{C}_{V-1}^{I})^* C_{V}^*$.	
\end{lemma}
\begin{proof}
The proof is done by induction on $V$.
For $V = 3$, the formula can be obtained directly from Lemma~\ref{le:aux1}:
\[
	\mathcal{M}^{33} = C_3^*(C_3^*I_2C_2^*P_2C_3^*)^*C_3^* = C_3^*(\mathbb{C}_2^I)^*C_3^*.
\]
Suppose now that the formula holds for $V-1$ with $V \geq 4$; we prove that it holds also for $V$.
By partitioning matrix $M$ in four blocks, such that the bottom-right block coincides with $C_V$, from Lemma~\ref{le:aux1} we get
\[
\pazocal{M}^{VV} = C_V^*\left(C_V^*
\begin{bsmallmatrix} \mathcal{E} & \cdots & \mathcal{E} & I_{V-1} \end{bsmallmatrix} 
\begin{bsmallmatrix} 
    C_2 & \cdots & \mathcal{E} & \mathcal{E}\\
    \svdots & \sddots & \svdots & \svdots\\
    \mathcal{E} & \cdots & C_{V-2} & P_{V-2}\\
    \mathcal{E} & \cdots & I_{V-2} & C_{V-1}
\end{bsmallmatrix}^*
\begin{bsmallmatrix} \mathcal{E} \\ \svdots \\ \mathcal{E} \\ P_{V-1} \end{bsmallmatrix} 
C_V^*\right)^* C_V^*;
\] 
finally, we can use the induction hypothesis to simplify the latter expression into 
\[
\mathcal{M}^{VV} = C_V^*(C_V^* I_{V-1} C_{V-1}^* (\mathbb{C}_{V-2}^I)^* C_{V-1}^* P_{V-1} C_V^*)^* C_V^* = C_V^*(\mathbb{C}_{V-1}^I)^* C_V^*.
\] 
\end{proof}

\begin{remark}\label{re:auxii}
	More in general, one can prove that, for all $i\in\dint{2,V}$,
	\[
		\pazocal{M}^{ii} = C_i^*(\mathbb{C}_{i}^P \oplus \mathbb{C}_{i-1}^I)^*C_i^*,
	\]
	considering $\mathbb{C}_V^P = \mathbb{C}_{1}^I = \pazocal{E}$.
\end{remark}

\begin{lemma}\label{le:aux2V}
$\pazocal{M}^{2V} = (\mathbb{C}^P_2)^* \mathbb{P}_2 (\mathbb{C}^P_3)^* \mathbb{P}_3 \cdots (\mathbb{C}^P_{V-1})^* \mathbb{P}_{V-1}$.	
\end{lemma}
\begin{proof}
The proof is done by induction on $V$.
For $V = 3$, the formula can be obtained directly from Lemma~\ref{le:aux1}:
\[
	\mathcal{M}^{23} = (C_2^*P_2C_3^*I_2C_2^*)^*C_2^*P_2C_3^* = (\mathbb{C}_2^P)^* \mathbb{P}_2.
\]
Suppose now that the formula holds for $V-1$ with $V \geq 4$; we prove that it holds also for $V$.
By partitioning matrix $M$ in four blocks, such that the upper-left block coincides with $C_2$, from Lemma~\ref{le:aux1} we get that $\pazocal{M}^{2V}$ is equal to the rightmost block of matrix
\begin{equation*}
	\left(C_2^* 
	\begin{bsmallmatrix}
		P_2 & \pazocal{E} & \cdots & \pazocal{E} 
	\end{bsmallmatrix}
	\begin{bsmallmatrix}
		C_3 & P_3 & \cdots & \pazocal{E}\\
		I_3 & C_4 & \cdots & \pazocal{E}\\
		\svdots & \svdots & \sddots & \svdots\\
		\pazocal{E} & \pazocal{E} & \cdots & C_V
	\end{bsmallmatrix}^*
	\begin{bsmallmatrix}
		I_2 \\ \pazocal{E} \\ \svdots \\ \pazocal{E} 
	\end{bsmallmatrix}
	C_2^*
	\right)^*
	C_2^* 
	\begin{bsmallmatrix}
		P_2 & \pazocal{E} & \cdots & \pazocal{E} 
	\end{bsmallmatrix}
	\begin{bsmallmatrix}
		C_3 & P_3 & \cdots & \pazocal{E}\\
		I_3 & C_4 & \cdots & \pazocal{E}\\
		\svdots & \svdots & \sddots & \svdots\\
		\pazocal{E} & \pazocal{E} & \cdots & C_V
	\end{bsmallmatrix}^*.
\end{equation*}
Expanding the expression using the induction hypothesis and Lemma~\ref{le:aux22}, we get
\[
	\begin{array}{rcl}
	\pazocal{M}^{2V} &=& (\mathbb{C}_2^P)^* C_2^* P_2 (\mathbb{C}^P_3)^* \mathbb{P}_3 (\mathbb{C}^P_4)^* \mathbb{P}_4 \cdots (\mathbb{C}^P_{V-1})^* \mathbb{P}_{V-1} \\
	&=& (\mathbb{C}_2^P)^* C_2^* P_2 (C_3^* P_3 \cdots I_3 C_3^*)^* C_3^*P_3 C_4^* (\mathbb{C}^P_4)^* \mathbb{P}_4 \mathbb \cdots (\mathbb{C}^P_{V-1})^* \mathbb{P}_{V-1} \\
	&\eqtop{\eqref{eq:aux2}}& (\mathbb{C}_2^P)^* C_2^* P_2 (C_3^*C_3^* P_3 \cdots I_3 C_3^*)^* C_3^* C_3^*P_3 C_4^* (\mathbb{C}^P_4)^* \mathbb{P}_4 \mathbb \cdots (\mathbb{C}^P_{V-1})^* \mathbb{P}_{V-1} \\
	&\eqtop{\eqref{eq:aux3}}& (\mathbb{C}_2^P)^* C_2^* P_2 C_3^*(C_3^* P_3 \cdots I_3 C_3^*C_3^*)^* C_3^*P_3 C_4^* (\mathbb{C}^P_4)^* \mathbb{P}_4 \mathbb \cdots (\mathbb{C}^P_{V-1})^* \mathbb{P}_{V-1} \\
	&\eqtop{\eqref{eq:aux2}}& (\mathbb{C}_2^P)^* \mathbb{P}_2(\mathbb{C}_3^P)^* \mathbb{P}_3 (\mathbb{C}^P_4)^* \mathbb{P}_4 \mathbb \cdots (\mathbb{C}^P_{V-1})^* \mathbb{P}_{V-1},
	\end{array}
\]
which is the desired formula.
\end{proof}

\begin{remark}\label{re:aux2i}
	With similar techniques, one can prove that, for all $V\geq 4$ and for all $i\in\dint{3,V-1}$,
	\[
		\pazocal{M}^{2i} = (\mathbb{C}^P_2)^* \mathbb{P}_2 (\mathbb{C}^P_3)^* \mathbb{P}_3 \cdots (\mathbb{C}^P_{i-1})^* \mathbb{P}_{i-1} (\mathbb{C}^P_{i})^*.
	\]
\end{remark}

\begin{lemma}\label{le:auxV2}
$\pazocal{M}^{V2} = (\mathbb{C}_{V-1}^{I})^* \mathbb{I}_{V-1} (\mathbb{C}_{V-2}^{I})^* \mathbb{I}_{V-2} \cdots (\mathbb{C}_{2}^{I})^* \mathbb{I}_{2}$.	
\end{lemma}
\begin{proof}
The proof is done by induction on $V$.
For $V = 3$, the formula can be obtained directly from Lemma~\ref{le:aux1}:
\[
	\mathcal{M}^{32} = (C_3^*I_2C_2^*P_2C_3^*)^*C_3^*I_2C_2^* = (\mathbb{C}_2^I)^* \mathbb{I}_2.
\]
Suppose now that the formula holds for $V-1$ with $V \geq 4$; we prove that it holds also for $V$.
By partitioning matrix $M$ in four blocks, such that the bottom-right block coincides with $C_{V}$, from Lemma~\ref{le:aux1} we get that $\pazocal{M}^{V2}$ is equal to the leftmost block of matrix
\begin{equation*}
	\left(C_V^* 
	\begin{bsmallmatrix}
		\pazocal{E} & \cdots & \pazocal{E} & I_{V-1}
	\end{bsmallmatrix}
	\begin{bsmallmatrix}
		C_2 & \cdots & \pazocal{E}\\
		\svdots & \sddots & \svdots\\
		\pazocal{E} & \cdots & C_{V-1}
	\end{bsmallmatrix}^*
	\begin{bsmallmatrix}
		\pazocal{E} \\ \svdots \\ \pazocal{E} \\ P_{V-1}
	\end{bsmallmatrix}
	C_V^*
	\right)^*
	C_V^* 
	\begin{bsmallmatrix}
		\pazocal{E} & \cdots & \pazocal{E} & I_{V-1}
	\end{bsmallmatrix}
	\begin{bsmallmatrix}
		C_2 & \cdots & \pazocal{E}\\
		\svdots & \sddots & \svdots\\
		\pazocal{E} & \cdots & C_{V-1}
	\end{bsmallmatrix}^*.
\end{equation*}
Expanding the expression using the induction hypothesis, Lemma~\ref{le:auxVV},~\eqref{eq:aux2}, and~\eqref{eq:aux3}, we can get the desired formula with computations analogous to those in the proof of Lemma~\ref{le:aux2V}.
\end{proof}

\begin{remark}\label{re:auxi2}
	With similar techniques, one can prove that, for all $V\geq 4$ and for all $i\in\dint{3,V-1}$,
	\[
		\pazocal{M}^{i2} = (\mathbb{C}^P_i)^* \mathbb{I}_{i-1} (\mathbb{C}^P_{i-1})^* \mathbb{I}_{i-2} \cdots (\mathbb{C}^P_{3})^* \mathbb{I}_{2} (\mathbb{C}^P_{2})^*.
	\]
\end{remark}

In the next proposition, we focus on the top-left $n\times n$ block of matrix $(\lambda P_v \oplus \lambda^{-1} I_v \oplus C_v)^*$; we denote this block by $N\in\Rbar^{n\times n}$.

\begin{proposition}\label{pr:aux}
$N = C_1^*(\lambda \mathbb{M}^P \oplus \lambda^{-1} \mathbb{M}^I \oplus \mathbb{M}^C)^*C_1^*$.
\end{proposition}
\begin{proof}
Partitioning matrix $\lambda P_v\oplus \lambda^{-1} I_v \oplus C_v$ into four blocks, such that the top-left block is $C_1$, Lemma~\ref{le:aux1} gives us
\[\def\arraystretch{1.5}
	\begin{array}{rcl}
N &=& C_1^*\left(C_1^*
\begin{bsmallmatrix}
P_1 & \pazocal{E} & \cdots & \pazocal{E} & \lambda^{-1} I_V
\end{bsmallmatrix}
\begin{bsmallmatrix}
	C_2 & P_2 & \cdots & \pazocal{E}& \pazocal{E}\\
	I_2 & C_3 & \cdots & \pazocal{E}& \pazocal{E}\\
	\svdots & \svdots & \sddots & \svdots & \svdots\\
	\pazocal{E} & \pazocal{E} & \cdots & C_{V-1} & P_{V-1}\\
	\pazocal{E} & \pazocal{E} & \cdots & I_{V-1} & C_V
\end{bsmallmatrix}^*
\begin{bsmallmatrix}
	I_1 \\ \pazocal{E} \\ \svdots \\ \pazocal{E} \\ \lambda P_V
\end{bsmallmatrix}
C_1^*\right)^* C_1^*\\
&=& 
C_1^*(C_1^*
(P_1 \pazocal{M}^{22} I_1 \oplus \lambda^{-1}I_V \pazocal{M}^{V2} I_1 \\
& \oplus & \lambda P_1 \pazocal{M}^{2V} P_V \oplus I_V \pazocal{M}^{VV} P_V)
C_1^*)^* C_1^*\\
&=& C_1^*(\mathbb{C}^P_1 \oplus \lambda^{-1}\mathbb{M}^I \oplus \lambda \mathbb{M}^P \oplus \mathbb{C}^I_V)^* C_1^*,
\end{array}
\]
where the latter step can be obtained from Lemmas~\ref{le:aux22}-\ref{le:auxV2},~\eqref{eq:aux2}, and~\eqref{eq:aux3}.
Finally, the commutativity of $\oplus$ concludes the proof.
\end{proof}

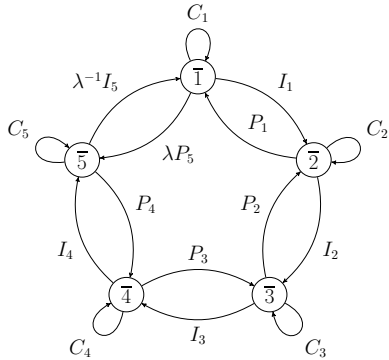
\begin{figure}
	\centering
	\resizebox{.45\linewidth}{!}{
	\begin{tikzpicture}[node distance=1.4cm and 1.4cm,>=stealth',bend angle=45,double distance=.5mm,arc/.style={->,>=stealth'},place/.style={circle,thick,minimum size=10mm,draw}]
\def \radius {3.5cm}
\Large

\foreach \t in {1,...,5}
{
\node [place,
       label=center:{$\overline{\t}$}
       ] at ({-360/5*(\t-1)+90}:\radius) (t\t) {};
\draw [arc] (t\t) to [out=-360/5*\t+360/5+90+25,in=-360/5*\t+360/5+90-25,loop] node[auto] {\textcolor{black}{$C_{\t}$}} (t\t);
}
\draw [arc] (t1) to [bend left=30] node[auto] {\textcolor{black}{$I_{1}$}} (t2);
\draw [arc] (t2) to [bend left=30] node[above right] {\textcolor{black}{$P_1$}} (t1);
\draw [arc] (t2) to [bend left=30] node[auto] {\textcolor{black}{$I_2$}} (t3);
\draw [arc] (t3) to [bend left=30] node[auto] {\textcolor{black}{$P_2$}} (t2);
\draw [arc] (t3) to [bend left=30] node[auto] {\textcolor{black}{$I_3$}} (t4);
\draw [arc] (t4) to [bend left=30] node[auto] {\textcolor{black}{$P_3$}} (t3);
\draw [arc] (t4) to [bend left=30] node[auto] {\textcolor{black}{$I_4$}} (t5);
\draw [arc] (t5) to [bend left=30] node[auto] {\textcolor{black}{$P_4$}} (t4);
\draw [arc] (t5) to [bend left=30] node[auto] {$\textcolor{black}{\lambda^{-1} I_5}$} (t1);
\draw [arc] (t1) to [bend left=30] node[auto] {$\textcolor{black}{\lambda P_5}$} (t5);

\end{tikzpicture}
	}
	\caption{Lumped-node representation of $\graph(\lambda P_v \oplus \lambda^{-1} I_v \oplus C_v)$ when $V=5$.
	}
	\label{fi:switching_graph}
\end{figure}

After the technical propositions and lemmas above, we are finally ready to introduce the graphical part of the proof.
First of all, it is useful to visualize the structure of $\graph(\lambda P_v \oplus \lambda^{-1} I_v \oplus C_v)$; the lumped-node representation of Figure~\ref{fi:switching_graph} illustrates it in the case $V = 5$.
In this simplified representation, $\overline{k}$ indicates the set of nodes $\dint{(k-1)n+1,kn}$ for any $k\in\dint{1,5}$, and the matrix, say $Y$, associated to an arc from $\overline{k}_1$ to $\overline{k}_2$ indicates that an arc in $\graph(\lambda P_v \oplus \lambda^{-1} I_v \oplus C_v)$ from node $(k_1-1)n+j$ to node $(k_2-1)n+i$ exists iff $Y_{ij}\neq -\infty$, and that its weight is equal to $Y_{ij}$, for all $i,j\in\dint{1,n}$.

Recall that, for a square matrix $M$ such that $M^* = \begin{bsmallmatrix}\mathcal{M}^{11}& \mathcal{M}^{12}\\\mathcal{M}^{21} & \mathcal{M}^{22}\end{bsmallmatrix}$, where $\mathcal{M}^{ii}$ has dimension $n_i\times n_i$, $\graph(M)$ does not contain circuits with positive weight visiting any node $j\in\dint{1,n_1}$ (respectively, $j\in\dint{n_1+1,n_2}$) iff $\graph(\mathcal{M}^{11}) \in\nonegset$ (respectively, $\graph(\mathcal{M}^{22})\in\nonegset$).
We partition the set of circuits of $\graph(\lambda P_v \oplus \lambda^{-1} I_v \oplus C_v)$ in $V$ subsets as follows: (1) circuits visiting at least one node in $\overline{1}$, (2) circuits that do not visit any node in $\overline{1}$, but that visit at least one node in $\overline{2}$, (3) circuits that do not visit any node in $\overline{1}$ or $\overline{2}$, but that visit at least one node in $\overline{3}$, \dots, ($V$) circuits that do not visit any node in $\overline{1},\overline{2},\ldots,\overline{V-1}$, but that visit at least one node in $\overline{V}$.
From Proposition~\ref{pr:aux}, the maximal weights of circuits in the first subset correspond to elements in the diagonal of $N = C_1^*(\lambda \mathbb{M}^P \oplus \lambda^{-1} \mathbb{M}^I \oplus \mathbb{M}^C)^*C_1^*$; thus these are non-positive iff $\graph(C_1)\in\nonegset$ and $\lambda\in\solNCP{\lambda \mathbb{M}^P\oplus\lambda^{-1}\mathbb{M}^I\oplus\mathbb{M}^C}$.
Let $k\in\dint{2,V-1}$.
From Lemma~\ref{le:aux22}, the maximal weights of circuits in the $k$\textsuperscript{th} subset correspond to elements in the diagonal of $C_k^*(\mathbb{C}^P_k)^* C_k^*$; thus, these are non-positive iff $\graph(C_k)\in\nonegset$ and $\graph(\mathbb{C}^P_k)\in\nonegset$.
As for circuits in the $V$\textsuperscript{th} subset, their maximal weights come from the diagonal elements of $C_{V}^*$, which are finite iff $\graph(C_{V})\in\nonegset$.

\section{\zor{Proof of Proposition~\ref{pr:intermittent_periodic}}}\label{ap:Proposition7}

Consider SLDIs~\eqref{eq:dynamics} under intermittently periodic schedule $w = u_0 v_1^{m_1} u_1 v_2^{m_2} u_2 \cdots v_q^{m_q} u_q$, where $U_q > 0$ (the case in which $U_q=0$ is analogous but must be considered separately): for all $k\in\dint{1,U_0}$,\footnote{In this section, $u_{hk}$ and $v_{hk}$ denote the $k$\textsuperscript{th} mode in subschedules $u_h$ and $v_h$, respectively.}
\begin{equation}\label{eq:periodic_1}
	\begin{array}{rcl}
		A^0_{u_{0k}} \otimes x(k) \leq & x(k) & \leq B^0_{u_{0k}}\stimes x(k)\\
		A^1_{u_{0k}} \otimes x(k) \leq & x(k+1) & \leq B^1_{u_{0k}}\stimes x(k),
	\end{array}
\end{equation}
for all $h\in\dint{1,q}$, $j\in\dint{0,m_h-1}$, $k\in\dint{1,V_h}$,
\begin{equation}\label{eq:periodic}
	\begin{array}{rl}
		A^0_{v_{hk}} \otimes x(K_{h-1}+jV_h+k) &\leq  x(K_{h-1}+jV_h+k) \\ & \leq B^0_{v_{hk}}\stimes x(K_{h-1}+jV_h+k)\\
		A^1_{v_{hk}} \otimes x(K_{h-1}+jV_h+k) &\leq x(K_{h-1}+jV_h+k+1)\\ & \leq B^1_{v_{hk}}\stimes x(K_{h-1}+jV_h+k),
	\end{array}
\end{equation}
for all $h\in\dint{1,q}$, $k\in\dint{1,U_h}$,
\begin{equation}\label{eq:periodic_2}
	\begin{array}{rl}
		A^0_{u_{hk}} \otimes x(K_{h-1}+m_hV_h+k) & \leq x(K_{h-1}+m_hV_h+k)\\ & \leq B^0_{u_{hk}}\stimes x(K_{h-1}+m_hV_h+k),
	\end{array}
\end{equation}
and for all $h\in\dint{1,q}$, $k\in\dint{1,U_h}$ with $(h,k)\neq(q,U_q)$,
\begin{equation}\label{eq:periodic_3}
	\begin{array}{rl}
		A^1_{u_{hk}} \otimes x(K_{h-1}+m_hV_h+k) & \leq x(K_{h-1}+m_hV_h+k+1) \\ & \leq B^1_{u_{hk}}\stimes x(K_{h-1}+m_hV_h+k).
	\end{array}
\end{equation}
To search intermittently periodic trajectories, we substitute~\eqref{eq:intermittently_periodic} in~\eqref{eq:periodic}, obtaining:
for all $h\in\dint{1,q}$, $j\in\dint{0,m_h-1}$, $k\in\dint{1,V_h}$,
\begin{equation}\label{eq:intermittently_aux1}
	\begin{array}{rcl}
		A^0_{v_{hk}} \otimes \cancel{\lambda_h^{j}}x(K_{h-1}+k) \leq & \cancel{\lambda_h^{j}}x(K_{h-1}+k) & \leq B^0_{v_{hk}}\stimes \cancel{\lambda_h^{j}}x(K_{h-1}+k),
	\end{array}
\end{equation}
for all $h\in\dint{1,q}$, $j\in\dint{0,m_h-1}$, $k\in\dint{1,V_h-1}$,
\begin{equation}\label{eq:intermittently_aux2}
	\begin{array}{rl}
		A^1_{v_{hk}} \otimes \cancel{\lambda_h^{j}}x(K_{h-1}+k) &\leq \cancel{\lambda_h^{j}}x(K_{h-1}+k+1)\\ & \leq B^1_{v_{hk}}\stimes \cancel{\lambda_h^{j}}x(K_{h-1}+k),
	\end{array}
\end{equation}
for all $h\in\dint{1,q}$, $j\in\dint{0,m_h-2}$ ($k = V_h$),
\begin{equation}\label{eq:intermittently_aux3}
	\begin{array}{rcl}
		A^1_{v_{hV_h}} \otimes \cancel{\lambda_h^{j}}x(K_{h-1}+V_h) \leq& \cancel{\lambda_h^{j}}\lambda_hx(K_{h-1}+1) & \leq B^1_{v_{hV_h}}\stimes \cancel{\lambda_h^{j}}x(K_{h-1}+V_h),
	\end{array}
\end{equation}
and, for all $h\in\dint{1,q}$ ($j = m_h-1$, $k = V_h$),
\begin{equation}\label{eq:intermittently_aux4}
	\begin{array}{rl}
		A^1_{v_{hV_h}} \otimes \lambda_h^{m_h-1}x(K_{h-1}+V_h) &\leq x(K_{h-1}+m_h V_h +1)\\ & \leq B^1_{v_{hV_h}}\stimes \lambda_h^{m_h-1}x(K_{h-1}+V_h).
	\end{array}
\end{equation}
To cancel out the $\lambda_h^{m_h-1}$'s from~\eqref{eq:intermittently_aux4}, we perform a change of variables.
Let $\xi:\dint{1,U_0 + \left(\sum_{h=1}^{q}V_{h}+U_h\right)}\rightarrow \R^n$ be defined by: 
for all $h\in\dint{0,q-1}$, $k\in\dint{\left(\sum_{j=1}^{h}U_{j-1}+V_j\right)+1,\sum_{j=1}^{h+1}U_{j-1}+V_j}$,
\[
	\xi(k) = \left(\bigotimes_{j=1}^h \lambda_j^{-(m_j - 1)}\right) x\left(\left(\sum_{j=1}^h (m_j-1)V_j\right) + k\right),
\]
and for all $k\in\dint{\left(\sum_{j=1}^{q}U_{j-1}+V_j\right)+1,\left(\sum_{j=1}^{q}U_{j-1}+V_j\right) + U_q}$,
\[
	\xi(k) = \left(\bigotimes_{j=1}^q \lambda_j^{-(m_j - 1)}\right) x\left(\left(\sum_{j=1}^q (m_j-1)V_j\right) + k\right).
\]

By substituting $\xi$ into inequalities~\eqref{eq:periodic_1},~(\ref{eq:periodic_2}-\ref{eq:intermittently_aux4}) and simplifying whenever possible the terms in $\lambda_j$, we obtain that the above inequalities admit a solution in $x$ and $\lambda_1,\ldots,\lambda_q$ if and only if the following inequalities admit a solution in $\xi$ and $\lambda_1,\ldots,\lambda_q$: for all $h\in\dint{0,q}$, $k\in\dint{1,U_h}$,
\begin{equation}\label{eq:aux_init}
	\begin{array}{rl}
	A^0_{u_{hk}} \otimes \xi\left(\left(\sum_{j=1}^h U_{j-1}+V_j\right)+k\right) & \leq \xi\left(\left(\sum_{j=1}^h U_{j-1}+V_j\right)+k\right) \\
	& \leq B^0_{u_{hk}} \stimes \xi\left(\left(\sum_{j=1}^h U_{j-1}+V_j\right)+k\right),
	\end{array}
\end{equation}
for all $h\in\dint{0,q}$, $k\in\dint{1,U_h}$, $(h,k)\neq(q,U_q)$,
\begin{equation}
	\begin{array}{rl}
	A^1_{u_{hk}} \otimes \xi\left(\left(\sum_{j=1}^h U_{j-1}+V_j\right)+k\right) & \leq \xi\left(\left(\sum_{j=1}^h U_{j-1}+V_j\right)+k+1\right) \\
	& \leq B^1_{u_{hk}} \stimes \xi\left(\left(\sum_{j=1}^h U_{j-1}+V_j\right)+k\right),
	\end{array}
\end{equation}
for all $h\in\dint{1,q}$, $k\in\dint{1,V_h}$,
\begin{equation}
	\begin{array}{r}
	A^0_{v_{hk}} \otimes \xi\left(\left(U_0+\sum_{j=1}^{h-1} V_{j}+U_j\right)+k\right) \leq \xi\left(\left(U_0+\sum_{j=1}^{h-1} V_{j}+U_j\right)+k\right) \\
	\leq B^0_{v_{hk}} \stimes \xi\left(\left(U_0+\sum_{j=1}^{h-1} V_{j}+U_j\right)+k\right),\\
	A^1_{v_{hk}} \otimes \xi\left(\left(U_0+\sum_{j=1}^{h-1} V_{j}+U_j\right)+k\right) \leq \xi\left(\left(U_0+\sum_{j=1}^{h-1} V_{j}+U_j\right)+k+1\right) \\
	\leq B^1_{v_{hk}} \stimes \xi\left(\left(U_0+\sum_{j=1}^{h-1} V_{j}+U_j\right)+k\right),\\
	\end{array}
\end{equation}
and for all $h\in\dint{1,q}$ ($k=V_h$),
\begin{equation}\label{eq:aux_fina}
	\begin{array}{r}
	A^1_{v_{hV_h}} \otimes \xi\left(\left(U_0+\sum_{j=1}^{h-1} V_{j}+U_j\right)+V_h\right) \leq \lambda_h\xi\left(\left(U_0+\sum_{j=1}^{h-1} V_{j}+U_j\right)+1\right) \\
	\leq B^1_{v_{hV_h}} \stimes \xi\left(\left(U_0+\sum_{j=1}^{h-1} V_{j}+U_j\right)+V_h\right).
	\end{array}
\end{equation}
The final step in the proof is to rewrite the inequalities~(\ref{eq:aux_init}-\ref{eq:aux_fina}) in terms of the extended vector
\[
	\tilde{\xi} = \begin{bmatrix}
		\xi(1)\\\svdots\\\xi\left(U_0 + \left(\sum_{h=1}^q V_{h}+U_h\right)\right)
	\end{bmatrix}
\]
and use Proposition~\ref{pr:max-min} to get the max-plus-linear inequality
\begin{equation}\label{eq:final_aux}
	A(\lambda_1,\ldots,\lambda_q) \otimes \tilde{\xi} \leq \tilde{\xi},
\end{equation}
where $A(\lambda_1,\ldots,\lambda_q)$ is a matrix with $(U_0 + \left(\sum_{h=1}^q V_{h}+U_h\right))n$ rows and columns.
Since only the first power of $\lambda_1,\ldots,\lambda_q$ shows up in~(\ref{eq:aux_init}-\ref{eq:aux_fina}), matrix $A(\lambda_1,\ldots,\lambda_q)$ has the desired form: $A(\lambda_1,\ldots,\lambda_q) = \bigoplus_{h=1}^q (\lambda_h P_h\oplus \lambda_h^{-1}I_h)\oplus C$.
Therefore, because of Proposition~\ref{pr:nonegset_inequality}, \eqref{eq:final_aux} translates into an MPIC-NCP.

\section{Proof of Theorem~\ref{th:intermittent_periodic}}\label{ap:intermittent_periodic_improved}

\begin{figure}
	\centering
	\resizebox{1\linewidth}{!}{
%
%

\begin{tikzpicture}[node distance=3cm,>=stealth',bend angle=45,double distance=.5mm,arc/.style={->,>=stealth'},place/.style={circle,thick,minimum size=10mm,draw},on grid]
\def \radius {3cm}
\Large

\node [place,label=center:{$\overline{2}$}] at ({-360/5*(2-1)-90+360/10}:\radius) (t2) {};
\draw [arc] (t2) to [out=-360/5*2+360/5-90+360/10+25,in=-360/5*2+360/5-90+360/10-25,loop] node[auto] {\textcolor{black}{$C_{\wP_1}$}} (t2);

\node [place,label=center:{$\overline{3}$}] at ({-360/5*(3-1)-90+360/10}:\radius) (t3) {};
\draw [arc] (t3) to [out=-360/5*3+360/5-90+360/10+25,in=-360/5*3+360/5-90+360/10-25,loop] node[auto] {\textcolor{black}{$C_{\wP_1}$}} (t3);

\node [place,label=center:{$\overline{4}$}] at ({-360/5*(4-1)-90+360/10}:\radius) (t4) {};
\draw [arc] (t4) to [out=-360/5*4+360/5-90+360/10+25,in=-360/5*4+360/5-90+360/10-25,loop] node[auto] {\textcolor{black}{$C_{\wP_3}$}} (t4);

\node [place,label=center:{$\overline{5}$}] at ({-360/5*(5-1)-90+360/10}:\radius) (t5) {};
\draw [arc] (t5) to [out=-360/5*5+360/5-90+360/10+25,in=-360/5*5+360/5-90+360/10-25,loop] node[auto] {\textcolor{black}{$C_{\wP_2}$}} (t5);

\node [place,label=center:{$\overline{6}$}] at ({-360/5*(6-1)-90+360/10}:\radius) (t6) {};
\draw [arc] (t6) to [out=-360/5*6+360/5-90+360/10+25,in=-360/5*6+360/5-90+360/10-25,loop] node[auto] {\textcolor{black}{$C_{\wP_4}$}} (t6);

\draw [arc] (t2) to [bend left=30] node[above right] {\textcolor{black}{$I_{\wP_1}$}} (t3);
\draw [arc] (t3) to [bend left=30] node[above right] {\textcolor{black}{$P_{\wP_1}$}} (t2);
\draw [arc] (t3) to [bend left=30] node[auto] {\textcolor{black}{$I_{\wP_1}$}} (t4);
\draw [arc] (t4) to [bend left=30] node[auto] {\textcolor{black}{$P_{\wP_1}$}} (t3);
\draw [arc] (t4) to [bend left=30] node[auto] {\textcolor{black}{$I_{\wP_3}$}} (t5);
\draw [arc] (t5) to [bend left=30] node[auto] {\textcolor{black}{$P_{\wP_3}$}} (t4);
\draw [arc] (t5) to [bend left=30] node[above left] {\textcolor{black}{$I_{\wP_2}$}} (t6);
\draw [arc] (t6) to [bend left=30] node[auto] {\textcolor{black}{$P_{\wP_2}$}} (t5);
\draw [arc] (t6) to [bend left=30] node[auto] {$\textcolor{black}{\lambda_1^{-1} I_{\wP_4}}$} (t2);
\draw [arc] (t2) to [bend left=30] node[auto] {$\textcolor{black}{\lambda_1 P_{\wP_4}}$} (t6);

\node [place,label=center:{$\overline{1}$},left=of t2] (t1) {}; 
\draw [arc] (t1) to [bend left=15] node[auto] {$I_\winit$} (t2);
\draw [arc] (t2) to [bend left=15] node[auto] {$P_\winit$} (t1);
\draw [arc] (t1) to [out=90+25,in=90-25,loop] node[auto] {$C_{\winit}$} (t1);

\node [place,label=center:{$\overline{7}$},right=of t6] (t7) {}; 
\draw [arc] (t6) to [bend left=15] node[auto] {$I_{\wP_4}$} (t7);
\draw [arc] (t7) to [bend left=15] node[auto] {$P_{\wP_4}$} (t6);
\draw [arc] (t7) to [out=90+25,in=90-25,loop] node[auto] {$C_{\wP_1}$} (t7);

\node [place,label=center:{$\overline{8}$},right=of t7] (t8) {}; 
\draw [arc] (t7) to [bend left=15] node[auto] {$I_{\wP_1}$} (t8);
\draw [arc] (t8) to [bend left=15] node[auto] {$P_{\wP_1}$} (t7);
\draw [arc] (t8) to [out=90+25,in=90-25,loop] node[auto] {$C_{\wP_3}$} (t8);

\node [place,label=center:{$\overline{9}$},right=of t8] (t9) {}; 
\draw [arc] (t8) to [bend left=15] node[auto] {$I_{\wP_3}$} (t9);
\draw [arc] (t9) to [bend left=15] node[auto] {$P_{\wP_3}$} (t8);
\draw [arc] (t9) to [out=90+25,in=90-25,loop] node[auto] {$C_{\wP_2}$} (t9);

\node [place,label=center:{$\overline{10}$},right=of t9] (t10) {}; 
\draw [arc] (t9) to [bend left=15] node[auto] {$I_{\wP_2}$} (t10);
\draw [arc] (t10) to [bend left=15] node[auto] {$P_{\wP_2}$} (t9);
\draw [arc] (t10) to [out=90+25,in=90-25,loop] node[auto] {$C_{\wP_4}$} (t10);

\node [place,label=center:{$\overline{11}$},right=of t10] (t11) {}; 
\draw [arc] (t11) to [out=-360/5*2+360/5-90+360/10+25,in=-360/5*2+360/5-90+360/10-25,loop] node[auto] {\textcolor{black}{$C_{\wP_2}$}} (t11);
\draw [arc] (t10) to [bend left=15] node[auto] {$I_{\wP_4}$} (t11);
\draw [arc] (t11) to [bend left=15] node[auto] {$P_{\wP_4}$} (t10);

\node [place,label=center:{$\overline{12}$}] at ($({-360/5*(3-1)-90+360/10}:\radius)+(t11)+(0.5878*\radius,0.8090*\radius)$) (t12) {};
\draw [arc] (t12) to [out=-360/5*3+360/5-90+360/10+25,in=-360/5*3+360/5-90+360/10-25,loop] node[auto] {\textcolor{black}{$C_{\wP_4}$}} (t12);
\draw [arc] (t11) to [bend left=30] node[above right] {$I_{\wP_2}$} (t12);
\draw [arc] (t12) to [bend left=30] node[auto] {$P_{\wP_2}$} (t11);

\node [place,label=center:{$\overline{13}$}] at ($({-360/5*(4-1)-90+360/10}:\radius)+(t11)+(0.5878*\radius,0.8090*\radius)$) (t13) {};
\draw [arc] (t13) to [out=-360/5*4+360/5-90+360/10+25,in=-360/5*4+360/5-90+360/10-25,loop] node[auto] {\textcolor{black}{$C_{\wP_1}$}} (t13);
\draw [arc] (t12) to [bend left=30] node[auto] {$I_{\wP_4}$} (t13);
\draw [arc] (t13) to [bend left=30] node[auto] {$P_{\wP_4}$} (t12);

\node [place,label=center:{$\overline{14}$}] at ($({-360/5*(5-1)-90+360/10}:\radius)+(t11)+(0.5878*\radius,0.8090*\radius)$) (t14) {};
\draw [arc] (t14) to [out=-360/5*5+360/5-90+360/10+25,in=-360/5*5+360/5-90+360/10-25,loop] node[auto] {\textcolor{black}{$C_{\wP_3}$}} (t14);
\draw [arc] (t13) to [bend left=30] node[auto] {$I_{\wP_1}$} (t14);
\draw [arc] (t14) to [bend left=30] node[auto] {$P_{\wP_1}$} (t13);

\node [place,label=center:{$\overline{15}$}] at ($({-360/5*(6-1)-90+360/10}:\radius)+(t11)+(0.5878*\radius,0.8090*\radius)$) (t15) {};
\draw [arc] (t15) to [out=-360/5*6+360/5-90+360/10+25,in=-360/5*6+360/5-90+360/10-25,loop] node[auto] {\textcolor{black}{$C_{\wP_3}$}} (t15);
\draw [arc] (t14) to [bend left=30] node[auto] {$I_{\wP_3}$} (t15);
\draw [arc] (t15) to [bend left=30] node[auto] {$P_{\wP_3}$} (t14);
%
%

\draw [arc] (t15) to [bend left=30] node[auto] {$\lambda_2^{-1} I_{\wP_3}$} (t11);
\draw [arc] (t11) to [bend left=30] node[auto] {$\lambda_2 P_{\wP_3}$} (t15);

\end{tikzpicture}

	}
	\caption{Lumped-node representation of $\graph(A(\lambda_1,\lambda_2))$ from~\eqref{eq:intermittent_periodic_matrix}.
	}
	\label{fi:intermittent_graph}
\end{figure}
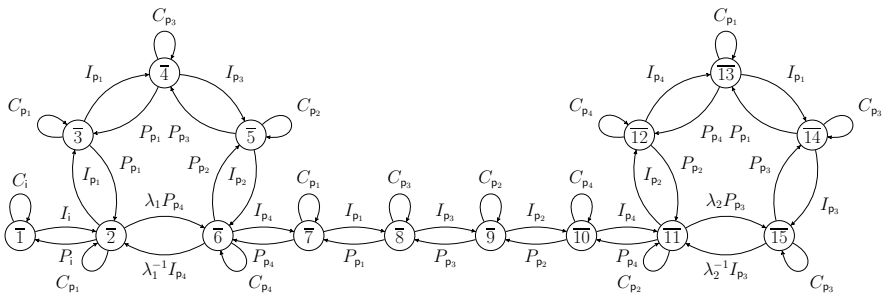

For notational simplicity, we show how to prove the theorem for the matrix in~\eqref{eq:intermittent_periodic_matrix}; the extension to the general case is simple but cumbersome.

\newcommand{\DrawHorLine}[3][]{%
    \tikz[overlay,remember picture]{%
    \draw[#1] ($({#2})+(-0.5em,-0.3em)$) -- ($({#3})+(0.3em,-0.3em)$);}%
}%
\newcommand{\DrawVerLine}[3][]{%
    \tikz[overlay,remember picture]{%
    \draw[#1] ($({#2})+(-0.3em,0.7em)$) -- ($({#3})+(-0.3em,-0.3em)$);}%
}%

First of all, it is useful to draw the parametric precedence graph $\graph(A(\lambda_1,\lambda_2))$ using the lumped-node representation (which was presented in Appendix~\ref{ap:improved_algorithm}); the graph is shown in Figure~\ref{fi:intermittent_graph}.
We recall that the objective is to describe in a more compact way the set $\solNCP{A(\lambda_1,\lambda_2)}$ of all parameters $(\lambda_1,\lambda_2)$ for which $\graph(A(\lambda_1,\lambda_2))$ does not contain circuits with positive weight.
Before illustrating the proof, it is convenient to relabel the nodes of $\graph(A(\lambda_1,\lambda_2))$ as in Figure~\ref{fi:intermittent_graph_relabeled}; the graph in Figure~\ref{fi:intermittent_graph_relabeled} therefore corresponds to the precedence graph $\graph(\overline{A}(\lambda_1,\lambda_2))$, where $\overline{A}$ is obtained permuting rows and columns of $A(\lambda_1,\lambda_2)$ as follows:
\begin{equation}\label{eq:intermittent_periodic_matrix_relabeled}
	\scriptsize \overline{A}(\lambda_1,\lambda_2) = 
	\begin{bsmallmatrix}
		C_{\wP_1} & \pazocal{E}  & \tikzmark{topleft} I_{\winit} & P_{\wP_1} & \pazocal{E} & \pazocal{E} & \lambda_1^{-1}I_{\wP_4} & \pazocal{E} & \pazocal{E} & \pazocal{E} & \pazocal{E} & \pazocal{E} & \pazocal{E} & \pazocal{E} & \pazocal{E} \\
		\tikzmark{toplefth} \pazocal{E} & C_{\wP_2} & \pazocal{E} & \pazocal{E} & \pazocal{E} & \pazocal{E} & \pazocal{E} & \pazocal{E} & \pazocal{E} & \pazocal{E} & I_{\wP_4} & P_{\wP_2} & \pazocal{E} & \pazocal{E} & \lambda_2^{-1}I_{\wP_3} \tikzmark{toprighth} \\
		P_{\winit} & \pazocal{E} & C_\winit & \pazocal{E} & \pazocal{E} & \pazocal{E} & \pazocal{E} & \pazocal{E} & \pazocal{E} & \pazocal{E} & \pazocal{E} & \pazocal{E} & \pazocal{E} & \pazocal{E} & \pazocal{E} \\
		I_{\wP_1} & \pazocal{E} & \pazocal{E} & C_{\wP_1} & P_{\wP_1} & \pazocal{E} & \pazocal{E} & \pazocal{E} & \pazocal{E} & \pazocal{E} & \pazocal{E} & \pazocal{E} & \pazocal{E} & \pazocal{E} & \pazocal{E} \\
		\pazocal{E} & \pazocal{E} & \pazocal{E} & I_{\wP_1} & C_{\wP_3} & P_{\wP_3} & \pazocal{E} & \pazocal{E} & \pazocal{E} & \pazocal{E} & \pazocal{E} & \pazocal{E} & \pazocal{E} & \pazocal{E} & \pazocal{E} \\
		\pazocal{E} & \pazocal{E} & \pazocal{E} & \pazocal{E} & I_{\wP_3} & C_{\wP_2} & P_{\wP_2} & \pazocal{E} & \pazocal{E} & \pazocal{E} & \pazocal{E} & \pazocal{E} & \pazocal{E} & \pazocal{E} & \pazocal{E} \\
		\lambda_1 P_{\wP_4} & \pazocal{E} & \pazocal{E} & \pazocal{E} & \pazocal{E} & I_{\wP_2} & C_{\wP_4} & P_{\wP_4} & \pazocal{E} & \pazocal{E} & \pazocal{E} & \pazocal{E} & \pazocal{E} & \pazocal{E} & \pazocal{E} \\
		\pazocal{E} & \pazocal{E} & \pazocal{E} & \pazocal{E} & \pazocal{E} & \pazocal{E} & I_{\wP_4} & C_{\wP_1} & P_{\wP_1} & \pazocal{E} & \pazocal{E} & \pazocal{E} & \pazocal{E} & \pazocal{E} & \pazocal{E} \\
		\pazocal{E} & \pazocal{E} & \pazocal{E} & \pazocal{E} & \pazocal{E} & \pazocal{E} & \pazocal{E} & I_{\wP_1} & C_{\wP_3} & P_{\wP_3} & \pazocal{E} & \pazocal{E} & \pazocal{E} & \pazocal{E} & \pazocal{E} \\
		\pazocal{E} & \pazocal{E} & \pazocal{E} & \pazocal{E} & \pazocal{E} & \pazocal{E} & \pazocal{E} & \pazocal{E} & I_{\wP_3} & C_{\wP_2} & P_{\wP_2} & \pazocal{E} & \pazocal{E} & \pazocal{E} & \pazocal{E} \\
		\pazocal{E} & P_{\wP_4} & \pazocal{E} & \pazocal{E} & \pazocal{E} & \pazocal{E} & \pazocal{E} & \pazocal{E} & \pazocal{E} & I_{\wP_2} & C_{\wP_4} & \pazocal{E} & \pazocal{E} & \pazocal{E} & \pazocal{E} \\
		\pazocal{E} & I_{\wP_2} & \pazocal{E} & \pazocal{E} & \pazocal{E} & \pazocal{E} & \pazocal{E} & \pazocal{E} & \pazocal{E} & \pazocal{E} & \pazocal{E} & C_{\wP_4} & P_{\wP_4} & \pazocal{E} & \pazocal{E} \\
		\pazocal{E} & \pazocal{E} & \pazocal{E} & \pazocal{E} & \pazocal{E} & \pazocal{E} & \pazocal{E} & \pazocal{E} & \pazocal{E} & \pazocal{E} & \pazocal{E} & I_{\wP_4} & C_{\wP_1} & P_{\wP_1} & \pazocal{E} \\
		\pazocal{E} & \pazocal{E} & \tikzmark{bottomleft} \pazocal{E} & \pazocal{E} & \pazocal{E} & \pazocal{E} & \pazocal{E} & \pazocal{E} & \pazocal{E} & \pazocal{E} & \pazocal{E} & \pazocal{E} & I_{\wP_1} & C_{\wP_3} & P_{\wP_3} \\
		\pazocal{E} & \lambda_2 P_{\wP_3} & \tikzmark{bottomleft} \pazocal{E} & \pazocal{E} & \pazocal{E} & \pazocal{E} & \pazocal{E} & \pazocal{E} & \pazocal{E} & \pazocal{E} & \pazocal{E} & \pazocal{E} & \pazocal{E} & I_{\wP_3} & C_{\wP_3} 
	\end{bsmallmatrix}.
\end{equation}
\DrawVerLine{topleft}{bottomleft}%
\DrawHorLine{toplefth}{toprighth}%

Clearly, $\graph(A(\lambda_1,\lambda_2))\in\nonegset$ if and only if $\graph(\overline{A}(\lambda_1,\lambda_2))\in\nonegset$.
We partition matrix $\overline{A}(\lambda_1,\lambda_2)$ in four blocks, as shown in~\eqref{eq:intermittent_periodic_matrix_relabeled}, and we use the notation
\[
	\overline{A}(\lambda_1,\lambda_2)^* = \begin{bmatrix}
		a & b \\ c & d
	\end{bmatrix}^* = \begin{bmatrix}
		\pazocal{M}^{11} & \pazocal{M}^{12}\\ \pazocal{M}^{21} & \pazocal{M}^{22}
	\end{bmatrix},
\]
where $a$ and $\pazocal{M}^{11}$ are matrices of dimension $10\times 10$ (in general, $qn\times qn$).
Note that only blocks $b$ and $c$ depend on $\lambda_1$ and $\lambda_2$, whereas $a$ and $d$ are numerical matrices.

Then, we partition the set of circuits in $\graph(\overline{A}(\lambda_1,\lambda_2))$ in two subsets:
\begin{enumerate}
	\item the set of all circuits that do not pass through any node in $\overline{1}\cup\overline{2}$,
	\item the set containing all remaining circuits.
\end{enumerate}
It is easy to see that there are no circuits with positive weight among those in the first subset if and only if $\graph(d)\in\nonegset$; due to the block tridiagonal structure of matrix $d$, verifying whether $\graph(d)$ contains circuits with positive weight can be done in linear time in the length of $u_0,\ldots,u_q$, $v_1,\ldots,v_q$ (i.e., the lengths of subschedules $\winit$, $\wP_1\wP_1\wP_3\wP_2\wP_4$, and $\wP_1\wP_3\wP_2\wP_4$) using the methods seen in Appendix~\ref{ap:improved_algorithm} (or, equivalently, in~\cite[Theorem 3]{zorzenon2022switched}).

As for the second subset of circuits, they all visit at least one node in $\overline{1}\cup\overline{2}$; therefore, their weight is non-positive if and only if $\graph(\pazocal{M}^{11})\in\nonegset$.
We can compute $\pazocal{M}^{11}$ using Lemma~\ref{le:aux1}:
\[
	\pazocal{M}^{11} = a^* (a^* bd^*ca^*)^*a^*.
\]
Therefore, $\graph(\pazocal{M}^{11})\in\nonegset$ if and only if
\begin{enumerate}
	\item $\graph(a)\in\nonegset$,
	\item $\graph(a^* bd^* ca^*)\in\nonegset$.
\end{enumerate}

Since $a = \begin{bsmallmatrix}
	C_{\wP_1} & \pazocal{E}\\
	\pazocal{E} & C_{\wP_2}
\end{bsmallmatrix}$, the first condition is equivalent to having $\graph(C_{\wP_1})\in\nonegset$ and $\graph(C_{\wP_2})\in\nonegset$ (more generally, one has to check the non-positiveness of the circuits in $q$ precedence graphs, each of which contains $n$ nodes).

To check the second condition, we must compute $a^* bd^* ca^*$.
The main challenge here is to find a closed formula for
\[
	bd^*c = 
	\begin{bsmallmatrix}
		I_{\winit} & P_{\wP_1} & \pazocal{E} & \pazocal{E} & \lambda_1^{-1}I_{\wP_4} & \pazocal{E} & \pazocal{E} & \pazocal{E} & \pazocal{E} & \pazocal{E} & \pazocal{E} & \pazocal{E} & \pazocal{E}\\
		\pazocal{E} & \pazocal{E} & \pazocal{E} & \pazocal{E} & \pazocal{E} & \pazocal{E} & \pazocal{E} & \pazocal{E} & I_{\wP_4} & P_{\wP_2} & \pazocal{E} & \pazocal{E} & \lambda_2^{-1}I_{\wP_3}
	\end{bsmallmatrix}\otimes
\]
\[
	\otimes
	\begin{bsmallmatrix}
		C_\winit & \pazocal{E} & \pazocal{E} & \pazocal{E} & \pazocal{E} & \pazocal{E} & \pazocal{E} & \pazocal{E} & \pazocal{E} & \pazocal{E} & \pazocal{E} & \pazocal{E} & \pazocal{E} \\
		\pazocal{E} & C_{\wP_1} & P_{\wP_1} & \pazocal{E} & \pazocal{E} & \pazocal{E} & \pazocal{E} & \pazocal{E} & \pazocal{E} & \pazocal{E} & \pazocal{E} & \pazocal{E}& \pazocal{E}  \\
		\pazocal{E} & I_{\wP_1} & C_{\wP_3} & P_{\wP_3} & \pazocal{E} & \pazocal{E} & \pazocal{E} & \pazocal{E} & \pazocal{E} & \pazocal{E} & \pazocal{E} & \pazocal{E} & \pazocal{E} \\
		\pazocal{E} & \pazocal{E} & I_{\wP_3} & C_{\wP_2} & P_{\wP_2} & \pazocal{E} & \pazocal{E} & \pazocal{E} & \pazocal{E} & \pazocal{E} & \pazocal{E} & \pazocal{E} & \pazocal{E} \\
		\pazocal{E} & \pazocal{E} & \pazocal{E} & I_{\wP_2} & C_{\wP_4} & P_{\wP_4} & \pazocal{E} & \pazocal{E} & \pazocal{E} & \pazocal{E} & \pazocal{E} & \pazocal{E} & \pazocal{E} \\
		\pazocal{E} & \pazocal{E} & \pazocal{E} & \pazocal{E} & I_{\wP_4} & C_{\wP_1} & P_{\wP_1} & \pazocal{E} & \pazocal{E} & \pazocal{E} & \pazocal{E} & \pazocal{E} & \pazocal{E} \\
		\pazocal{E} & \pazocal{E} & \pazocal{E} & \pazocal{E} & \pazocal{E} & I_{\wP_1} & C_{\wP_3} & P_{\wP_3} & \pazocal{E} & \pazocal{E} & \pazocal{E} & \pazocal{E} & \pazocal{E} \\
		\pazocal{E} & \pazocal{E} & \pazocal{E} & \pazocal{E} & \pazocal{E} & \pazocal{E} & I_{\wP_3} & C_{\wP_2} & P_{\wP_2} & \pazocal{E} & \pazocal{E} & \pazocal{E} & \pazocal{E} \\
		\pazocal{E} & \pazocal{E} & \pazocal{E} & \pazocal{E} & \pazocal{E} & \pazocal{E} & \pazocal{E} & I_{\wP_2} & C_{\wP_4} & \pazocal{E} & \pazocal{E} & \pazocal{E} & \pazocal{E} \\
		\pazocal{E} & \pazocal{E} & \pazocal{E} & \pazocal{E} & \pazocal{E} & \pazocal{E} & \pazocal{E} & \pazocal{E} & \pazocal{E} & C_{\wP_4} & P_{\wP_4} & \pazocal{E} & \pazocal{E} \\
		\pazocal{E} & \pazocal{E} & \pazocal{E} & \pazocal{E} & \pazocal{E} & \pazocal{E} & \pazocal{E} & \pazocal{E} & \pazocal{E} & I_{\wP_4} & C_{\wP_1} & P_{\wP_1} & \pazocal{E} \\
		\pazocal{E} & \pazocal{E} & \pazocal{E} & \pazocal{E} & \pazocal{E} & \pazocal{E} & \pazocal{E} & \pazocal{E} & \pazocal{E} & \pazocal{E} & I_{\wP_1} & C_{\wP_3}& P_{\wP_3} \\
		\pazocal{E} & \pazocal{E} & \pazocal{E} & \pazocal{E} & \pazocal{E} & \pazocal{E} & \pazocal{E} & \pazocal{E} & \pazocal{E} & \pazocal{E} & \pazocal{E} & I_{\wP_3}& C_{\wP_3}  
	\end{bsmallmatrix}^*
	\begin{bsmallmatrix}
		P_{\winit} & \pazocal{E} \\
		I_{\wP_1} & \pazocal{E} \\
		\pazocal{E} & \pazocal{E} \\
		\pazocal{E} & \pazocal{E} \\
		\lambda_1 P_{\wP_4} & \pazocal{E} \\
		\pazocal{E} & \pazocal{E} \\
		\pazocal{E} & \pazocal{E} \\
		\pazocal{E} & \pazocal{E} \\
		\pazocal{E} & P_{\wP_4} \\
		\pazocal{E} & I_{\wP_2} \\
		\pazocal{E} & \pazocal{E} \\
		\pazocal{E} & \pazocal{E} \\
		\pazocal{E} & \lambda_2 P_{\wP_3}
	\end{bsmallmatrix},
\]
which can be obtained by using Remarks~\ref{re:auxii},~\ref{re:aux2i}, and~\ref{re:auxi2}; the computation takes linear time complexity in the length of subschedules $u_0,\ldots,u_q$, $v_0,\ldots,v_q$.
The resulting matrix $a^* bd^* ca^*$, of the form $\lambda_1 P_1 \oplus \lambda_1^{-1} I_1 \oplus \lambda_2 P_2 \oplus \lambda_2^{-1} I_2 \oplus C$ (in general $\bigoplus_{i=1}^q \lambda_i P_i \oplus \lambda_i^{-1} I_i \oplus C$), has a number of rows and columns equal to $10$ (in general $qn$).
In conclusion, we reduced an MPIC-NCP on a graph with $75$ nodes (in general, $(U_0+\sum_{i=1}^q V_i + U_i)n$ nodes) to another MPIC-NCP on a graph with only $10$ nodes (in general, $pn$ nodes).

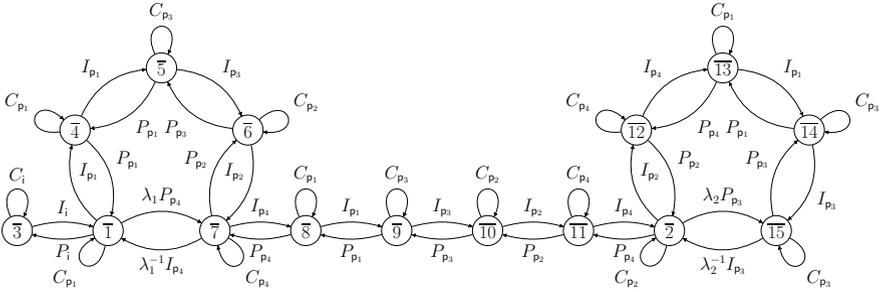
\begin{figure}
	\centering
	\resizebox{1\linewidth}{!}{
%
%

\begin{tikzpicture}[node distance=3cm,>=stealth',bend angle=45,double distance=.5mm,arc/.style={->,>=stealth'},place/.style={circle,thick,minimum size=10mm,draw},on grid]
\def \radius {3cm}
\Large

\node [place,label=center:{$\overline{1}$}] at ({-360/5*(2-1)-90+360/10}:\radius) (t2) {};
\draw [arc] (t2) to [out=-360/5*2+360/5-90+360/10+25,in=-360/5*2+360/5-90+360/10-25,loop] node[auto] {\textcolor{black}{$C_{\wP_1}$}} (t2);

\node [place,label=center:{$\overline{4}$}] at ({-360/5*(3-1)-90+360/10}:\radius) (t3) {};
\draw [arc] (t3) to [out=-360/5*3+360/5-90+360/10+25,in=-360/5*3+360/5-90+360/10-25,loop] node[auto] {\textcolor{black}{$C_{\wP_1}$}} (t3);

\node [place,label=center:{$\overline{5}$}] at ({-360/5*(4-1)-90+360/10}:\radius) (t4) {};
\draw [arc] (t4) to [out=-360/5*4+360/5-90+360/10+25,in=-360/5*4+360/5-90+360/10-25,loop] node[auto] {\textcolor{black}{$C_{\wP_3}$}} (t4);

\node [place,label=center:{$\overline{6}$}] at ({-360/5*(5-1)-90+360/10}:\radius) (t5) {};
\draw [arc] (t5) to [out=-360/5*5+360/5-90+360/10+25,in=-360/5*5+360/5-90+360/10-25,loop] node[auto] {\textcolor{black}{$C_{\wP_2}$}} (t5);

\node [place,label=center:{$\overline{7}$}] at ({-360/5*(6-1)-90+360/10}:\radius) (t6) {};
\draw [arc] (t6) to [out=-360/5*6+360/5-90+360/10+25,in=-360/5*6+360/5-90+360/10-25,loop] node[auto] {\textcolor{black}{$C_{\wP_4}$}} (t6);

\draw [arc] (t2) to [bend left=30] node[above right] {\textcolor{black}{$I_{\wP_1}$}} (t3);
\draw [arc] (t3) to [bend left=30] node[above right] {\textcolor{black}{$P_{\wP_1}$}} (t2);
\draw [arc] (t3) to [bend left=30] node[auto] {\textcolor{black}{$I_{\wP_1}$}} (t4);
\draw [arc] (t4) to [bend left=30] node[auto] {\textcolor{black}{$P_{\wP_1}$}} (t3);
\draw [arc] (t4) to [bend left=30] node[auto] {\textcolor{black}{$I_{\wP_3}$}} (t5);
\draw [arc] (t5) to [bend left=30] node[auto] {\textcolor{black}{$P_{\wP_3}$}} (t4);
\draw [arc] (t5) to [bend left=30] node[above left] {\textcolor{black}{$I_{\wP_2}$}} (t6);
\draw [arc] (t6) to [bend left=30] node[auto] {\textcolor{black}{$P_{\wP_2}$}} (t5);
\draw [arc] (t6) to [bend left=30] node[auto] {$\textcolor{black}{\lambda_1^{-1} I_{\wP_4}}$} (t2);
\draw [arc] (t2) to [bend left=30] node[auto] {$\textcolor{black}{\lambda_1 P_{\wP_4}}$} (t6);

\node [place,label=center:{$\overline{3}$},left=of t2] (t1) {}; 
\draw [arc] (t1) to [bend left=15] node[auto] {$I_\winit$} (t2);
\draw [arc] (t2) to [bend left=15] node[auto] {$P_\winit$} (t1);
\draw [arc] (t1) to [out=90+25,in=90-25,loop] node[auto] {$C_{\winit}$} (t1);

\node [place,label=center:{$\overline{8}$},right=of t6] (t7) {}; 
\draw [arc] (t6) to [bend left=15] node[auto] {$I_{\wP_4}$} (t7);
\draw [arc] (t7) to [bend left=15] node[auto] {$P_{\wP_4}$} (t6);
\draw [arc] (t7) to [out=90+25,in=90-25,loop] node[auto] {$C_{\wP_1}$} (t7);

\node [place,label=center:{$\overline{9}$},right=of t7] (t8) {}; 
\draw [arc] (t7) to [bend left=15] node[auto] {$I_{\wP_1}$} (t8);
\draw [arc] (t8) to [bend left=15] node[auto] {$P_{\wP_1}$} (t7);
\draw [arc] (t8) to [out=90+25,in=90-25,loop] node[auto] {$C_{\wP_3}$} (t8);

\node [place,label=center:{$\overline{10}$},right=of t8] (t9) {}; 
\draw [arc] (t8) to [bend left=15] node[auto] {$I_{\wP_3}$} (t9);
\draw [arc] (t9) to [bend left=15] node[auto] {$P_{\wP_3}$} (t8);
\draw [arc] (t9) to [out=90+25,in=90-25,loop] node[auto] {$C_{\wP_2}$} (t9);

\node [place,label=center:{$\overline{11}$},right=of t9] (t10) {}; 
\draw [arc] (t9) to [bend left=15] node[auto] {$I_{\wP_2}$} (t10);
\draw [arc] (t10) to [bend left=15] node[auto] {$P_{\wP_2}$} (t9);
\draw [arc] (t10) to [out=90+25,in=90-25,loop] node[auto] {$C_{\wP_4}$} (t10);

\node [place,label=center:{$\overline{2}$},right=of t10] (t11) {}; 
\draw [arc] (t11) to [out=-360/5*2+360/5-90+360/10+25,in=-360/5*2+360/5-90+360/10-25,loop] node[auto] {\textcolor{black}{$C_{\wP_2}$}} (t11);
\draw [arc] (t10) to [bend left=15] node[auto] {$I_{\wP_4}$} (t11);
\draw [arc] (t11) to [bend left=15] node[auto] {$P_{\wP_4}$} (t10);

\node [place,label=center:{$\overline{12}$}] at ($({-360/5*(3-1)-90+360/10}:\radius)+(t11)+(0.5878*\radius,0.8090*\radius)$) (t12) {};
\draw [arc] (t12) to [out=-360/5*3+360/5-90+360/10+25,in=-360/5*3+360/5-90+360/10-25,loop] node[auto] {\textcolor{black}{$C_{\wP_4}$}} (t12);
\draw [arc] (t11) to [bend left=30] node[above right] {$I_{\wP_2}$} (t12);
\draw [arc] (t12) to [bend left=30] node[auto] {$P_{\wP_2}$} (t11);

\node [place,label=center:{$\overline{13}$}] at ($({-360/5*(4-1)-90+360/10}:\radius)+(t11)+(0.5878*\radius,0.8090*\radius)$) (t13) {};
\draw [arc] (t13) to [out=-360/5*4+360/5-90+360/10+25,in=-360/5*4+360/5-90+360/10-25,loop] node[auto] {\textcolor{black}{$C_{\wP_1}$}} (t13);
\draw [arc] (t12) to [bend left=30] node[auto] {$I_{\wP_4}$} (t13);
\draw [arc] (t13) to [bend left=30] node[auto] {$P_{\wP_4}$} (t12);

\node [place,label=center:{$\overline{14}$}] at ($({-360/5*(5-1)-90+360/10}:\radius)+(t11)+(0.5878*\radius,0.8090*\radius)$) (t14) {};
\draw [arc] (t14) to [out=-360/5*5+360/5-90+360/10+25,in=-360/5*5+360/5-90+360/10-25,loop] node[auto] {\textcolor{black}{$C_{\wP_3}$}} (t14);
\draw [arc] (t13) to [bend left=30] node[auto] {$I_{\wP_1}$} (t14);
\draw [arc] (t14) to [bend left=30] node[auto] {$P_{\wP_1}$} (t13);

\node [place,label=center:{$\overline{15}$}] at ($({-360/5*(6-1)-90+360/10}:\radius)+(t11)+(0.5878*\radius,0.8090*\radius)$) (t15) {};
\draw [arc] (t15) to [out=-360/5*6+360/5-90+360/10+25,in=-360/5*6+360/5-90+360/10-25,loop] node[auto] {\textcolor{black}{$C_{\wP_3}$}} (t15);
\draw [arc] (t14) to [bend left=30] node[auto] {$I_{\wP_3}$} (t15);
\draw [arc] (t15) to [bend left=30] node[auto] {$P_{\wP_3}$} (t14);
%
%

\draw [arc] (t15) to [bend left=30] node[auto] {$\lambda_2^{-1} I_{\wP_3}$} (t11);
\draw [arc] (t11) to [bend left=30] node[auto] {$\lambda_2 P_{\wP_3}$} (t15);

\end{tikzpicture}

	}
	\caption{Lumped-node representation of $\graph(\overline{A}(\lambda_1,\lambda_2))$ from~\eqref{eq:intermittent_periodic_matrix_relabeled}.
	}
	\label{fi:intermittent_graph_relabeled}
\end{figure}

\section*{Acknowledgment}

We would like to thank Angelos Koumpis for implementing in C MEX a cache-aware max-plus matrix multiplication, which helped achieving the computational time shown in Figure~\ref{fig:comparison}.

This work was funded by the Deutsche Forschungsgemeinschaft (DFG, German Research Foundation), Projektnummer RA 516/14-1.
Partially supported by RVO 67985840, by the GACR grant 19-06175J, and by Deutsche Forschungsgemeinschaft (DFG, German Research Foundation) under Germany's Excellence Strategy -- EXC 2002/1 "Science of Intelligence" -- project number 390523135.

\section*{Conflict of interest}

The authors have no conflict of interest to declare that are relevant to this article.

\bibliography{references.bib}



\end{document}